\documentclass[a4paper, 11pt]{article}
\usepackage{amsmath}
\usepackage{amssymb}
\usepackage{amsthm}
\usepackage{bm}
\usepackage{graphicx}
\usepackage{caption}
\usepackage{subcaption}
\usepackage{enumitem}
\usepackage[colorlinks=true,linkcolor=blue,citecolor=blue]{hyperref}
\usepackage[dvipsnames]{xcolor}
\usepackage{epstopdf}
\usepackage[top=0.75in, bottom=1in, left=0.75in, right=0.75in]{geometry}
\usepackage{titling}
\usepackage{comment}
\usepackage{mathrsfs}
\usepackage{mathtools}
\usepackage{tikz}
\usepackage{footmisc}
\usepackage{algorithm}
\usepackage{algorithmicx}
\usepackage{algpseudocode}
\usepackage{authblk}
\usepackage{booktabs}
\usepackage{multirow}
\usepackage{xspace}
\graphicspath{}         

\usepackage[backend=bibtex8,giveninits=true,sorting=none,doi=false,isbn=false,url=false,maxbibnames=5]{biblatex}
\renewbibmacro{in:}{}
\theoremstyle{plain}
\newtheorem{theorem}{Theorem}
\newtheorem{corollary}[theorem]{Corollary}

\newtheorem{proposition}[theorem]{Proposition}

\theoremstyle{definition}
\newtheorem{definition}[theorem]{Definition}
 
\theoremstyle{remark}

\newcommand\norm[1]{\left\lVert#1\right\rVert}
\newcommand\ip[1]{\left\langle#1\right\rangle}

\newcommand\tr{\textnormal{Tr}}
\newcommand\id{\mathbb{I}}
\newcommand\lb{\textnormal{(}}
\newcommand\rb{\textnormal{)}}
\newcommand\strat{\mathfrak{S}}      
\newcommand\corr{\mathcal{C}}        
\newcommand\winvec{\mathfrak{W}}     
\newcommand\dout{d_{\textnormal{o}}} 
\DeclareMathOperator*{\argmax}{\textnormal{argmax}}
\DeclareMathOperator*{\argmin}{\textnormal{argmin}}
\DeclareMathOperator*{\conv}{\textnormal{conv}}
\DeclareMathOperator{\AM}{\textnormal{AM}}
\DeclareMathOperator{\HM}{\textnormal{HM}}
\newcommand{\BC}{{\mathcal B}}
\newcommand{\NC}{{\mathcal N}}
\newcommand{\ZC}{{\mathcal Z}}
\newcommand{\ngmac}{\textnormal{NG-MAC}\xspace}

\begin{document}
\title{On the separation of correlation-assisted\\ sum capacities of multiple access channels}
\author[1,3]{Akshay Seshadri}
\author[2]{Felix Leditzky}
\author[3]{Vikesh Siddhu \thanks{Present address: IBM Quantum}}
\author[1,3,4]{Graeme Smith \thanks{This paper was presented at the 2022 IEEE International Symposium on Information Theory (ISIT)~\cite{isit}.}}
\affil[1]{Department of Physics, University of Colorado, Boulder, CO 80309, USA}
\affil[2]{Department of Mathematics \& IQUIST, University of Illinois at Urbana-Champaign, Urbana, IL 61801, USA}
\affil[3]{JILA, University of Colorado/NIST, Boulder, CO 80309, USA}
\affil[4]{Center for Theory of Quantum Matter, University of Colorado, Boulder, CO 80309, USA}
\date{\today}
\maketitle

\begin{abstract}
    The capacity of a channel characterizes the maximum rate at which information can be transmitted through the channel asymptotically faithfully.
	For a channel with multiple senders and a single receiver, computing its sum capacity is possible in theory, but challenging in practice because of the nonconvex optimization involved.
    To address this challenge, we investigate three topics in our study.
    In the first part, we study the sum capacity of a family of multiple access channels (MACs) obtained from nonlocal games.
    For any MAC in this family, we obtain an upper bound on the sum rate that depends only on the properties of the game when allowing assistance from an arbitrary set of correlations between the senders.
	This approach can be used to prove separations between sum capacities when the senders are allowed to share different sets of correlations, such as classical, quantum or no-signalling correlations.
    We also construct a specific nonlocal game to show that the approach of bounding the sum capacity by relaxing the nonconvex optimization can give arbitrarily loose bounds.
    Owing to this result, in the second part, we study algorithms for non-convex optimization of a class of functions we call Lipschitz-like functions.
    This class includes entropic quantities, and hence these results may be of independent interest in information theory.
    Subsequently, in the third part, we show that one can use these techniques to compute the sum capacity of an arbitrary two-sender MACs to a fixed additive precision in quasi-polynomial time.
    We showcase our method by efficiently computing the sum capacity of a family of two-sender MACs for which one of the input alphabets has size two.
    Furthermore, we demonstrate with an example that our algorithm may compute the sum capacity to a higher precision than using the convex relaxation.
\end{abstract}

\section{Introduction}
Studying information transmission over a noisy channel is of fundamental importance in communication theory. In his landmark paper, Shannon studied the rate of transmission over a channel having a single sender and a single receiver \cite{shannon1948mathematical}. The maximum number of bits per use of channel that can be transmitted through the channel with asymptotically vanishing error is known as the capacity of the channel~\cite{thomas2006elements}. Shannon showed that the channel capacity can be calculated by maximizing the mutual information between the input and the output of the channel over all possible input probability distributions~\cite{shannon1948mathematical, thomas2006elements}. This is a convex optimization problem that can be solved using standard optimization techniques \cite{chiang2004geometric} or more specialized methods like the Arimoto-Blahut algorithm~\cite{arimoto1972algorithm, blahut1972computation}.

The importance of Shannon's work was soon recognized, and consequently, Shannon's ideas were generalized in various ways. In this study, we focus on one such generalization, namely a channel that has multiple senders and a single receiver, commonly known as a multiple access channel (MAC). For such a channel, a tuple of rates is called achievable if each sender can send information through their respective input at their respective rate such that the total error of transmission vanishes asymptotically. The set of all such achievable rate tuples is called the capacity region, and Ahlswede~\cite{ahlswede1974capacity} and Liao~\cite{liao1972multiple} were the first to give an entropic characterization of it. The total rate at which asymptotically error-free transmission through a MAC is possible is called the sum capacity of the MAC.

While there has been much research on MACs since the work of Ahlswede \& Liao, there are no efficient (polynomial time) algorithms to date that compute the sum capacity of a MAC. This difficulty stems from the fact that, unlike the computation of the capacity of a point-to-point channel, the optimization involved in computing the sum capacity is nonconvex~\cite{calvo2010computation, buhler2011note}. Proposals to solve this nonconvex problem efficiently were found to be unsuitable. On the contrary, it was recently shown that computing the sum capacity to a precision that scales inversely with the cube of the input alphabet size is an NP-hard problem~\cite{leditzky2020playing}. Therefore, one should not expect efficient general-purpose algorithms for computing the capacity region or the sum capacity of a MAC.

Since computing the sum capacity of a MAC is a hard nonconvex problem, a common approach that is adopted to circumvent this optimization is to relax the optimization to obtain a convex problem. The relaxation gives an upper bound on the sum capacity which can be efficiently computed. Such an approach was adopted, for example, by Calvo \textit{et al.}~\cite{calvo2010computation}.
Since we generally expect that there are no efficient algorithms to compute the sum capacity, quantifying the performance of such an upper bound becomes important and essential.

We undertake this task of elucidating the performance of the upper bound on sum capacity obtained through a convex relaxation. For this purpose, we consider MACs that are constructed from nonlocal games. Such MACs were introduced by Leditzky \textit{et al.}~\cite{leditzky2020playing} based on previous work of Quek \& Shor~\cite{quek2017quantum}, and subsequently generalized in \cite{notzel2020entanglement}. We present an analytical upper bound on the sum capacity of such MACs. Our bound extends to cases where the senders of these MACs can share arbitrary correlations, such as classical, quantum, and no-signalling correlations. Our upper bound depends only on the number of question tuples in the nonlocal game and the maximum winning probability of the game when the questions are drawn uniformly and answers are obtained using the strategies allowed by the shared correlations. Using these bounds, we obtain separations between the sum rate obtained from using different sets of correlations. These separations help distinguish the ability of these correlations to assist communication in MAC coding scenarios.
In particular, our bound gives a separation between the sum capacity and the entanglement-assisted sum rate for the MAC obtained from the magic square game. The separation found here is roughly $5$ times larger than the previously reported $1\%$ separation in~\cite{leditzky2020playing}.
Furthermore, using our bounds, we show how prior bounds on the sum capacity obtained using convex relaxation~\cite{calvo2010computation} can be arbitrarily loose.

Our result highlights the need to find better techniques to bound the sum capacity from above. We take a step in this direction by showing that computing the sum capacity is equivalent to optimizing a Lipschitz-like function. Thereupon, we present some algorithms for optimization of Lipschitz-like functions, which may be of independent interest. We show that these algorithms can compute the sum capacity of two-sender MACs to a given additive precision in quasi-polynomial time. Instead of a fixed precision, one of our algorithms can also accept a fixed number of iterations as an input and output an upper bound on the sum capacity. In particular, for a specific family of two-sender MACs that includes binary MACs, the number of iterations required for convergence grows at most polynomially with the dimensions and inverse precision. Thus, while it might not be possible to efficiently compute the sum capacity for an arbitrary MAC in practice, we can nevertheless efficiently compute the sum capacity for a large family of MACs.

\subsection{Organization of the paper}
We organize the paper in three parts, each one presenting different contributions of this work that may be of independent interest.
Focus of the first part (Sec.~\ref{secn:NG_MAC_relaxed_sum_capacity}) is MACs obtained from non-local games. We investigate the advantages of sharing correlations between various senders of a nonlocal games MAC. In Sec.~\ref{secn:NG_MAC_sum_capacity_upper_bound}, we obtain an upper bound on the sum capacity of MACs constructed from nonlocal games, and show separations between different sets of correlations. In Sec.~\ref{secn:relaxedALcapacity_sumcapacity_separation} we accomplish another primary goal of our study, showing that the convex relaxation of the sum capacity can be arbitrarily loose.

Motivated by this result, in the second part (Sec.~\ref{secn:Lipschitzlike_optimization}) we develop methods to solve non-convex problems like sum capacity computation. To this end, we define and study optimization of Lipschitz-like functions in Sec.~\ref{secn:lipschitz-like-functions}.
We then give an overview of the optimization algorithms for Lipschitz-like functions featured in this work in Sec.~\ref{secn:overview-optimization-algorithms}, and discuss them in greater detail in the following sections.
That is, in Sec.~\ref{secn:Lipschitzlike_optimization_1D} and Sec.~\ref{secn:Lipschitzlike_optimization_simplex} we develop algorithms for global optimization of such functions over a closed interval and over the standard simplex respectively. Results on optimization of Lipschitz-like functions over arbitrary compact and convex domains can be found in App.~\ref{app:Lipschitzlike_optimization_compact_convex_domain}. We also present relevant convergence guarantees and complexity analysis for our algorithms. In this way, we generalize certain prior algorithms for optimizing Lipschitz continuous functions. Our generalized algorithm may be of independent interest for non-convex optimization in information theory. For instance, we show that some entropic quantities that may be derived from Shannon and von Neumann entropies are Lipschitz-like functions.

The third part (Sec.~\ref{secn:sum_capacity_computation_algorithm}) applies algorithms developed in the second part for computing the sum capacity of an arbitrary two-sender MAC. In Sec.~\ref{secn:sum_capacity_computation_Lipschitz_like}, we prove that the sum capacity computation can be viewed as an optimization of a Lipschitz-like function. In Sec.~\ref{secn:sum_capacity_computation_algorithm_inputsize2}, we develop an efficient algorithm to compute the sum capacity of a large family of two-sender MACs, where one of the input alphabets is of size $2$. Subsequently, in Sec.~\ref{secn:sum_capacity_computation_algorithm_2senderMAC}, we present algorithms that can compute or upper-bound the sum capacity of an arbitrary two-sender MAC, along with a detailed complexity analysis. Finally, in Sec.~\ref{secn:sum_capacity_relaxed_sum_capacity_comparison}, we construct examples which demonstrate that our algorithm performs provably better than convex relaxation in computing the sum capacity.

We have tried to minimize the overlap between these three parts of our study, so that any one of these parts may be read independently without loss of continuity. In the next section, we provide a brief overview of notations, quantum states and measurements, nonlocal games, and multiple access channels. We recommend readers mainly interested in parts I and III to read Sec.~\ref{secn:preliminaries_MAC} on multiple access channels as it introduces definitions used in these parts.

\section{Preliminaries\label{secn:preliminaries}}
We briefly review some concepts that are used later. A short summary of the notation used throughout is given below.

We denote the set of natural numbers as $\mathbb{N} = \{0, 1, \dotsc\}$ and the set of positive integers as $\mathbb{N}_+ = \{1, 2, \dotsc\}$. For any integer $N \geq 1$, let $[N] := \{1, \dotsc, N\}$. We denote the $(n - 1)$-dimensional standard simplex by $\Delta_n = \{x \in \mathbb{R}^n \mid x \geq 0,\ \sum_{i = 1}^n x_i = 1\}$. When $x \in \mathbb{R}^n$ is a vector, we interpret the inequality $x \geq 0$ component-wise, i.e., $x_i \geq 0$ for all $i \in [n]$. We denote the non-negative orthant in $n$-dimensional Euclidean space by $\mathbb{R}^n_+ = \{x \in \mathbb{R}^n \mid x \geq 0\}$. The Euclidean inner product between two vectors $x, y \in \mathbb{R}^n$ is denoted by $\ip{x, y}$. The space of $(m\times n)$-matrices with entries in $\mathbb{C}$ is denoted by $\mathbb{C}^{m\times n}$. The Kronecker product of two matrices $A$ and $B$ is denoted by $A \otimes B$.

For a random variable $X$ taking values in a finite alphabet $\mathcal{X}$, we denote by $p_X(x)$ the probability of $X$ taking the value $x\in \mathcal{X}$. The Shannon entropy of the random variable $X$ is denoted by $H(X) = -\sum_{x \in \mathcal{X}} p_X(x) \log(p_X(x))$. The unit of entropy is referred to as bits when the base of the logarithm in $H(X)$ is $2$, whereas the unit is referred to as nats when the base of the logarithm is $e$. The mutual information between two random variables $X$ and $Y$ is defined as $I(X; Y) = H(X) + H(Y) - H(X, Y)$. A single-input single-output channel is a triple $(\mathcal{X}, \mathcal{Z}, \mathcal{N})$ consisting of an input alphabet $\mathcal{X}$, an output alphabet $\mathcal{Z}$ and a probability transition matrix $\mathcal{N}(z | x)$ giving the probability of transmitting $z \in \mathcal{Z}$ given the input $x \in \mathcal{X}$. The capacity of the channel $\mathcal{N}$ is given by the single-letter formula $C(\mathcal{N}) = \max_{p_X} I(X; Z)$, where $X$ and $Z$ are the random variables describing the input and output of $\mathcal{N}$, respectively~\cite{thomas2006elements}.

\subsection{Quantum states and measurements}
A quantum state (or density matrix) $\rho \in \mathbb{C}^{d \times d}$ is a self-adjoint, positive semi-definite (PSD) matrix with unit trace~\cite{nielsen2002quantum}. A measurement of this quantum state can be described by a positive operator-valued measure (POVM), a collection of PSD matrices $\{F_1, \dotsc, F_M\}$ of size $d \times d$ satisfying $\sum_{m = 1}^M F_m = \mathbb{I}$, where $\mathbb{I}$ is the identity matrix on the Hilbert space $\mathbb{C}^d$. Each POVM element $F_m$ is associated to a measurement outcome $m \in [M]$, which is obtained with probability $\text{Prob}(m) = \tr(\rho F_m)$~\cite{nielsen2002quantum}.

In quantum mechanics, a Hermitian operator $\mathcal{O} \in \mathbb{C}^{d \times d}$ is often called an \textit{observable}. It can be measured in the following sense. First, we write the spectral decomposition of $\mathcal{O}$ as $\mathcal{O} = \sum_{m = 1}^M \lambda_m \mathbb{P}_m$, where $\mathbb{P}_m \in \mathbb{C}^{d \times d}$ is the projector onto the eigenspace of $\mathcal{O}$ corresponding to the eigenvalue $\lambda_m \in \mathbb{R}$ for $m = 1, \dotsc, M$. Then, the POVM associated with measuring the observable $\mathcal{O}$ is given by $\{\mathbb{P}_m \mid m \in [M]\}$~\cite{nielsen2002quantum}.
For qubits, i.e., two-dimensional quantum systems described by quantum states $\rho\in\mathbb{C}^{2\times 2}$, the Pauli matrices
\begin{equation}
    \sigma_x = \begin{pmatrix} 0 & 1 \\
                               1 & 0 \end{pmatrix},
    \quad
    \sigma_y = \begin{pmatrix} 0 & -i \\
                               i & 0 \end{pmatrix},
    \quad
    \text{and} \quad
    \sigma_z = \begin{pmatrix} 1 & 0 \\
                               0 & -1 \end{pmatrix}
    \label{eq:PauliMat}
\end{equation}
represent three commonly used observables associated to the spin of an electron along different axes.

Let two parties, Alice and Bob, share a quantum state $\rho$. If they perform local measurements with POVMs $\{A_1, \dotsc, A_I\}$ and $\{B_1, \dotsc, B_J\}$, respectively, then the probability that Alice observes the outcome $i$ and Bob observes the outcome $j$ is given by $\text{Prob}(i, j) = \tr[\rho (A_i \otimes B_j)]$. In other words, the overall POVM for the measurement is given by $\{A_i \otimes B_j \mid i \in [I], j \in [J]\}$.

\subsection{Nonlocal games}
A nonlocal game is played between two players who each receive a question from a referee that they need to answer.
The players are not allowed to communicate with each other during the game. Prior to starting the game, one fixes the set of questions and answers and a winning criterion, which are known to everyone. A referee then randomly draws questions and hands them out to the players. If the answers of the players satisfy the winning condition, they win the game.

Formally, an $N$-player promise-free nonlocal game $G$ is a tuple $G = (\mathcal{X}_1, \dotsc, \mathcal{X}_N, \mathcal{Y}_1, \dotsc, \mathcal{Y}_N, \mathcal{W})$, where $\mathcal{X}_i$ and $\mathcal{Y}_i$ are the question and answer set for the $i$th player, respectively, and the winning condition $\mathcal{W} \subseteq (\mathcal{X}_1 \times \dotsm \times \mathcal{X}_N) \times (\mathcal{Y}_1 \times \dotsm \times \mathcal{Y}_N)$ determines the tuples of questions and answers that win the game~\cite{brassard2005quantum}. Throughout this study, we restrict our attention to the case when $\mathcal{X}_i$ and $\mathcal{Y}_i$ are finite sets for all $i \in [N]$. Unless stated otherwise, we will always refer to promise-free nonlocal games as nonlocal games.\footnote{One may consider nonlocal games for which the possible question tuples are restricted to a subset $P\subseteq \mathcal{X}_1 \times\dots \times \mathcal{X}_N$ called a \emph{promise}. A promise-free nonlocal game as defined above is one with $P = \mathcal{X}_1 \times\dots \times \mathcal{X}_N$. Any nonlocal game $G = (\mathcal{X}_1, \dotsc, \mathcal{X}_N, \mathcal{Y}_1, \dotsc, \mathcal{Y}_N, \mathcal{W})$ with promise $P$ can be turned into a promise-free game by defining a new winning condition $\mathcal{W}' = \mathcal{W} \cup [(\mathcal{X}_1 \times\dots \times \mathcal{X}_N)\setminus P] \times (\mathcal{Y}_1 \times \dotsm \times \mathcal{Y}_N)$, i.e., the players automatically win on question tuples not contained in the promise.}

For convenience, we denote the question set as $\mathcal{X} = \mathcal{X}_1 \times \dotsm \times \mathcal{X}_N$ and the answer set as $\mathcal{Y} = \mathcal{Y}_1 \times \dotsm \times \mathcal{Y}_N$. Given any question tuple $\bm{x} = (x_1,\dots,x_N) \in \mathcal{X}$, the $i$th question is given by $x_i \in \mathcal{X}_i$, and a similar notation is used for answers.

A strategy for the game $G=(\mathcal{X},\mathcal{Y},\mathcal{W})$ is any conditional probability distribution $p_{\bm{Y} | \bm{X}}(\bm{y} | \bm{x})$ on the answer tuples $\bm{y} \in \mathcal{Y}$ given a question tuple $\bm{x} \in \mathcal{X}$. There are different types of strategies that one can consider, depending on whether the players are allowed to use shared correlations after the game starts. We list a few strategies that are central to our study.

\begin{enumerate}
    \item \textit{Classical strategy}: This is the typical setup of a nonlocal game. Given a question $x_i \in \mathcal{X}_i$, the $i$th player decides an answer as per the probability distribution $p_{Y_i | X_i}(y_i | x_i)$, $i \in [N]$. This gives rise to a classical strategy
        \begin{equation}
            p_{\bm{Y} | \bm{X}}(y_1, \dotsc, y_N | x_1, \dotsc, x_N) = p_{Y_1 | X_1}(y_1 | x_1) \dotsm p_{Y_N | X_N}(y_N | x_N) \label{eqn:classical_strategy}
        \end{equation}
        In the general setting the players are allowed to choose such a strategy on the basis of the outcome of a random variable shared by all players, which is called a probabilistic strategy. Therefore, the set of classical strategies $\strat_{\text{cl}}$ corresponds to the convex hull of conditional probability distributions of the form given in Eq.~\eqref{eqn:classical_strategy}~\cite{palazuelos2016survey}. Operationally, convex combinations of product distributions correspond to shared randomness, and hence the set of classical correlations we use is similar to those defined by local hidden variable theories~\cite{palazuelos2016survey, mermin1993hidden}. However, when the questions are drawn uniformly at random, the classical strategy maximizing the probability of winning the game is of the form given in Eq.~\eqref{eqn:classical_strategy}, with $p_{Y_i | X_i}(y_i | x_i) = \delta_{y_i, f_i(x_i)}$ for $i \in [N]$ \cite{brassard2005quantum}. Here, $f_i\colon \mathcal{X}_i \to \mathcal{Y}_i$ is a function that outputs an answer given a question. Such a strategy is called a deterministic strategy.

    \item \textit{Quantum strategy}: The players have access to a shared quantum state $\rho$, but cannot communicate otherwise. Given a question $x_i \in \mathcal{X}_i$, the $i$th player performs a local measurement with some POVM $\{E^{(x_i)}_y \mid y \in \mathcal{Y}_i\}$. Subsequently, a quantum strategy is described by the probability distribution
        \begin{equation}
            p_{\bm{Y} | \bm{X}}(y_1, \dotsc, y_N | x_1, \dotsc, x_N) = \tr\left[\rho \left(E^{(x_1)}_{y_1} \otimes \dotsm \otimes E^{(x_N)}_{y_N}\right)\right] \label{eqn:quantum_strategy}
        \end{equation}
        We denote the set of quantum strategies by $\strat_{\text{Q}}$. For any given classical strategy, one can construct a quantum state and POVMs so that the quantum strategy reduces to the given classical strategy, and therefore, $\strat_{\text{cl}} \subseteq \strat_{\text{Q}}$~\cite{palazuelos2016survey}.

    \item \textit{No-signalling strategy}: A strategy $p_{\bm{Y} | \bm{X}}(y_1, \dotsc, y_N | x_1, \dotsc, x_N)$ is said to be no-signalling if
        \begin{equation}
            p_{Y_i | \bm{X}}(y_i | x_1, \dotsc, x_N) = p_{Y_i | X_i}(y_i | x_i), \quad i \in [N] \label{eqn:no_signalling_strategy}
        \end{equation}
        for all $x_1, \dotsc, x_N$~\cite{russo2017extended}. Informally, this means that players must respect locality, i.e., no information can be transmitted between players ``faster than light". As a consequence, the strategy used by each player cannot depend on the questions received by the other players. We denote the set of no-signalling strategies by $\strat_{\text{NS}}$. Both classical and quantum strategies are no-signalling, and therefore, $\strat_{\text{cl}} \subseteq \strat_{\text{Q}}\subseteq \strat_{\text{NS}}$. 

    \item\textit{Full communication}: When we impose no restriction on the distribution $p_{\bm{Y} | \bm{X}}(\bm{y} | \bm{x})$, we are implicitly allowing for communication between the players \textit{after} they have received the questions. The set of all possible conditional probability distributions $p_{\bm{Y} | \bm{X}}(\bm{y} | \bm{x})$ is denoted by $\strat_{\text{all}}$, which corresponds to allowing full communication between the players. This contains all classical, quantum and no-signalling strategies. Communication between the players is usually not allowed in nonlocal games, but we introduce this setting here for later use nevertheless.
\end{enumerate}

In summary, we have the following hierarchy of correlations: $\strat_{\text{cl}} \subseteq \strat_{\text{Q}} \subseteq \strat_{\text{NS}} \subseteq \strat_{\text{all}}$.

Suppose that the questions are drawn randomly from the set $\mathcal{X}$ as per the probability distribution $\pi(\bm{x})$. Say the players obtain their answers using strategies in some set $\strat$. Then the maximum winning probability of the game $G$ is given by
\begin{equation}
    \omega_\pi^\strat(G) = \sup_{p_{\bm{Y} | \bm{X}} \in \strat} \sum_{(\bm{x}, \bm{y}) \in \mathcal{W}} \pi(\bm{x}) p_{\bm{Y} | \bm{X}}(\bm{y} | \bm{x}). \label{eqn:max_winning_prob_game}
\end{equation}
Notice in particular that, if $\strat \subseteq \strat'$, then $\omega_\pi^\strat(G) \leq \omega_\pi^{\strat'}(G)$. Hence, as we go from classical to quantum to no-signalling strategies, the winning probability never decreases.

We will mainly be concerned with the scenario where the questions are drawn uniformly at random, and so $\pi = \pi_U$ is the uniform distribution on $\mathcal{X}$ with $\pi_U(\bm{x}) = 1/|\mathcal{X}|$ for all $\bm{x}\in\mathcal{X}$. In this case, we use the notation $\omega_{\pi_U}^\strat(G) = \omega^\strat(G)$. The maximum winning probability when the questions are drawn uniformly and answers are generated using classical strategies is written as $\omega^{\text{cl}}(G)$.
Similarly, when using quantum strategies, we write $\omega^{\text{Q}}(G)$, and when using no-signalling strategies, we write $\omega^{\text{NS}}(G)$.

\subsection{Multiple Access Channels}
\label{secn:preliminaries_MAC}
A discrete memoryless multiple access channel without feedback, simply referred
to as multiple access channel in this study, is a tuple $(\mathcal{B}_1,
\dotsc, \mathcal{B}_N, \mathcal{Z}, \mathcal{N}(z | b_1, \dotsc, b_N))$
consisting of input alphabets $\mathcal{B}_1, \dotsc, \mathcal{B}_N$, an output
alphabet $\mathcal{Z}$, and a probability transition matrix $\mathcal{N}(z |
b_1, \dotsc, b_N)$ that gives the probability 
that the channel output is $z \in \ZC$ when the input is 
$b_1 \in \BC_1, \dotsc, b_N \in \BC_N$~\cite{thomas2006elements}.
For simplicity of notation, we will
denote a MAC by the probability transition matrix $\mathcal{N}$ when the input
and output alphabets are understood.

Since transmission over the channel is error-prone, one usually encodes the messages and transmits them over multiple uses of the channel, and subsequently, the transmitted symbols are decoded to reconstruct the original message. Formally, an $((M_1, \dotsc, M_N), n)$-code for a multiple access channel consists of message sets $\mathcal{M}_i = \{1, \dotsc, M_i\}$ for $i \in [N]$, encoding functions $g^{(e)}_i\colon \mathcal{M}_i \to \mathcal{B}_i^{\times n}$ for $i \in [N]$, and a decoding function $g^{(d)}\colon \mathcal{Z}^n \to \mathcal{M}_1 \times \dotsm \times \mathcal{M}_N \cup \{\text{err}\}$, where $\text{err}$ is an error symbol~\cite{thomas2006elements}. For convenience, we denote the message set as $\mathcal{M} = \mathcal{M}_1 \times \dotsm \times \mathcal{M}_N$. The performance of the code can be quantified by the average probability of error in reconstructing the message,
\begin{equation*}
    P^{(n)}_e = \frac{1}{|\mathcal{M}|} \sum_{m \in \mathcal{M}} \text{Prob}\{g^{(d)}(Z^n) \neq m \mid \text{message } m \text{ was sent}\}.
\end{equation*}
Here we assume that the messages are chosen uniformly at random and transmitted through the channel. Note that the above code uses the channel $n$ times. We say that a rate tuple $(R_1, \dotsc, R_N)$ is \textit{achievable} if there exists a sequence of codes $((\lceil 2^{n R_1} \rceil, \dotsc, \lceil 2^{n R_N} \rceil), n)$ such that $P^{(n)}_e \to 0$ as $n\to\infty$~\cite{thomas2006elements}.

The \textit{capacity region} of a multiple access channel is defined as the closure of the set of achievable rate tuples $(R_1, \dotsc, R_N)$.
For $N = 2$ one obtains a two-sender MAC.
Ahlswede~\cite{ahlswede1974capacity} and Liao~\cite{liao1972multiple} were the
first to give a so-called single-letter characterization of the capacity region
of a two-sender MAC. The capacity region for an $N$-sender MAC $\mathcal{N}$
can be written as
\begin{multline}
    \text{Cap}(\mathcal{N}) = \conv \Bigg\{(R_1, \dotsc, R_N) \mid 0 \leq \sum_{j \in J} R_i \leq I(B_J; Z | B_{J^c})\ \forall \varnothing \neq J \subseteq [N],\\
    p(b_1, \dotsc, b_N) = p(b_1) \dotsm p(b_N)\Bigg\} \label{eqn:MAC_capacity_region}
\end{multline}
where $p(b_1, \dotsc, b_N)$ is a probability distribution
over the joint random variable $B_1, \dotsc, B_N$ describing the channel input,
$p(b_i)$ is a probability distribution for the random variable $B_i$
corresponding to the $i$-th sender's input, random variable $Z$ describes the
channel output, and for any set $J \subseteq [N]$, $B_J = \{B_j \mid j \in
J\}$~\cite{calvo2010computation}.
Crucially, the optimization in Eq.~\eqref{eqn:MAC_capacity_region} restricts the joint
input random variable to have a product distribution, i.e., $p(b_1, \dotsc,
b_N) = p(b_1) \dotsm p(b_N)$.
The main quantity of interest in our study is the \textit{sum capacity} of a MAC~\cite{calvo2010computation},
\begin{align}
    S(\mathcal{N}) &= \sup \left\{\sum_{i = 1}^N R_i \mid (R_1, \dotsc, R_N) \in \text{Cap}(\mathcal{N})\right\} \nonumber \\
                   &= \max_{p(b_1 \dotsm b_N)} I(B_1, \dotsc, B_N; Z) \quad \text{such that} \quad p(b_1, \dotsm, b_N) = p(b_1) \dotsm p(b_N).
                   \label{eqn:sum_capacity}
\end{align}

Informally, the sum capacity represents the maximum sum rate at which the senders can send information through the MAC such that the transmission error vanishes asymptotically. Because the maximization involved in computing the sum capacity is constrained to be over product distributions on the input, the resulting optimization problem is nonconvex. This nonconvexity is the main source of difficulty in computing the sum capacity of a MAC in practice~\cite{buhler2011note, leditzky2020playing}. A common approach to avoid this difficult is to relax the nonconvex constraint and maximize over all possible probability distributions on the input. Such an approach was adopted by Calvo \textit{et al.}~\cite{calvo2010computation}, leading to the upper bound
\begin{equation}
    C(\mathcal{N}) = \max_{p(b_1, \dotsc, b_N)} I(B_1, \dotsc, B_N; Z) \label{eqn:relaxed_AL_capacity}
\end{equation}
on the sum capacity $S(\mathcal{N})$, where we now maximize over \emph{arbitrary} input
probability distributions. Note that $C(\mathcal{N})$ is the capacity of the
channel $\mathcal{N}$ when we think about it as a single-input single-output
channel. 
Since it is a relaxation of the sum capacity corresponding to the
capacity region, we call $C$ the \textit{relaxed sum capacity}.

\section{Nonlocal games MAC and relaxed sum capacity\label{secn:NG_MAC_relaxed_sum_capacity}}
\subsection{Bounding the sum capacity of MACs from nonlocal games\label{secn:NG_MAC_sum_capacity_upper_bound}}
Suppose we are given a promise-free nonlocal game $G$ with question sets $\mathcal{X}_1, \dotsc, \mathcal{X}_N$, answer sets $\mathcal{Y}_1, \dotsc, \mathcal{Y}_N$, and a winning condition $\mathcal{W} \subseteq (\mathcal{X}_1 \times \dotsm \times \mathcal{X}_N) \times (\mathcal{Y}_1 \times \dotsm \times \mathcal{Y}_N)$. Following Leditzky \textit{et al.}~\cite{leditzky2020playing}, we construct a MAC $\mathcal{N}_G$ with input alphabets $\mathcal{B}_i = \mathcal{X}_i \times \mathcal{Y}_i$ for $i \in [N]$, output alphabet $\mathcal{Z} = \mathcal{X}_1 \times \dotsm \mathcal{X}_N$, and a probability transition matrix
\begin{align}
    \mathcal{N}_G(\widehat{x}_1, \dotsc, \widehat{x}_N | x_1, y_1, \dotsc, x_N, y_N) &= \begin{cases} \delta_{\hat{x}_1, x_1} \dotsm \delta_{\hat{x}_N, x_N} &(x_1, \dotsc, x_N, y_1, \dotsc, y_N) \in \mathcal{W} \\
        \frac{1}{d} &(x_1, \dotsc, x_N, y_1, \dotsc, y_N) \notin \mathcal{W},
    \end{cases} \label{eqn:MAC_nonlocal_game}
\end{align}
where $d = |\mathcal{Z}| = |\mathcal{X}_1 \times \dotsm \times \mathcal{X}_N|$.

In words, the channel $\mathcal{N}_G$ takes question-answer pairs from each player as input, and if they win the game, then the questions are transmitted without any noise. However, if they lose the game, a question tuple chosen uniformly at random is output by the channel. For convenience, we will denote the input to the MAC $\mathcal{N}_G$ as $\mathcal{XY} = (\mathcal{X}_1 \times \mathcal{Y}_1) \times \dotsm \times (\mathcal{X}_N \times \mathcal{Y}_N)$, and write $\bm{xy} = (x_1, y_1, \dotsc, x_N, y_N)$. We will usually abbreviate the phrase ``MAC obtained from a nonlocal game" to ``nonlocal game MAC" or \ngmac.

Before diving into any technical details, we explain why such MACs are suitable for obtaining bounds on the sum capacity that are better than the relaxed sum capacity. We begin by noting that the sum capacity of the \ngmac $\mathcal{N}_G$ can be written as
\begin{equation*}
    S(\mathcal{N}_G) = \max_{p^{(1)}(x_1, y_1) \dotsm p^{(N)}(x_N, y_N)} I(X_1, Y_1, \dotsc, X_N, Y_N; Z)
\end{equation*}
where $(X_i, Y_i)$ is the random variable (with distribution $p^{(i)}$) describing the $i$th input and $Z$ is the random variable describing the output of $\mathcal{N}_G$. Note that $p^{(i)}$ is a probability distribution over the question-answer pairs of the $i$th player. By writing $p^{(i)}(x_i, y_i) = \pi^{(i)}(x_i) p_{Y_i | X_i}(y_i | x_i)$, we can break $p^{(i)}$ into a distribution $\pi^{(i)}$ over questions and a strategy $p_{Y_i | X_i}$ chosen by the $i$th player. As a result, we can write $p^{(1)} \dotsm p^{(N)} = \pi p_{\bm{Y} | \bm{X}}$, where $\pi = \pi^{(1)} \dotsm \pi^{(N)}$ is some distribution over the questions and $p_{\bm{Y} | \bm{X}} = p_{Y_1 | X_1} \dotsm p_{Y_N | X_N}$ is a classical strategy chosen by the players. Such a decomposition allows us to optimize separately over questions and strategies. By performing suitable relaxations, we can obtain a bound on the sum capacity that depends only on the winning probability of the game (see Thm.~\eqref{thm:NG_MAC_correlation_assisted_sum_capacity_bound} for a precise statement). In fact, such a proof technique allows us to bound the sum capacity even when the players are allowed to use different sets of strategies such as those obtained from quantum or no-signalling correlations. The resulting bound is helpful in obtaining separations between the communication capabilities of different sets of correlations.

In contrast, the relaxed sum capacity $C(\mathcal{N}_G)$ is computed by maximizing over all possible probability distributions. For a nonlocal game MAC, this amounts to maximizing over all distributions over the questions and all possible strategies (allowing full communication between the players). Using such strategies, the players will always win the game, assuming that the game has at least one correct answer for every question. This results in the trivial bound $C(\mathcal{N}_G) = \ln(d)$ on the sum capacity, since the players can always win the game. On the other extreme, if the winning probability of the game is zero, then we have $S(\mathcal{N}_G) = 0$. Following this line of thinking, we construct a game in Section~\ref{secn:relaxedALcapacity_sumcapacity_separation} that allows us to give an arbitrarily large separation between $S(\mathcal{N}_G)$ and $C(\mathcal{N}_G)$.

The above discussion suggests that the sum capacity of the \ngmac $\mathcal{N}_G$ increases with the winning probability of the game $G$, an observation that was also noted by Leditzky \textit{et al.}~\cite{leditzky2020playing}. On the other hand, we know that the winning probability of the game can increase if we allow the players to use a larger set of strategies. This motivates us to allow the senders to share some set of correlations to play the game so as to increase the sum capacity of \ngmac.

\subsubsection{Correlation assistance\label{secn:NG_MAC_strategy_assistance}}
By a correlation, we mean a probability distribution $P(\bm{y}' | \bm{xy})$, where $\bm{xy} \in \mathcal{XY}$ is the input to the \ngmac, while $\bm{y}' \in \mathcal{Y}$ is an answer to the nonlocal game $G$.\footnote{The correlation $P$ may depend on the answers $\bm{y} \in \mathcal{Y}$ in addition to the questions $\bm{x} \in \mathcal{X}$ because we wish to construct a larger MAC that has the same structure as the \ngmac $\mathcal{N}_G$ and has assistance from the correlation $P$. See Fig.~\ref{fig:correlation_assisted_nonlocal_games_MAC} for example. In practice, the correlation will often produce the answer tuple $\bm{y}'$ solely from the question tuple $\bm{x}$, using some strategy for the non-local game $G$.}
We allow the $N$ senders to share the correlation $P$ and perform local post-processing of the answer generated by $P$ to obtain their final answer for the game for the input questions. This post-processing can be expressed as a probability distribution $f_i(\overline{y}_i | x_i, y_i, y_i')$ over the answers $\overline{y}_i \in \mathcal{Y}_i$ given the input question-answer pair $(x_i, y_i) \in \mathcal{X}_i \times \mathcal{Y}_i$ and the answer $y_i' \in \mathcal{Y}_i$ generated by the correlation $P$. For convenience, denote 
\begin{equation}
    f(\overline{\bm{y}} | \bm{xy}, \bm{y}') = \prod_{i = 1}^N f_i(\overline{y}_i | x_i, y_i, y_i'),  \label{eqn:local_postprocessing}
\end{equation}
which is a distribution over the answers $\overline{\bm{y}}\in \mathcal{Y}$ given input question-answer pairs $\bm{xy}\in \mathcal{XY}$ and answers $\bm{y}'\in\mathcal{Y}$ generated by the correlation $P$.

This gives rise to the channel $\mathcal{A}_{P, f}$ having input and output alphabets $\mathcal{XY}$ and the probability transition matrix
\begin{equation}
    \mathcal{A}_{P, f}(\overline{\bm{xy}} | \bm{xy}) = \delta_{\overline{\bm{x}}, \bm{x}} \sum_{\bm{y}' \in \mathcal{Y}} f(\overline{\bm{y}} | \bm{xy}, \bm{y}') P(\bm{y}' | \bm{xy}). \label{eqn:strategy_assistance_channel}
\end{equation}
This probability transition matrix gives a probability distribution over the question-answer pairs $\overline{\bm{xy}} \in \mathcal{XY}$ given some input question-answer pair $\bm{xy} \in \mathcal{XY}$. The probability distributions $f_i$ denote local post-processing by the $i$th sender to generate the final answer $\overline{y}_i$. Note that the definition of $\mathcal{A}_{P, f}$ in Eq.~\eqref{eqn:strategy_assistance_channel} ensures that the input questions $\bm{x} \in \mathcal{X}$ are transmitted without any change, i.e., $\overline{\bm{x}} = \bm{x}$. The post-processings are only used to generate the final answers. This definition is slightly more restrictive than the one given in Leditzky \textit{et al.}~\cite{leditzky2020playing}, where post-processing of both questions and answers is considered. We impose such a restriction on post-processing so as to ensure that the set of strategies induced by the channel $\mathcal{A}_{P, f}$ has a close relation to the set of correlations $\corr$ shared by the senders. For example, classical or quantum correlations shared by the senders give rise to classical or quantum strategies to play the nonlocal game $G$. We make this idea precise in the following discussion.

An input probability distribution $p(\bm{xy})$ over the question-answer pairs $\mathcal{XY}$ can be written as $p(\bm{xy}) = \pi(\bm{x}) p_{\bm{Y} | \bm{X}}(\bm{y} | \bm{x})$, a product of probability distribution over the questions $\pi(\bm{x})$ and a strategy $p_{\bm{Y} | \bm{X}}(\bm{y} | \bm{x})$. This input probability distribution is modified by the channel $\mathcal{A}_{P, f}$ to give
\begin{equation}
    \overline{p}(\overline{xy}) = \sum_{\bm{xy} \in \mathcal{XY}} \mathcal{A}_{P, f}(\overline{\bm{xy}} | \bm{xy}) p(\bm{xy})
    = \pi(\overline{\bm{x}}) \overline{p}_{\bm{Y} | \bm{X}}(\overline{\bm{y}} | \overline{\bm{x}}), \label{eqn:strategy_assistance_channel_induced_probability}
\end{equation}
where the new strategy $\overline{p}_{\bm{Y} | \bm{X}}$ is given by
\begin{equation}
    \mathcal{A}_{P, f}(p_{\bm{Y} | \bm{X}}) \coloneqq \overline{p}_{\bm{Y} | \bm{X}}(\overline{\bm{y}} | \overline{\bm{x}}) = \sum_{\bm{y} \in \mathcal{Y}} \sum_{\bm{y}' \in \mathcal{Y}} f(\overline{\bm{y}} | \bm{\overline{x}y}, \bm{y}')
    P(\bm{y}' | \bm{\overline{x}y}) p_{\bm{Y} | \bm{X}}(\bm{y} | \overline{\bm{x}}). \label{eqn:assistance_induced_strategy}
\end{equation}
We use the notation $\mathcal{A}_{P, f}(p_{\bm{Y} | \bm{X}})$ to emphasize that $\overline{p}_{\bm{Y} | \bm{X}}$ is the strategy induced by the action of $\mathcal{A}_{P, f}$ on $p_{\bm{Y} | \bm{X}}$. Therefore, the channel $\mathcal{A}_{P, f}$ modifies the input strategy $p_{\bm{Y} | \bm{X}}$ by incorporating assistance from the correlation $P$ and local post-processings $f$. For this reason, we call $\mathcal{A}_{P, f}$ the \textit{correlation-assistance channel}.

If the senders have access to some set of correlations $\corr$, we can define the set of strategies induced by these correlations as
\begin{equation}
    \strat_\corr = \left\{\mathcal{A}_{P, f}(p_{\bm{Y} | \bm{X}}) \mid p_{\bm{Y} | \bm{X}} = \prod_{i = 1}^N p_{Y_i | X_i},\ P \in \corr,\ f \in \text{PP}\right\}. \label{eqn:correlation_induced_strategies}
\end{equation}
where $\text{PP}$ is the set of all local post-processings of the answers generated by the correlation (as defined in Eq.~\eqref{eqn:local_postprocessing}).
We now show that there is a close relation between the correlations $\corr$ shared by the senders and the strategies $\strat_\corr$ induced by these correlations.

If $\corr_{\text{cl}}$ is the set of classical correlations, then any $P \in \corr_{\text{cl}}$ can be written as $P(\bm{y}' | \bm{xy}) = \prod_{i = 1}^N P_i(y_i' | x_i, y_i)$ and convex combinations thereof, where $P_i$ are some distributions over $\mathcal{Y}_i$ given a question-answer pair from $(\mathcal{X}_i, \mathcal{Y}_i)$ for $i \in [N]$. Then, from Eq.~\eqref{eqn:assistance_induced_strategy}, we find that the strategy $\mathcal{A}_{P, f}(p_{\bm{Y} | \bm{X}})$ induced by the correlation assistance channel is also a classical strategy (since the input $p_{\bm{Y} | \bm{X}}$ is a classical strategy). Therefore, the set of strategies induced by classical correlations (as defined in Eq.~\eqref{eqn:correlation_induced_strategies}) is the set of classical strategies $\strat_{\text{cl}}$.

On the other hand, if $\corr_{\text{Q}}$ is the set of quantum correlations, then any $P \in \corr_{\text{Q}}$ can be written as
\begin{equation*}
    P(\bm{y}' | \bm{xy}) = \tr\left[\rho \left( E^{(x_1, y_1)}_{y_1'} \otimes \dotsm \otimes E^{(x_N, y_N)}_{y_N'}\right)\right]
\end{equation*}
where $\rho$ is the quantum state shared between the senders and $\{E^{(x_i, y_i)}_{y'} \mid y' \in \mathcal{Y}_i\}$ is the local measurement implemented by the $i$th player upon receiving the input $(x_i, y_i)$ for each $i \in [N]$. Now, we define new local measurement operators
\begin{equation*}
    \overline{E}^{(\overline{x}_i)}_{\overline{y}_i} = \sum_{y_i, y_i' \in \mathcal{Y}_i} p_{Y_i | X_i}(y_i | \overline{x}_i) f_i(\overline{y}_i | \overline{x}_i, y_i y_i') E^{(\overline{x}_i, y_i)}_{y_i'},
\end{equation*}
so that the strategy induced by the correlation $P$ can be written using Eq.~\eqref{eqn:assistance_induced_strategy} as
\begin{equation*}
    \overline{p}_{\bm{Y} | \bm{X}}(\overline{\bm{x}} | \overline{\bm{y}}) = \tr\left[ \rho \left(\overline{E}^{(\overline{x}_1)}_{\overline{y}_1} \otimes \dotsm \otimes \overline{E}^{(\overline{x}_N)}_{\overline{y}_N}\right) \right].
\end{equation*}
This shows that the set of strategies induced by quantum correlations, as defined in Eq.~\eqref{eqn:correlation_induced_strategies}, is the set of quantum strategies $\strat_{\text{Q}}$. Similarly, one can verify that the set of strategies induced by no-signalling correlations $\corr_{\text{NS}}$ is the set of no-signalling strategies $\strat_{\text{NS}}$.

We now elaborate on how one can use the correlation-assistance channel to boost the sum capacity of the nonlocal games MAC. Given the \ngmac $\mathcal{N}_G$ obtained from a nonlocal game $G$, some correlation $P$ shared by the senders and local post-processings $f$, we define the correlation-assisted \ngmac $\mathcal{N}_G \circ \mathcal{A}_{P, f}$. That is, the input question-answer pair is first passed through the correlation-assistance channel $\mathcal{A}_{P, f}$, which tries to improve the strategy for playing the game, and the modified question-answer pair is passed on to the \ngmac $\mathcal{N}_G$. A schematic of this procedure for the case of two senders is shown in Fig.~\ref{fig:correlation_assisted_nonlocal_games_MAC}. If the local post-processing $f$ discards information about the input questions as well as the answers generated by the correlation $P$, i.e., $f(\overline{\bm{y}} | \bm{xy}, \bm{y}') = \delta_{\overline{\bm{y}}, \bm{y}}$, then $\mathcal{A}_{P, f}$ becomes the identity channel. Therefore, the correlation-assisted \ngmac $\mathcal{N}_G \circ \mathcal{A}_{P, f}$ is at least as powerful as the \ngmac $\mathcal{N}_G$ if we allow the senders to perform any local post-processing.

\begin{figure}[H]
    \vspace{0.5cm}
    \begin{center}
        \begin{tikzpicture}
            \draw (-1.85, 2.5) node {$(X_1, Y_1)$};
            \draw (-1.85, 0.5) node {$(X_2, Y_2)$};
            \draw (-1, 2.5) -- (-0.5, 2.5) -- (0, 3) -- (2.5, 3);
            \draw (-1, 2.5) -- (-0.5, 2.5) -- (0, 2) -- (0.5, 2);
            \draw (-1, 0.5) -- (-0.5, 0.5) -- (0, 1) -- (0.5, 1);
            \draw (-1, 0.5) -- (-0.5, 0.5) -- (0, 0) -- (2.5, 0);
            \draw (0.5, 0.5) rectangle (1.5, 2.5);
            \draw (1, 1.5) node {$P$};
            \draw (1.5, 2) -- (2.5, 2);
            \draw (1.5, 1) -- (2.5, 1);
            \draw (2, 2.3) node {$Y_1'$};
            \draw (2, 0.7) node {$Y_2'$};
            \draw (2.5, 1.75) rectangle (3.2, 3.25);
            \draw (2.5, 1.25) rectangle (3.2, -0.25);
            \draw (2.85, 2.5) node {$f_1$};
            \draw (2.85, 0.5) node {$f_2$};
            \draw (3.2, 2.5) -- (5.2, 2.5);
            \draw (3.2, 0.5) -- (5.2, 0.5);
            \draw (4.3, 2.85) node {$(X_1, \overline{Y}_1)$};
            \draw (4.3, 0.15) node {$(X_2, \overline{Y}_2)$};
            \draw[dashed, thick, blue] (-0.65, -0.5) rectangle (3.45, 3.5);
            \draw (1.5, -1) node {$\mathcal{A}_{P, f}$};
            \draw (5.2, -0.25) rectangle (6.8, 3.25);
            \draw (6, 1.5) node {$\mathcal{N}_G$};
            \draw (6.8, 1.5) -- (7.3, 1.5);
            \draw (8.1, 1.5) node {$(\widehat{X}_1, \widehat{X}_2)$};
        \end{tikzpicture}
    \end{center}
    \caption{Correlation-assisted \ngmac $\mathcal{N}_G \circ \mathcal{A}_{P, f}$ for the case of two senders, obtained from the nonlocal games MAC $\mathcal{N}_G$ defined in Eq.~\eqref{eqn:MAC_nonlocal_game} and correlation-assistance channel $\mathcal{A}_{P, f}$ defined in Eq.~\eqref{eqn:strategy_assistance_channel}.}
    \label{fig:correlation_assisted_nonlocal_games_MAC}
\end{figure}
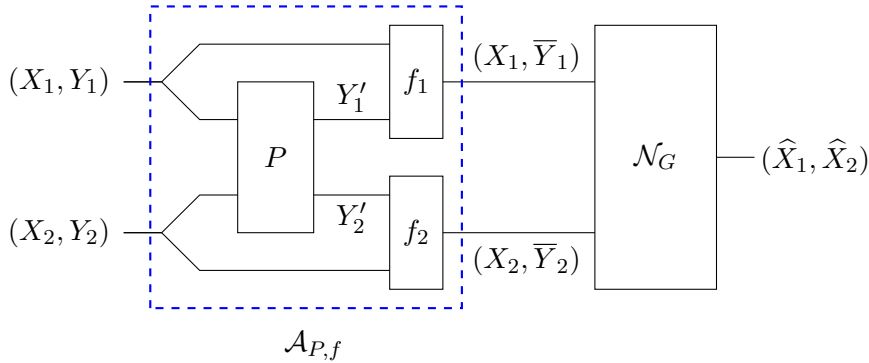

Suppose that the senders share a set of correlations $\corr$. The \textit{$\corr$-assisted achievable rate region} of the \ngmac $\mathcal{N}_G$ is defined as
\begin{equation}
    \text{Cap}_{\corr}^{(1)}(\mathcal{N}_G) = \bigcup_{\substack{P \in \corr,\\ f \in \text{PP}}} \text{Cap}(\mathcal{N}_G \circ \mathcal{A}_{P, f}) \label{eqn:C-assisted_achievable_rate_region}
\end{equation}
where $\text{Cap}(\mathcal{N}_G \circ \mathcal{A}_{P, f})$ is the capacity region defined in Eq.~\eqref{eqn:MAC_capacity_region} evaluated for the correlation-assisted \ngmac $\mathcal{N}_G \circ \mathcal{A}_{P, f}$.\footnote{The superscript $^{(1)}$ signifies that $\text{Cap}_{\corr}^{(1)}(\cdot)$ is merely a region consisting of \emph{achievable} rate pairs, and hence contained in the full capacity region $\text{Cap}_{\corr}(\cdot)$. To determine whether $\text{Cap}_{\corr}^{(1)}(\cdot) = \text{Cap}_{\corr}(\cdot)$ is outside the scope of this work.} 
The \textit{$\corr$-assisted achievable sum rate} of the \ngmac $\mathcal{N}_G$ is
\begin{equation}
    S_{\corr}(\mathcal{N}_G) = \sup \left\{\sum\nolimits_{i = 1}^N R_i \mid (R_1, \dotsc, R_N) \in \text{Cap}_{\corr}^{(1)}(\mathcal{N}_G)\right\}. \label{eqn:C-assisted_achievable_sum_rate}
\end{equation}

We now derive an alternate expression of $S_{\corr}(\mathcal{N}_G)$ that is more convenient for computation. Prior to obtaining this expression, note that for any given family of sets $\{F_i\}_{i \in \mathcal{I}}$ indexed by some set $\mathcal{I}$ and any function $g\colon \cup_{i \in \mathcal{I}} F_i \to \mathbb{R}$, we have
\begin{equation*}
    \sup_{r \in \cup_{i \in \mathcal{I}} F_i} g(r) = \sup_{i \in \mathcal{I}} \sup_{r \in F_i} g(r).
\end{equation*}
Using this equation along with Eq.~\eqref{eqn:sum_capacity}, we can write
\begin{equation}
    S_{\corr}(N_G) = \sup_{P \in \corr, f \in \text{PP}} S(\mathcal{N}_G \circ \mathcal{A}_{P, f})
    = \sup_{P \in \corr, f \in \text{PP}} \sup_{p^{(1)} \dotsm p^{(N)}} I(X_1, Y_1, \dotsc, X_N, Y_N; Z), \label{eqn:correlation_assisted_achievable_sum_rate}
\end{equation}
where $(X_i, Y_i)$ is the random variable (with distribution $p^{(i)}$) describing the input for $i \in [N]$, and $Z$ is the random variable describing the output of the \ngmac $\mathcal{N}_G \circ \mathcal{A}_{P, f}$.

Since $S_{\corr}(\mathcal{N}_G)$ corresponds to a maximization over all possible local post-processings, we must have $S(\mathcal{N}_G) \leq S_{\corr}(\mathcal{N}_G)$ for any set of correlations $\corr$. Furthermore, if $\corr_1 \subseteq \corr_2$, then $\text{Cap}_{\corr_1}^{(1)}(\mathcal{N}_G) \subseteq \text{Cap}_{\corr_2}^{(1)}(\mathcal{N}_G)$, and consequently also $S_{\corr_1}(\mathcal{N}_G) \leq S_{\corr_2}(\mathcal{N}_G)$. Finally, note that we compute the relaxed sum capacity $C(\mathcal{N}_G)$ by maximizing over all possible distributions over the questions and all possible strategies $\strat_{\text{all}}$ that the players can use (since the maximization in Eq.~\eqref{eqn:relaxed_AL_capacity} is over all input probability distributions). Because $\strat_\corr \subseteq \strat_{\text{all}}$ for any set of correlations $\corr$, we have $S_\corr(\mathcal{N}_G) \leq C(\mathcal{N}_G)$. Therefore, we obtain a hierarchy,
\begin{equation}
    S(\mathcal{N}_G) \leq S_{\text{cl}}(\mathcal{N}_G) \leq S_{\text{Q}}(\mathcal{N}_G) \leq S_{\text{NS}}(\mathcal{N}_G) \leq C(\mathcal{N}_G), \label{eqn:correlation_assisted_sum_capacity_hierarchy}
\end{equation}
where ``$\text{cl}$'', ``$\text{Q}$'', and ``$\text{NS}$'' denote classical, quantum, and no-signalling correlations, respectively. Note that the sum capacity $S(\mathcal{N}_G)$ might not be equal to $S_{\text{cl}}(\mathcal{N}_G)$ because classical correlations can be convex combinations of product distributions.

We now proceed to obtaining a bound on the $\corr$-assisted achievable sum rate.

\subsubsection{Bounding the correlation-assisted sum rate}
Let $\corr$ be any set of correlations shared between the senders. In order to bound $S_{\corr}(\mathcal{N}_G)$, we first obtain an optimization problem in terms of distributions over questions and strategies induced by the shared correlations. For a given correlation $P \in \corr$, the input and output of $\mathcal{N}_G \circ \mathcal{A}_{P, f}$ can be described as follows. The channel $\mathcal{A}_{P, f}$ takes the input random variable $(X_1, Y_1, \dotsc, X_N, Y_N)$ and outputs $(\overline{X}_1, \overline{Y}_1, \dotsc, \overline{X}_N, \overline{Y}_N)$. The output of $\mathcal{A}_{P, f}$ becomes the input to $\mathcal{N}_G$ that returns $\mathcal{Z}$, i.e.,
\begin{equation*}
    (X_1, Y_1, \dotsc, X_N, Y_N) \xrightarrow{\mathcal{A}_{P, f}} (\overline{X}_1, \overline{Y}_1, \dotsc, \overline{X}_N, \overline{Y}_N) \xrightarrow{\mathcal{N}_G} Z
\end{equation*}
forms a Markov chain. From the data processing inequality~\cite{thomas2006elements}, we obtain
\begin{equation}
    I(X_1, Y_1, \dotsc, X_N, Y_N; Z) \leq I(\overline{X}_1, \overline{Y}_1, \dotsc, \overline{X}_N, \overline{Y}_N; Z). \label{eqn:mutual_info_ineq_strategy_assistance}
\end{equation}
Then, using Eq.~\eqref{eqn:correlation_assisted_achievable_sum_rate} and Eq.~\eqref{eqn:mutual_info_ineq_strategy_assistance}, we get
\begin{equation}
    S_{\corr}(\mathcal{N}_G) \leq \sup_{\overline{p}} I(\overline{X}_1, \overline{Y}_1, \dotsc, \overline{X}_N, \overline{Y}_N; Z) \label{eqn:C-assisted_sum_rate_MI_upper_bound}
\end{equation}
where the probability distribution $\overline{p}$ defined in Eq.~\eqref{eqn:strategy_assistance_channel_induced_probability} is obtained by varying product distributions $p(x, y) = \prod_{i = 1}^N p^{(i)}(x_i, y_i)$ input to $\mathcal{A}_{P, f}$, the correlation $P \in \mathcal{C}$, and the post-processing $f \in \text{PP}$.

We can reinterpret the above equation as a maximum over distributions over questions and strategies for playing the game $G$ induced by the correlations $\corr$. First, we write $p^{(i)} = \pi^{(i)} p_{Y_i | X_i}$, where $\pi^{(i)}$ is a distribution over the questions $\mathcal{X}_i$ and $p_{Y_i | X_i}$ is a strategy chosen by $i$th player for $i \in [N]$. Therefore, the input probability distribution in Eq.~\eqref{eqn:correlation_assisted_achievable_sum_rate} can be written as $p^{(1)} \dotsm p^{(N)} = \pi p_{\bm{Y} | \bm{X}}$, where $\pi = \pi^{(1)} \dotsm \pi^{(N)}$ is a distribution over questions $\mathcal{X}$, and $p_{\bm{Y} | \bm{X}} = \prod_{i = 1}^N p_{Y_i | X_i}$ is a classical strategy chosen by the players. In particular, the input strategy $p_{\bm{Y} | \bm{X}}$ is always a classical strategy. As noted in Eq.~\eqref{eqn:assistance_induced_strategy}, the channel $\mathcal{A}_{P, f}$ takes this input strategy and returns a new strategy $\mathcal{A}_{P, f}(p_{\bm{Y} | \bm{X}})$ that incorporates assistance from the shared correlation $P$. Since the senders have access to the set of correlations $\corr$, we can write Eq.~\eqref{eqn:C-assisted_sum_rate_MI_upper_bound} as
\begin{equation}
    S_{\corr}(\mathcal{N}_G) \leq \sup_{\pi^{(1)} \dotsm \pi^{(N)}} \sup_{\overline{p}_{\bm{Y} | \bm{X}} \in \strat_\corr} I(\overline{X}_1, \overline{Y}_1, \dotsc, \overline{X}_N, \overline{Y}_N; Z),
    \label{eqn:C-assisted_sum_rate_MI_upper_bound_strategy}
\end{equation}
where $\strat_\corr$ is the set of strategies induced by $\corr$ as defined in Eq.~\eqref{eqn:correlation_induced_strategies}. To obtain an upper bound, we will perform relaxations of the RHS of the above equation, and solve the resulting optimization problems.

We begin by writing the RHS of Eq.~\eqref{eqn:C-assisted_sum_rate_MI_upper_bound_strategy} in a form that is amenable for calculations. To that end, note that given an input probability distribution $\overline{p}(\overline{\bm{xy}}) = \pi(\overline{\bm{x}}) \overline{p}_{\bm{Y} | \bm{X}}(\overline{\bm{y}} | \overline{\bm{x}})$, the probability distribution corresponding to the output of the channel $\mathcal{N}_G$ is given by
\begin{align}
    p(\bm{z}) &= \sum_{\overline{\bm{xy}} \in \mathcal{W}} \mathcal{N}_G(\bm{z} | \bm{xy}) \overline{p}(\overline{\bm{xy}})
    + \sum_{\overline{\bm{xy}} \notin \mathcal{W}} \mathcal{N}_G(\bm{z} | \bm{xy}) \overline{p}(\overline{\bm{xy}}) \nonumber \\
              &= \pi(\bm{z}) \sum_{\bm{y} : \bm{zy} \in \mathcal{W}} \overline{p}_{\bm{Y} | \bm{X}}(\bm{y} | \bm{z}) + \frac{1}{d} p_L \label{eqn:outputprob_elementwise}
\end{align}
where $d = |\mathcal{Z}|$ denotes the number of question pairs, while $p_L = \sum_{\overline{\bm{xy}} \notin \mathcal{W}} \overline{p}(\overline{\bm{xy}})$ denotes the probability of losing the game when questions are drawn as per the probability distribution $\pi$.

Note that $\mathcal{Z} = \mathcal{X}_1 \times \dotsm \times \mathcal{X}_N = \mathcal{X}$ is the set of question tuples output by the \ngmac $\mathcal{N}_G$. Since $\mathcal{Z}$ is a finite set of size $d$, we can fix a labelling for the elements of $\mathcal{Z}$ and write $\mathcal{Z} = \{\bm{z}_1, \dotsc, \bm{z}_d\}$. Each $\bm{z}_i$ corresponds to a particular question tuple. Then, we can define the contribution of a given strategy $\overline{p}_{\bm{Y} | \bm{X}}$ towards winning the game for each question tuple. We use $\mathcal{Z}$ and $\mathcal{X}$ interchangeably in the following discussion.
\begin{definition}[Winning vector]
    Given a strategy $\overline{p}_{\bm{Y} | \bm{X}}$ for playing the game $G = (\mathcal{X}, \mathcal{Y}, \mathcal{W})$, we let
    \begin{equation}
        w_i = \sum_{\bm{y} : \bm{z}_i \bm{y} \in \mathcal{W}} \overline{p}_{\bm{Y} | \bm{X}}(\bm{y} | \bm{z}_i) \qquad\text{for $i \in [d]$} \label{eqn:strategy_coeff}
    \end{equation}
    denote the contribution of the strategy towards winning the game $G$ for question $\bm{z}_i$. We call the vector $\bm{w} = (w_1, \dotsc, w_d)$ the \textit{winning vector} corresponding to the strategy $\overline{p}_{\bm{Y} | \bm{X}}$.
    Let $\winvec_\corr$ denote the set of winning vectors allowed by the strategies $\strat_\corr$,
    \begin{equation}
        \winvec_\corr = \left\{\bm{w} \in [0, 1]^d \mid  w_i = \sum\nolimits_{\bm{y} : \bm{z}_i \bm{y} \in \mathcal{W}} \overline{p}_{\bm{Y} | \bm{X}}(\bm{y} | \bm{z}_i),\ \text{ for }\overline{p}_{\bm{Y} | \bm{X}} \in \strat_\corr \text{ and $i \in [d]$.}\right\}.
        \label{eqn:strategy_coeff_set}
    \end{equation}
\end{definition}
Observe that $w_i \in [0, 1]$ for all $i \in [d]$, so that $\bm{w}$ is an element of the unit hypercube in $\mathbb{R}^d$. Note that we have $w_i = 1$ for a fixed strategy $\overline{p}_{\bm{Y} | \bm{X}}$ if and only if the players always win the game $G$ when asked the question $\bm{z}_i$ using the strategy $\overline{p}_{\bm{Y} | \bm{X}}$. On the other extreme, $w_i = 0$ if and only if the players always lose the game $G$ when asked the question $\bm{z}_i$ using the strategy $\overline{p}_{\bm{Y} | \bm{X}}$. Generally, questions are drawn with probability $\pi$ over $\mathcal{X}$. The probability of winning the game for question $i$ is $\pi_i w_i$, where $\pi_i = \pi(\bm{z}_i)$ is the probability of drawing the question tuple $\bm{z}_i$. The total probability of winning the game is $p_W = \sum_{i = 1}^d \pi_i w_i$ and the probability of losing the game is
\begin{equation}
    p_L = 1 - \sum_{i = 1}^d \pi_i w_i = 1 - \ip{\pi, \bm{w}}. \label{eqn:losingprob}
\end{equation}
Defining the matrix $\overline{W}$ with components
\begin{equation}
    \overline{W}_{ij} = w_i \delta_{ij} + \frac{1 - w_j}{d}, \label{eqn:W_matrix_elts}
\end{equation}
one may write the output probability $p(\bm{z})$ in Eq.~\eqref{eqn:outputprob_elementwise} as
\begin{equation}
    p = \overline{W} \pi. \label{eqn:outputprob}
\end{equation}

The mutual information $I(\overline{X}_1, \overline{Y}_1, \dotsc, \overline{X}_N, \overline{Y}_N; Z)$ can be written as
\begin{align}
    I(\overline{X}_1, \overline{Y}_1, \dotsc, \overline{X}_N, \overline{Y}_N; Z)  &= H(Z) - H(Z | \overline{X}_1, \overline{Y}_1, \dotsc, \overline{X}_N, \overline{Y}_N)\notag\\
                                                                                  &= H(Z) - p_L \ln(d), \label{eq:MI-formula}
\end{align}
where we used the fact that $H(Z | \overline{\bm{xy}}) = 0$ when $\overline{\bm{xy}} \in \mathcal{W}$ whereas $H(Z | \overline{\bm{xy}}) = \ln(d)$ when $\overline{\bm{xy}} \notin \mathcal{W}$.
Note that the formula \eqref{eq:MI-formula} was first derived in Ref.~\cite{leditzky2020playing} for nonlocal games MAC with two senders.
Using Eq.~\eqref{eqn:outputprob} and Eq.~\eqref{eqn:losingprob}, we obtain
\begin{equation}
    \mathscr{I}_{\bm{w}}(\pi) \coloneqq I(\overline{X}_1, \overline{Y}_1, \dotsc, \overline{X}_N, \overline{Y}_N; Z) = H(\overline{W} \pi) + \ip{\pi, \bm{w}} \ln(d) - \ln(d) \label{eqn:mutualinfo}
\end{equation}
where the notation, $\mathscr{I}_{\bm{w}}(\pi)$, for the mutual information emphasizes that it is only a function of the distribution $\pi$ over questions and the winning vector $\bm{w}$.

The RHS in Eq.~\eqref{eqn:C-assisted_sum_rate_MI_upper_bound_strategy} can be written as
\begin{align}
    \sup_{\pi^{(1)} \dotsm \pi^{(N)}} \sup_{\overline{p}_{\bm{Y} | \bm{X}}} I(\overline{X}_1, \overline{Y}_1, \dotsc, \overline{X}_N, \overline{Y}_N; Z)
                &= \sup_{\pi^{(1)} \dotsm \pi^{(N)}} \sup_{\bm{w} \in \winvec_\corr} \mathscr{I}_{\bm{w}}(\pi^{(1)} \dotsm \pi^{(N)}) \nonumber \\
                &\leq \sup_{\pi \in \Delta_d} \sup_{\bm{w} \in \winvec_\corr} \mathscr{I}_{\bm{w}}(\pi), \label{eqn:mutualinfo_optimization_relaxation_intermediate}
\end{align}
where $\winvec_\corr$ is the set of winning vectors defined in Eq.~\eqref{eqn:strategy_coeff_set}. To obtain Eq.~\eqref{eqn:mutualinfo_optimization_relaxation_intermediate}, we relax the product distribution constraint $\pi^{(1)} \dotsm \pi^{(N)}$ over the questions to obtain a maximization over all distribution $\pi \in \Delta_d$ over the questions, where $\Delta_d$ denotes the $(d - 1)$-dimensional standard simplex. This relaxation differs from that of Eq.~\eqref{eqn:relaxed_AL_capacity} used in obtaining $C(\mathcal{N}_G)$ in that we only relax the distribution over the questions, but not the whole probability distribution. 

For a fixed $\bm{w}$, the function $\mathscr{I}_{\bm{w}}(\pi)$ is continuous in $\pi$ over the compact set $\Delta_d$. Thus, the maximization in Eq.~\eqref{eqn:mutualinfo_optimization_relaxation_intermediate} can be written as
\begin{equation}
    \sup_{\bm{w} \in \winvec_\corr} \max_{\pi \in \Delta_d} \mathscr{I}_{\bm{w}}(\pi). \label{eqn:mutualinfo_optimization_relaxation}
\end{equation}
The inner optimization in Eq.~\eqref{eqn:mutualinfo_optimization_relaxation} is a convex problem since $\mathscr{I}_{\bm{w}}(\pi)$ is concave in $\pi$ and $\Delta_d$ is a convex set. However, $\mathscr{I}_{\bm{w}}(\pi)$ is not jointly concave in $\pi$ and $\bm{w}$, and moreover, $\winvec_\corr$ need not be a convex set. Therefore, the optimization in Eq.~\eqref{eqn:mutualinfo_optimization_relaxation} is generally nonconvex.

Our goal is to obtain an upper bound on the optimization in Eq.~\eqref{eqn:mutualinfo_optimization_relaxation}. To give a general idea of our approach to obtaining this bound, we list the main steps we will carry out.
\begin{itemize}[leftmargin=1.5cm]
    \item[Step 1:] For a fixed $\bm{w}$, we obtain an upper bound on the inner optimization in Eq.~\eqref{eqn:mutualinfo_optimization_relaxation}. This bound is tight when either $\bm{w} \in \{0, 1\}^d$ or $\bm{w} > 0$ component-wise.
    \item[Step 2:] 
        We relax the set of allowed $\bm{w}$ values to bound the outer optimization in Eq.~\eqref{eqn:mutualinfo_optimization_relaxation} from above.
\end{itemize}

This procedure will result in the upper bound noted in Eq.~\eqref{eqn:NG_MAC_correlation_assisted_sum_capacity_bound}.
We explain the steps in detail in the following subsections.

\subsubsection*{Step 1: Bounding the inner optimization over question distributions}
We obtain an upper bound on $\max_{\pi \in \Delta_d} \mathscr{I}_{\bm{w}}(\pi)$ by considering two cases. First, we perform this optimization exactly for the case when $\bm{w} \in \{0, 1\}^d$. Next, for any $\bm{w} \in \winvec_\corr$, we find an upper bound on $\max_{\pi \in \Delta_d} \mathscr{I}_{\bm{w}}(\pi)$ using the result of case 1. The upper bound obtained in case 2 is tight when $\bm{w} > 0$.

\noindent \textbf{Case 1: optimizing $\max_{\pi \in \Delta_d} \mathscr{I}_{\bm{w}}(\pi)$ for fixed $\bm{w} \in \{0, 1\}^d$} \newline
Winning vectors $\bm{w} \in \{0, 1\}^d$ arise from strategies that either always win or always lose the game for any given question. Deterministic strategies, for example, give rise to such winning vectors. Recall that a classical deterministic strategy corresponds to functions $f_i: \mathcal{X}_i \to \mathcal{Y}_i$, $i \in [N]$, chosen by the players. Such functions give rise to the classical strategy $p_{Y_i | X_i}(y | x) = \delta_{y, f(x)}$ that is $1$ at $y = f_i(x)$ and zero elsewhere. It follows from Eq.~\eqref{eqn:strategy_coeff} that $w_i \in \{0, 1\}$ for all $i \in \{1, \dotsc, d\}$, where $\bm{w}$ is the winning vector defined by such a deterministic strategy.

The following proposition gives the result of the optimization $\max_{\pi \in \Delta_d} \mathscr{I}_{\bm{w}}(\pi)$ when $\bm{w} \in \{0, 1\}^d$. Note that $\winvec_\corr$ need not contain such winning vectors. Computing $\max_{\pi \in \Delta_d} \mathscr{I}_{\bm{w}}(\pi)$ for $\bm{w} \in \{0, 1\}^d$ is a means to providing a bound for Eq.~\eqref{eqn:mutualinfo_optimization_relaxation}.
\begin{proposition}
    \label{prop:deterministic_strategy_maximum}
    Let $G$ be a nonlocal game, and let $\mathcal{N}_G$ be the MAC obtained from this nonlocal game. Let $\bm{w}$ denote a winning vector as defined in Eq.~\eqref{eqn:strategy_coeff}, such that $w_i \in \{0, 1\}$ for all $i \in [d]$. Let $\mathscr{I}_{\bm{w}}$ denote mutual information between input and outputs of $\mathcal{N}_G$, as defined in Eq.~\eqref{eqn:mutualinfo}. Let $\mathcal{K} = \{i \in [d] \mid w_i = 1\}$ denote the set of questions for which the strategy gives a correct answer. Denoting $K = |\mathcal{K}|$ and $\Delta_d$ to be the $(d - 1)$-dimensional standard simplex, we have
    \begin{equation}
        \max_{\pi \in \Delta_d} \mathscr{I}_{\bm{w}}(\pi) = \begin{cases} 0 & K = 0 \\
                                                                          \mathscr{I}^*_{K} &0 < K < d \\
                                                                          \ln d       &K = d.
                                                            \end{cases}
    \end{equation}
    The quantity $\mathscr{I}^*_K$ is given by the expression
    \begin{equation}
        \mathscr{I}^*_{K} = \ln\left(K + (d - K) d^{-\frac{d}{d - K}}\right). \label{eqn:deterministic_strategy_maximum_mutualinfo}
    \end{equation}
\end{proposition}
\begin{proof}
    See Appendix~\ref{proof:deterministic_strategy_maximum}.
\end{proof}

Observe that the maximum only depends on the total number $d$ of questions, as well as the number $K$ of questions that can be answered correctly using the deterministic strategy.

\noindent \textbf{Case 2: Bounding $\max_{\pi \in \Delta_d} \mathscr{I}_{\bm{w}}(\pi)$ for fixed $\bm{w} \in \winvec_\corr$} \newline
\noindent When we work with arbitrary winning vectors, it is more challenging to maximize the mutual information over distributions on the questions. To make this maximization easier, we first show that the maximum mutual information $\mathscr{I}^*_{d - 1}$ corresponding to a winning vector that can answer exactly $d - 1$ questions correctly will always be larger than the maximum mutual information $\max_{\pi \in \Delta_d} \mathscr{I}_{\bm{w}}(\pi)$ for any winning vector $\bm{w}$ that answers no more than $d - 1$ questions correctly. We therefore turn our attention to $\bm{w}$ that doesn't necessarily do worse than this case, and obtain an expression for the maximum mutual information in terms of such $\bm{w}$.

\begin{proposition}
    \label{prop:arbitrary_strategy_maximum_nearperfectgame}
    Let $G$ be a nonlocal game, and let $\mathcal{N}_G$ be the MAC obtained from this nonlocal game. Suppose that the senders of $\mathcal{N}_G$ share a set of correlations $\corr$. Let $\bm{w} \in \winvec_\corr$ be any winning vector allowed by the correlations $\corr$ as defined in Eq.~\eqref{eqn:strategy_coeff_set}. Let $\mathcal{K} = \{i \in [d] \mid w_i \neq 0\}$ be the set of questions with non-zero probability of winning the game using this strategy, and denote $K = |\mathcal{K}|$. Then, the following statements hold.
    \begin{enumerate}[label=\arabic*)]
        \item Suppose that $\max_{\pi \in \Delta_d} \mathscr{I}_{\bm{w}}(\pi)$ is achieved at $\pi^*$. Denote $\mathcal{K}^* = \{i \in [d] \mid w_i \pi_i^* \neq 0\}$ and $K^* = |\mathcal{K}^*|$ $\lb$we have $\mathcal{K}^* \subseteq \mathcal{K}$$\rb$. Then, if $K^* < d$, we have $\max_{\pi \in \Delta_d} \mathscr{I}_{\bm{w}}(\pi) \leq \mathscr{I}^*_{d - 1}$, where $\mathscr{I}^*_{d - 1}$ is given by Eq.~\eqref{eqn:deterministic_strategy_maximum_mutualinfo}.
        \item As a consequence of the above result, we restrict our attention to strategies with $K^* = K = d$. In that case, we have
            \begin{align}
                \max_{\pi \in \Delta_d} \mathscr{I}_{\bm{w}}(\pi) &\equiv \mathscr{I}^*(\bm{w}) = \ln\left(\sum_{j = 1}^d \exp\left[d w_{\textnormal{eff}} \ln d \left(1 - \frac{1}{w_j}\right)\right]\right)
                                                                                                                                                                                                        \label{eqn:mutualinfo_w}
                \intertext{where}
                w_{\textnormal{eff}} &= \left(\sum_{i = 1}^d \frac{1}{w_i}\right)^{-1}. \label{eqn:weff}
            \end{align}
    \end{enumerate}
\end{proposition}
\begin{proof}
    See Appendix~\ref{proof:arbitrary_strategy_maximum_nearperfectgame}.
\end{proof}

Owing to the above result, we only need to focus on maximizing $\mathscr{I}^*(\bm{w})$ for those $\bm{w} \in \winvec_\corr$ with $w_i > 0$ for all $i \in [d]$. This is done in the next step.

\subsubsection*{Step 2: Bounding the outer optimization over winning vectors}
As noted in the previous step, our goal is to maximize $\mathscr{I}^*(\bm{w})$ with respect to the feasible winning vectors $\bm{w} \in \winvec_\corr$ with $w_i > 0$ for all $i \in [d]$. The set of (feasible) winning vectors $\winvec_\corr$ was defined in Eq.~\eqref{eqn:strategy_coeff_set}. Note that $\winvec_\corr$ depends on the winning condition $\mathcal{W}$ of the game as well as the set of correlations $\corr$ shared by the senders. Since we make no assumptions about the game or the set of correlations, it is difficult to optimize over $\winvec_\corr$. For this reason, we obtain a relaxation of the set $\winvec_\corr$, over which we optimize $\mathscr{I}^*(\bm{w})$. We will do this in two steps: (1) relate $\bm{w}$ to the winning probability when the questions are drawn uniformly, and (2) use the maximum winning probability $\omega^{\strat_\corr}(G)$ of the game (assumed to be known) corresponding to the strategies $\strat_\corr$ when the questions are drawn uniformly in order to obtain a convex set containing $\winvec_\corr$.

(1) From the definition of winning vector given in Eq.~\eqref{eqn:strategy_coeff}, we know that
\begin{equation*}
    w_i = \sum_{\bm{y} : \bm{z}_i \bm{y} \in \mathcal{W}} \overline{p}_{\bm{Y} | \bm{X}}(\bm{y} | \bm{z}_i).
\end{equation*}
Recall that the winning probability of the game can be written as $p_W = \sum_{i = 1}^d \pi_i w_i$ when the questions are drawn as per probability $\pi \in \Delta_d$. If the questions are drawn uniformly, then $\pi_U(z) = 1/d$ for all questions $\bm{z}_1, \dotsc, \bm{z}_d$. Therefore, $p_W = \sum_{i = 1}^d w_i / d$ is the winning probability determined by the winning vector $\bm{w}$ when the questions are drawn uniformly.

(2) We now look for a convex relaxation $\overline{\winvec_\corr}$ of $\winvec_\corr$. We want to make $\overline{\winvec_\corr}$ fairly independent of the winning set, except for dependence on $\omega^{\strat_\corr}(G)$ and the number of question tuples $d$ in the game.

Since $\omega^{\strat_\corr}(G)$ is the maximum winning probability using the set of strategies $\strat_\corr$ under consideration, we must have
\begin{equation}
    \frac{1}{d} \sum_{i = 1}^d w_i \leq \omega^{\strat_\corr}(G) \label{eqn:maximum_winning_prob_uniform}
\end{equation}
where $\bm{w} \in \winvec_\corr$. Now we make the relaxation that we allow \textit{any} winning vector that satisfies Eq.~\eqref{eqn:maximum_winning_prob_uniform}. Consequently, we define
\begin{equation}
    \overline{\winvec_\corr} = \left\{\bm{w} \in [0, 1]^d\ \Big|\ \frac{1}{d}\sum\nolimits_{i = 1}^d w_i \leq \omega^{\strat_\corr}(G)\right\}. \label{eqn:strategy_coeff_set_relaxation}
\end{equation}
Since any $\bm{w} \in \winvec_\corr$ will satisfy Eq.~\eqref{eqn:maximum_winning_prob_uniform}, we have $\winvec_\corr \subseteq \overline{\winvec_\corr}$, confirming that $\overline{\winvec_\corr}$ is a relaxation of $\winvec_\corr$. Such a relaxation may allow for strategies not described by $\strat_\corr$. Note that $\overline{\winvec_\corr}$ is a compact and convex set, and it depends only on the maximum winning probability and the number of questions in the game. Using this relaxation, we compute an upper bound on $\mathscr{I}^*(\bm{w})$ maximized over $\bm{w} \in \overline{\winvec_\corr}$ satisfying $\bm{w} > 0$ componentwise.

\begin{proposition}
    \label{prop:strategy_coeff_maximization_upper_bound}
    Let $G$ be a nonlocal game and let $\mathcal{N}_G$ be the \ngmac constructed from $G$, as defined in Eq.~\eqref{eqn:MAC_nonlocal_game}. Suppose that the senders share the set of correlations $\corr$, and let $\winvec_\corr$ be the corresponding set of winning vectors as defined in Eq.~\eqref{eqn:strategy_coeff_set}. Let $\overline{\winvec_\corr}$ be the convex relaxation of $\winvec_\corr$ defined in Eq.~\eqref{eqn:strategy_coeff_set_relaxation} that depends only on the number of question tuples $d$ in the game and the maximum winning probability $\omega^{\strat_\corr}(G)$ when the questions are drawn uniformly and answers given using strategies in $\strat_\corr$. Let $\mathscr{I}^*(\bm{w})$ be the function defined in Eq.~\eqref{eqn:mutualinfo_w}. Then the maximum of $\mathscr{I}^*(\bm{w})$ over winning vectors $\bm{w} > 0$ in $\overline{\winvec_\corr}$ is bounded from above by
    \begin{equation}
        \sup_{\bm{w} \in \overline{\winvec_\corr}, \bm{w} > 0} \mathscr{I}^*(\bm{w}) \leq \ln\left(d - 1 + d^{-(1 - \omega^{\strat_\corr}(G)) d}\right) \label{eqn:strategy_coeeff_maximization_upper_bound}
    \end{equation}
\end{proposition}
\begin{proof}
    See Appendix~\ref{proof:strategy_coeff_maximization_upper_bound}.
\end{proof}

\subsubsection*{Bound on the correlation-assisted achievable sum rate}
:e put all the above steps together to obtain a bound on $S_\corr(\mathcal{N}_G)$.

\begin{theorem}
    \label{thm:NG_MAC_correlation_assisted_sum_capacity_bound}
    Let $G$ be an $N$-player promise-free nonlocal game with $d$ question tuples, and let $\mathcal{N}_G$ be the MAC obtained from $G$ as defined in Eq.~\eqref{eqn:MAC_nonlocal_game}. Suppose that the senders share a set of correlations $\corr$. Let $\strat_\corr$ be the set of strategies induced by the correlations as defined in Eq.~\eqref{eqn:correlation_induced_strategies}. Let $\omega^{\strat_\corr}(G)$ denote the maximum winning probability of the game when the questions are drawn uniformly and answers are obtained using strategies in $\strat_\corr$. Let $S_{\corr}(\mathcal{N}_G)$ denote the $\mathcal{C}$-assisted achievable sum rate of the \ngmac $\mathcal{N}_G$ as defined in Eq.~\eqref{eqn:correlation_assisted_achievable_sum_rate}. Then, we have
    \begin{equation}
        S_{\corr}(\mathcal{N}_G) \leq \ln\left(d - 1 + d^{-(1 - \omega^{\strat_\corr}(G)) d}\right) \label{eqn:NG_MAC_correlation_assisted_sum_capacity_bound}
    \end{equation}
    with entropy measured in nats.
\end{theorem}
\begin{proof}
    To obtain an upper bound on $S_{\corr}(\mathcal{N}_G)$, we start with Eq.~\eqref{eqn:C-assisted_sum_rate_MI_upper_bound_strategy}. The RHS of Eq.~\eqref{eqn:C-assisted_sum_rate_MI_upper_bound_strategy} can be bounded by performing the maximization $\sup_{\bm{w} \in \winvec_\corr} \max_{\pi \in \Delta_d} \mathscr{I}_{\bm{w}}(\pi)$, where $\mathscr{I}_{\bm{w}}(\pi)$ is the mutual information defined in Eq.~\eqref{eqn:mutualinfo}. The set $\Delta_d$ denote the $(d - 1)$-dimensional standard simple, while $\winvec_\corr$ denotes the set of winning vectors induced by the correlations $\corr$ as defined in Eq.~\eqref{eqn:strategy_coeff_set}.

    In Prop.~\ref{prop:arbitrary_strategy_maximum_nearperfectgame}, we show that if $\bm{w} \in \winvec_\corr$ has one or more zero entries, then $\max_{\pi \in \Delta_d} \mathscr{I}_{\bm{w}}(\pi) \leq \mathscr{I}^*_{d - 1}$, where $\mathscr{I}^*_{d - 1}$ is given by Eq.~\eqref{eqn:deterministic_strategy_maximum_mutualinfo}. Therefore, we only maximize $\mathscr{I}^*(\bm{w}) = \max_{\pi \in \Delta_d} \mathscr{I}_{\bm{w}}(\pi)$ over winning vectors $\bm{w} \in \winvec_\corr$ satisfying $\bm{w} > 0$. The expression for $\mathscr{I}^*(\bm{w})$ in this case is given by Eq.~\eqref{eqn:mutualinfo_w}. We relax the set $\winvec_\corr$ to the compact and convex set $\overline{\winvec_\corr}$ defined in Eq.~\eqref{eqn:strategy_coeff_set_relaxation}. Then, we give an upper bound on $\sup_{\bm{w} \in \overline{\winvec_\corr}, \bm{w} > 0} \mathscr{I}^*(\bm{w})$ in Eq.~\eqref{eqn:strategy_coeeff_maximization_upper_bound}.

    By preceding remarks, we have
    \begin{equation}
        S_{\corr}(\mathcal{N}_G) \leq \max\left\{\mathscr{I}^*_{d - 1},\ \ln\left(d - 1 + d^{-(1 - \omega^{\strat_\corr}(G)) d}\right)\right\}, \label{eqn:NG_MAC_correlation_assisted_sum_capacity_bound_intermediate}
    \end{equation}
    while from Eq.~\eqref{eqn:deterministic_strategy_maximum_mutualinfo} we have
    \begin{equation*}
        \mathscr{I}^*_{d - 1} = \ln\left(d - 1 + d^{-d}\right).
    \end{equation*}
    Using $\ln\left(d - 1 + d^{-(1 - \omega^{\strat_\corr}(G)) d}\right) \geq \mathscr{I}^*_{d - 1}$ in Eq.~\eqref{eqn:NG_MAC_correlation_assisted_sum_capacity_bound_intermediate}, we obtain Eq.~\eqref{eqn:NG_MAC_correlation_assisted_sum_capacity_bound}.
\end{proof}

\begin{corollary}
    \label{cor:NG_MAC_sum_capacity_bound}
    Let $G$ be an $N$-player promise-free nonlocal game with $d$ question tuples, and let $\mathcal{N}_G$ be the MAC obtained from $G$ as defined in Eq.~\eqref{eqn:MAC_nonlocal_game}. Let $\omega^{\text{cl}}(G)$ denote the maximum winning probability of the game when the questions are drawn uniformly and answers are obtained using classical strategies. Let $S(\mathcal{N}_G)$ denote the sum capacity of the \ngmac $\mathcal{N}_G$. Then, we have
    \begin{equation}
        S(\mathcal{N}_G) \leq \ln\left(d - 1 + d^{-(1 - \omega^{\text{cl}}(G)) d}\right) \label{eqn:NG_MAC_sum_capacity_bound}
    \end{equation}
    with entropy measured in nats.
\end{corollary}
\begin{proof}
    From Eq.~\eqref{eqn:correlation_assisted_sum_capacity_hierarchy}, we know that $S(\mathcal{N}_G) \leq S_{\text{cl}}(\mathcal{N}_G)$. Then, using Thm.~\ref{thm:NG_MAC_correlation_assisted_sum_capacity_bound}, we obtain Eq.~\eqref{eqn:NG_MAC_sum_capacity_bound}.
\end{proof}

\noindent Note that the bounds on $S_{\corr}(\mathcal{N}_G)$ and $S(\mathcal{N}_G)$ given by Eq.~\eqref{eqn:NG_MAC_correlation_assisted_sum_capacity_bound} and Eq.~\eqref{eqn:NG_MAC_sum_capacity_bound}, respectively, lie between $\ln(d - 1)$ and $\ln(d)$. For sufficiently large $d$, when $\omega^{\strat_\corr}(G)$ is not close to $1$, the sum capacity is bounded above by $\approx \ln(d - 1)$. On the other hand, for $\omega^{\strat_\corr}(G) = 1$, we obtain an upper bound of $\ln(d)$, which can be achieved by $\bm{w} = (1, \dots, 1)^T$ as seen from Prop.~\ref{prop:deterministic_strategy_maximum}. Using this, we can obtain separations between the correlation-assisted achievable sum rate corresponding to two different sets of correlations.

\subsubsection{Separation between sum rates with assistance from different sets of correlations}
If $\corr_1$ and $\corr_2$ are two sets of correlations such that $\omega^{\strat_{\corr_1}}(G) < 1$ while $\omega^{\strat_{\corr_2}}(G) = 1$, then $S_{\corr_1}(\mathcal{N}_G) < S_{\corr_2}(\mathcal{N}_G) = \ln(d)$. We use this idea to provide separations of correlation-assisted achievable sum rate using classical, quantum and no-signalling correlations.

\subsubsection*{\texorpdfstring{Separating $S_{\textnormal{Q}}(\mathcal{N}_G)$ from $S(\mathcal{N}_G)$ for two-sender MACs}{Q-cl-2senders}}
Consider the Magic Square Game, $G_{\text{MS}}$, used previously
in~\cite{mermin1990simple, peres1990incompatible,aravind2002simple,brassard2005quantum} to obtain  a separation between $S(\mathcal{N}_G)$ and
$S_{\text{Q}}(\mathcal{N}_G)$. In this game, the referee selects a row $r \in
\{1, 2, 3\}$ and column $c \in \{1, 2, 3\}$ from a $3 \times 3$ grid uniformly
at random. The row is handed over to Alice while the column is given to Bob.
Without communicating with each other, Alice \& Bob need to fill bits in the
given row and column such that the total parity of bits in the row is even,
total parity of bits in the column is odd, and the bit at the intersection of
the given row and column match. 

There are $d = 9$ possible question pairs corresponding to the indices $(r, c)$. Classically, Alice \& Bob can win the game at least $8$ out of $9$ times by implementing for example the following strategy:
\begin{center}
\begin{tabular}{|c|c|c|}
    \hline & & \\[-0.3cm]
    1 & 0 & 1 \\
    \hline & &  \\[-0.3cm]
    1 & 1 & 0 \\
    \hline & &  \\[-0.3cm]
    1 & 0 & ? \\
    \hline
\end{tabular}
\end{center}
where the entry in each box indicates bits filled by Alice and Bob. It can be
shown that this strategy is optimal, therefore
$\omega^{\text{cl}}(g_{\text{ms}}) = 8/9$~\cite{brassard2005quantum}.

On the other hand, if Alice \& Bob are allowed to use a quantum strategy, then
they can share two copies of a maximally entangled Bell state,
\begin{equation*}
    \rho_{\text{Bell}} = \frac{1}{2} \begin{pmatrix} 1 & 0 & 0 & 1 \\
                                                     0 & 0 & 0 & 0 \\
                                                     0 & 0 & 0 & 0 \\
                                                     1 & 0 & 0 & 1 \end{pmatrix},
\end{equation*}
and submit answers to the referee based on a set of measurements given in Table~\ref{eq:mTable}. 
Alice and Bob answer $0$ if their measurement yields an eigenvector with
eigenvalue $1$, else they answer $1$.
Answers obtained this way can be shown to always satisfy the winning condition
of the magic square game, i.e., $\omega^{\text{Q}}(G_{\text{MS}}) = 1$ \cite{mermin1990simple, peres1990incompatible}.

\begin{table}
    \centering
    \caption{Measurement Table}\label{eq:mTable}
\begin{tabular}{|c|c|c|}
    \hline & & \\[-0.25cm]
    $\sigma_x \otimes \id$ & $\sigma_x \otimes \sigma_x$ & $\id \otimes \sigma_x$ \\[0.1cm]
    \hline & &  \\[-0.25cm]
    $-\sigma_x \otimes \sigma_z$ & $\sigma_y \otimes \sigma_y$ & $-\sigma_z \otimes \sigma_x$ \\[0.1cm]
    \hline & &  \\[-0.25cm]
    $\id \otimes \sigma_z$ & $\sigma_z \otimes \sigma_z$ & $\sigma_z \otimes \id$ \\[0.1cm]
    \hline
\end{tabular}
\end{table}

As the MAC obtained from the Magic Square Game has received attention in a previous study, we summarize the separation between sum capacity and entanglement-assisted sum rate given by our method in the following corollary.
\begin{corollary}
    \label{cor:NG_MAC_MagicSquareGame_Separation}
    Let $\mathcal{N}_{G_{\text{MS}}}$ denote the MAC obtained from the magic square game. Then, the sum capacity of this MAC is bounded above as
    \begin{equation*}
        S(\mathcal{N}_{G_{\text{MS}}}) \leq 3.02\ \textnormal{bits}.
    \end{equation*}
    On the other hand, using assistance from quantum correlations, we obtain $S_{\text{Q}}(\mathcal{N}_{G_{\text{MS}}}) = 3.17\ \textnormal{bits}$. This gives a separation of at least $0.15\ \textnormal{bits}$ between sum rate with and without entanglement assistance.
\end{corollary}
\begin{proof}
    Since $\omega^{\text{cl}}(G_{\text{MS}}) = 8/9$, we may use
Cor.~\eqref{cor:NG_MAC_sum_capacity_bound} to obtain a bound
$S(\mathcal{N}_{G_{\text{MS}}}) \leq 3.02$ bits.  At the same time, a perfect
quantum strategy is available, i.e., $\omega^{\text{Q}}(G_{\text{MS}}) = 1$,
and thus $S_{\text{Q}}(\mathcal{N}_{G_{\text{MS}}}) = 3.17$ bits \cite{leditzky2020playing}.
\end{proof}
Our bound of $3.02$ bits on the sum capacity of the \ngmac $\mathcal{N}_{G_{\text{MS}}}$ is
tighter than the previously reported bound of $3.14$ bits~\cite{leditzky2020playing}. Thus, our bound shows that entanglement assistance increases the sum rate by at least $4.7 \%$, in comparison with the previously known result of $0.9 \%$.

Since every quantum strategy is also a no-signalling strategy, we automatically
obtain a separation between $S(\mathcal{N}_{G_{\text{MS}}})$ and
$S_{\text{NS}}(\mathcal{N}_{G_{\text{MS}}})$. However,
$S_{\text{NS}}(\mathcal{N}_{G_{\text{MS}}}) =
S_{\text{Q}}(\mathcal{N}_{G_{\text{MS}}})$. In the following section, we use
a game different from $G_{\text{MS}}$ to obtain
a separation between the quantum and no-signalling assisted sum rates.

\subsubsection*{\texorpdfstring{Separating $S_{\textnormal{NS}}(\mathcal{N}_G)$ from $S(\mathcal{N}_G)$ and $S_{\textnormal{Q}}(\mathcal{N}_G)$ for two-sender MACs}{NS-Q-cl-2senders}}
In order to obtain a separation between the quantum-assisted sum rate and the
no-signalling assisted sum rate, we consider the Clauser-Horne-Shimony-Holt
(CHSH) game $G_{\text{CHSH}}$~\cite{clauser1969proposed,
cleve2004consequences}. In this game, a referee selects bits $x_1, x_2 \in \{0,
1\}$ uniformly at random, and gives them to Alice and Bob, respectively.
Upon receiving these question bits, Alice answers with the bit $y_1 \in \{0,
1\}$ and Bob with $y_2 \in \{0, 1\}$. Alice and Bob chose their answers without
communicating with each other. They win the game if
\begin{equation*}
    x_1 \land x_2 = y_1 \oplus y_2,
\end{equation*}
where $\land$ and $\oplus$ represent logical AND and bitwise addition modulo 2.
This game has a total of $d = 4$ question pairs.

It is known that the best classical strategy can answer only $3$ out of the $4$
question pairs correctly, i.e., $\omega^{\text{cl}}(G_{\text{CHSH}}) =
3/4$~\cite{cleve2004consequences}.
The optimal quantum strategy achieves a winning probability of
$\omega^{\text{Q}}(G_{\text{CHSH}}) = (1 + 1/\sqrt{2})/2 \approx 85.4
\%$~\cite{cleve2004consequences}. While there is no classical or quantum
strategy that can always win the game, one can construct a no-signalling
distribution, 
\begin{equation*}
    P_{\text{PR}}(y_1, y_2 | x_1, x_2) = \frac{1}{2} \delta_{x_1 \land x_2, y_1 \oplus y_2},
\end{equation*}
usually called the Popescu-Rohrlich (PR) box~\cite{popescu1994quantum}, which
represents a perfect strategy for winning the CHSH game.  Therefore
$\omega^{\text{NS}}(G_{\text{CHSH}}) = 1$.

Using $\omega^{\text{cl}}(G_{\text{CHSH}}) = 3/4$ in
Cor.~\eqref{cor:NG_MAC_sum_capacity_bound} gives
$S(\mathcal{N}_{G_{\text{CHSH}}}) \leq 1.7$ in the classical case.  On the
other hand, using $\omega^{\text{Q}}(G_{\text{CHSH}}) = (1 + 1/\sqrt{2})/2$ in
Thm.~\ref{thm:NG_MAC_correlation_assisted_sum_capacity_bound}, gives an upper
bound, $S_{\text{Q}}(\mathcal{N}_{G_{\text{CHSH}}}) \leq 1.78$ bits, in the
quantum case.  In the case of no-signalling, a perfect strategy is possible and
thus we have $S_{\text{NS}}(\mathcal{N}_{G_{\text{CHSH}}}) = 2$ bits. In this
way, we obtain a separation between the quantum and no-signalling assisted
achievable sum rate.

\subsubsection*{\texorpdfstring{Separating $S_{\textnormal{Q}}(\mathcal{N}_G)$ from $S(\mathcal{N}_G)$ for $N$-sender MACs}{Q-cl-Nsenders}}
We now consider a game $G_{\text{MPP}}$ that we call the multiparty parity game, which was first introduced
by Brassard \textit{et al.}~\cite{brassard2003multi}. In this game, $N$ players
are each handed a bit and they each answer by returning a bit. The players have
a \textit{promise}: the total number of ones in the $N$-bit string handed to
them is even. If this even number is divisible by $4$, then the winning
condition is that the total bit string returned by the players have an even
number of ones. Otherwise, the winning condition is to return a string with an
odd number of ones.

\noindent Formally, we have $\mathcal{X}_i = \{0, 1\}$ and $\mathcal{Y}_i = \{0, 1\}$ for $i \in [N]$. As before, we denote $\mathcal{X} = \mathcal{X}_1 \times \dotsm \times \mathcal{X}_N$ as the set of questions and $\mathcal{Y} = \mathcal{Y}_1 \times \dotsm \times \mathcal{Y}_N$ as the set of answers for the $N$ players. 
The promise,
\begin{equation*}
    \mathcal{P} = \left\{\bm{x} \in \mathcal{X} \mid \sum\nolimits_i x_i = 0\ (\text{mod } 2)\right\}
\end{equation*}
is a subset of $\mathcal{X}$ from which the questions are draw. 
The winning condition for the game can be described by the set
\begin{equation*}
    \mathcal{W}^{\mathcal{P}} = \left\{(\bm{x}, \bm{y}) \in \mathcal{P} \times \mathcal{Y} \mid \sum\nolimits_i y_i - \frac{1}{2} \sum\nolimits_i x_i\ =0\ (\text{mod } 2)\right\}.
\end{equation*}
Brassard \textit{et al.}~\cite{brassard2003multi} demonstrated that classical strategies can win this game with a probability of at most $\omega^{\text{cl}, \mathcal{P}}(G_{\text{MPP}}) = 1/2 + 2^{-\lceil N/2 \rceil}$ when the questions are drawn uniformly from the promise set. In contrast, a perfect quantum strategy is possible~\cite{brassard2003multi}.

Now we consider the following promise-free version of this game. Herein, the question and answer set remain the same, but the winning condition is defined as the set $\mathcal{W} = \mathcal{W}^{\mathcal{P}} \cup (\mathcal{P}^c \times \mathcal{Y})$. That is, the players win automatically if a question from outside the promise set is presented to them. To apply the bound obtained in Thm.~\ref{thm:NG_MAC_correlation_assisted_sum_capacity_bound}, we need to compute the maximum classical winning probability $\omega^{\text{cl}}(G_{\text{MPP}})$ for the promise-free case when questions are drawn uniformly.
To that end, we show how to compute the maximum winning probability $\omega^{\strat}(G)$ when we convert any game with a promise $\mathcal{P}$ to a promise-free game, assuming that the question are drawn uniformly and answers are given using strategies in $\strat$. Since $\mathcal{W}^{\mathcal{P}} \cap (\mathcal{P}^c \times \mathcal{Y}) = \varnothing$, we have
\begin{equation*}
    p_W = \sum_{(x, y) \in W^{\mathcal{P}}} \frac{1}{|\mathcal{X}|} p(y | x) + \sum_{x \in \mathcal{P}^c} p(x)
        = \frac{|\mathcal{P}|}{|\mathcal{X}|} p_W^{\mathcal{P}} + \left(1 - \frac{|\mathcal{P}|}{|\mathcal{X}|}\right).
\end{equation*}
Since the set of strategies $\strat$ chosen by the players have a maximum probability $\omega^{\strat, \mathcal{P}}(G)$ of winning when the questions are drawn uniformly from $\mathcal{P}$, we can infer that
\begin{equation*}
    \omega^{\strat}(G) = \frac{|\mathcal{P}|}{|\mathcal{X}|} \omega^{\strat, \mathcal{P}}(G) + \left(1 - \frac{|\mathcal{P}|}{|\mathcal{X}|}\right).
\end{equation*}

For the multiparty parity game, since half the $N$-bit strings are in $\mathcal{P}$ and the other half are in $\mathcal{P}^c$, we get
\begin{equation}
    \omega^{\text{cl}}(G_{\text{MPP}}) = \frac{3}{4} + 2^{-(\lceil N/2 \rceil + 1)}. \label{eqn:multiparty_parity_promisefree_game_max_winning_prob_uniform}
\end{equation}
Then, by Cor.~\eqref{cor:NG_MAC_sum_capacity_bound}, we find that the sum capacity for the MAC obtained from the multiparty parity game is bounded above as $S(\mathcal{N}_{G_{\text{MPP}}}) \leq \log(d - 1 + 2^{-(1 - \omega^{\text{cl}}(G)) d})$ bits, where $d = 2^N$ and $\omega^{\text{cl}}(G)$ is given by Eq.~\eqref{eqn:multiparty_parity_promisefree_game_max_winning_prob_uniform}. In particular, $S(\mathcal{N}_{G_{\text{MPP}}}) < \log(d)$. In contrast, since a perfect quantum strategy is available, we have $S_Q(\mathcal{N}_{G_{\text{MPP}}}) = \log(d)$, thus giving a separation between the sum capacity and the quantum-assisted sum rate for $N$-sender MACs. For example, when we have $N = 3$ senders, we obtain $S(\mathcal{N}_{G_{\text{MPP}}}) \leq 2.84$ bits. In contrast, $S_{\text{Q}}(\mathcal{N}_{G_{\text{MPP}}}) = 3$ bits in the quantum case.

\subsection{Looseness of convex relaxation of the sum capacity\label{secn:relaxedALcapacity_sumcapacity_separation}}
In the previous section, we looked at separations between the sum rates with assistance from classical, quantum and no-signalling strategies. In this section, we construct a game such that one can obtain an arbitrarily large separation between the sum capacity and the relaxed sum capacity.
Recall that the relaxed sum capacity corresponds to dropping the product distribution constraint in the maximization problem:
\begin{equation*}
    C(\mathcal{N}_G) = \max_{\overline{p}(\overline{\bm{xy}})} I(\overline{X}_1, \overline{Y}_1, \dotsc, \overline{X}_N, \overline{Y}_N; Z),
\end{equation*}
where $\overline{X}_1, \overline{Y}_1, \dotsc, \overline{X}_N, \overline{Y}_N$ are random variables describing the input to the \ngmac $\mathcal{N}_G$, while $Z$ is the random variable describing the output. We maximize over all possible input probability distributions $\overline{p}$, so that the resulting quantity is the capacity of $\mathcal{N}_G$ when we think of it as a single-input single-output channel. As noted in Sec.~\ref{secn:NG_MAC_strategy_assistance}, we have $S(\mathcal{N}_G) \leq S_{\corr}(\mathcal{N}_G) \leq C(\mathcal{N}_G)$ for any set of correlations $\corr$. Indeed, since we maximize over all probability distributions over the input to $\mathcal{N}_G$, we can write the relaxed sum capacity as
\begin{equation*}
    C(\mathcal{N}_G) = \max_{\overline{p}_{\bm{Y} | \bm{X}} \in \strat_{\text{all}}} \max_{\overline{\pi} \in \Delta_d} I(\overline{X}_1, \overline{Y}_1, \dotsc, \overline{X}_N, \overline{Y}_N; Z)
\end{equation*}
where $\strat_{\text{all}}$ denotes the set of all possible strategies that the players can use to play the game. In particular, this amounts to allowing the players to communicate after the questions have been handed over to them.

To analyze $C(\mathcal{N}_G)$, we study some properties of $\strat_{\text{all}}$. It can be verified that $\strat_{\text{all}}$ is a convex set. The extreme points of this set correspond to deterministic strategies $f\colon \mathcal{X} \to \mathcal{Y}$ that allow for communication between the players (see Prop.~\ref{prop:deterministic_strategy_extreme_point} in App.~\ref{app:strategy_nonlocal_games}). This implies that the maximum winning probability of the game $\omega^{\strat_{\text{all}}}(G)$, when the questions are drawn uniformly and answers are obtained using the strategies in $\strat_{\text{all}}$, is always achieved by a deterministic strategy of the form mentioned above. We now give an explicit description of a deterministic strategy (not necessarily unique) achieving the maximum winning probability.

The best deterministic strategy $f^{(D)}\colon \mathcal{X} \to \mathcal{Y}$ can be written as
\begin{equation}
    f^{(D)}(\bm{x}) = \begin{cases} \bm{y} &\exists \bm{y} \in \mathcal{Y} \text{ such that } (\bm{x}, \bm{y}) \in \mathcal{W} \\
                                    \bm{y}_{\text{o}} &\text{ otherwise, } 
                      \end{cases} \label{eqn:best_deterministic_strategy_relaxed_AL_region}
\end{equation}
where $\bm{y}_{\text{o}} \in \mathcal{Y}$ is an arbitrary element chosen beforehand. We note that for each $\bm{x} \in \mathcal{X}$, some element $\bm{y} \in \mathcal{Y}$ satisfying $(\bm{x}, \bm{y}) \in \mathcal{W}$ is chosen \textit{apriori} (if it exists), so that the function is well-defined, though not necessarily unique. In other words, $f^{(D)}$ gives the correct answer if a correct answer for the given question exists, and if not, it gives an arbitrary answer that is necessarily incorrect. It can, therefore, be inferred that the maximum winning probability can be written as
\begin{equation}
 \omega^{\strat_{\text{all}}}(G) = \frac{\left|\{\bm{x} \in \mathcal{X} \mid \exists \bm{y} \in \mathcal{Y} \text{ such that } (\bm{x}, \bm{y}) \in \mathcal{W}\}\right|}{d}.
                                                                                                                                                                \label{eqn:maximum_winning_prob_uniform_relaxed_AL_region}
\end{equation}
This is the best that one can do given any nonlocal game $G$. Note also that $\omega^{\strat_{\text{all}}}(G)$ can be directly computed from the winning condition $\mathcal{W}$.

The above observation directly leads to upper and lower bounds on $C(\mathcal{N}_G)$. The upper bound is obtained by using $\omega^{\strat_{\text{all}}}(G)$ from Eq.~\eqref{eqn:maximum_winning_prob_uniform_relaxed_AL_region} in Thm.~\ref{thm:NG_MAC_correlation_assisted_sum_capacity_bound}. Here, we implicitly use the fact that our upper bound is valid even when the questions are drawn arbitrarily. Let $\bm{w}^{(D)} \in \{0, 1\}^d$ be the winning vector corresponding to the best deterministic strategy $f^{(D)}$ given in Eq.~\eqref{eqn:best_deterministic_strategy_relaxed_AL_region}. Note that we maximize over all distributions over the questions when computing $C(\mathcal{N}_G)$. Therefore, using $\bm{w}^{(D)}$ in Prop.~\ref{prop:deterministic_strategy_maximum}, we obtain a lower bound on $C(\mathcal{N}_G)$. In particular, if there is at least one correct answer for every question, then $f^{(D)}$ is a perfect strategy and $C(\mathcal{N}_G) = \ln(d)$.

Now, we obtain a separation between $C(\mathcal{N}_G)$ and $S(\mathcal{N}_G)$. To that end, we construct a game called the signalling game.
\subsubsection{\texorpdfstring{Signalling game $G_{\textnormal{s}}$ and separation of $C(\mathcal{N}_{G_{\textnormal{s}}})$ from $S(\mathcal{N}_{G_{\textnormal{s}}})$ and $S_{\textnormal{NS}}(\mathcal{N}_{G_{\textnormal{s}}})$}{Separation of sum capacities using the signalling game}}
Consider a game where Alice \& Bob are each given a question from some set of questions. They win the game if they can correctly guess the question handed over to the other person. 
Since the game can be won if Alice \& Bob ``signal" their question to each other, we call this the signalling game $G_{\text{s}}$.

Formally, we consider question sets $\mathcal{X}_1$, $\mathcal{X}_2$ and answer sets $\mathcal{Y}_1 = \mathcal{X}_2$ and $\mathcal{Y}_2 = \mathcal{X}_1$, and the winning condition is defined as
\begin{equation*}
    \mathcal{W} = \{(x_1, x_2, y_1, y_2) \in (\mathcal{X}_1 \times \mathcal{X}_2) \times (\mathcal{Y}_1 \times \mathcal{Y}_2) \mid y_1 = x_2,\ y_2 = x_1\}.
\end{equation*}
Note that there is exactly one correct answer for each question pair $(x_1, x_2) \in \mathcal{X}_1 \times \mathcal{X}_2$.

Since we bound the sum capacity of the \ngmac $\mathcal{N}_{G_{\text{s}}}$ obtained from the signaling game $G_{\text{s}}$ using the maximum winning probability, we analyze the winning strategies for this game. To that end, consider some set of strategies $\strat$ that Alice and Bob use to play the game. For simplicity, we assume that this is a compact set (thinking of strategies as vectors as in Prop.~\ref{prop:deterministic_strategy_extreme_point}), which holds, for example, when $\strat$ is the set of no-signalling strategies (see Prop.~\ref{prop:NS_strategies_convex_polytope}). Suppose that $\omega^{\strat}(G) = 1$ corresponding to this set of strategies, and let $p^*_{\bm{Y} | \bm{X}}$ be the strategy that achieves this maximum winning probability. When the questions are drawn uniformly, the winning probability for strategy $p^*_{\bm{Y} | \bm{X}}$ is given by $p^*_W = (\sum_i w^*_i)/d$, where $\bm{w}^* = (w^*_1,\dots,w^*_d)$ is the winning vector corresponding to $p^*_{\bm{Y} | \bm{X}}$, as defined in Eq.~\eqref{eqn:strategy_coeff}. Since $p^*_W = \omega^{\strat}(G) = 1$ by assumption and $\bm{w}^* \in [0, 1]^d$, we must have $w^*_i = 1$ for all $i \in [d]$.

For convenience, denote $\mathcal{X} = \mathcal{X}_1 \times \mathcal{X}_2$ as the set of questions. Using Def.~\ref{eqn:strategy_coeff}, we can write the winning vector $\bm{w}^*$ as
\begin{equation*}
    w^*_{\bm{x}} = \sum_{\bm{y}\colon \bm{xy} \in \mathcal{W}} p^*_{\bm{Y} | \bm{X}}(\bm{y} | \bm{x})
\end{equation*}
where we label the components of $\bm{w}^*$ using the questions $\bm{x} \in \mathcal{X}$. Since $w^*_{\bm{x}} = 1$ for all $\bm{x} \in \mathcal{X}$ and because there is exactly one correct answer for each question $\bm{x}$, we can infer that $p^*_{\bm{Y} | \bm{X}}$ must be a deterministic strategy. Written explicitly, we have
\begin{equation*}
    p^*_{\bm{Y} | \bm{X}}(y_1, y_2 | x_1, x_2) = \delta_{y_1, x_2} \delta_{y_2, x_1}.
\end{equation*}

\noindent Note that $p^*_{\bm{Y} | \bm{X}}$ cannot be a no-signalling strategy. Indeed,
\begin{equation}
    p^*_{\bm{Y} | \bm{X}}(y_1 | x_1, x_2) = \sum_{y_2 \in \mathcal{Y}_2} p^*_{\bm{Y} | \bm{X}}(y_1, y_2 | x_1, x_2) = \delta_{y_1, x_2}. \label{eqn:perfect_strategy_signalling_game}
\end{equation}
This cannot satisfy the no-signalling condition $p^*_{\bm{Y} | \bm{X}}(y_1 | x_1, x_2) = p^*_{\bm{Y} | \bm{X}}(y_1 | x_1)$ given in Eq.~\eqref{eqn:no_signalling_strategy} because the RHS in Eq.~\eqref{eqn:perfect_strategy_signalling_game} depends on $x_2$. In other words, we cannot have $w_{\bm{x}} = 1$ for \textit{any} question $\bm{x} \in \mathcal{X}$ using a no-signalling strategy. In particular, the perfect strategy is not no-signalling, and therefore, $\omega^{\text{NS}}(G_{\text{s}}) < 1$ for no-signalling strategies. Subsequently, we also have $\omega^{\text{Q}}(G_{\text{s}}) < 1$ and $\omega^{\text{cl}}(G_{\text{s}}) < 1$, because the set of classical and quantum strategies are contained in the set of no-signalling strategies. It then follows from Thm.~\ref{thm:NG_MAC_correlation_assisted_sum_capacity_bound} and Cor.~\ref{cor:NG_MAC_sum_capacity_bound} that each of $S(\mathcal{N}_{G_{\text{s}}}), S_{\text{Q}}(\mathcal{N}_{G_{\text{s}}}), S_{\text{NS}}(\mathcal{N}_{G_{\text{s}}})$ is strictly less than $\ln(d)$.
On the other hand, since a perfect strategy is possible allowing communication between Alice \& Bob, we have $C(\mathcal{N}_{G_{\text{s}}}) = \ln(d)$. Therefore, we have obtained a separation between $C(\mathcal{N}_{G_{\text{s}}})$ and $S(\mathcal{N}_{G_{\text{s}}}), S_{\text{Q}}(\mathcal{N}_{G_{\text{s}}}), S_{\text{NS}}(\mathcal{N}_{G_{\text{s}}})$.

Below, we argue that this separation becomes arbitrarily large as the number of questions increases. To that end, we compute $\omega^{\text{cl}}(G_{\text{s}})$. Since the maximum winning probability obtained using classical strategies when questions are drawn uniformly is achieved by a deterministic strategy, it is sufficient to restrict our attention to classical deterministic strategies chosen by Alice \& Bob. Recall that a classical deterministic strategy corresponds to two functions $f_1\colon \mathcal{X}_1 \to \mathcal{Y}_1$ and $f_2\colon \mathcal{X}_2 \to \mathcal{Y}_2$ chosen by Alice \& Bob, respectively. This translates to the probability distribution $p^{(D)}_{\bm{Y} | \bm{X}}(y_1, y_2 | x_1, x_2) = \delta_{y_1, f_1(x)} \delta_{y_2, f_2(x)}$. Then, the winning probability using a classical deterministic strategy when the questions are drawn uniformly is given by
\begin{equation}
    p_W^{(D)} = \frac{1}{d} \sum_{x_1 \in \mathcal{X}_1, x_2 \in \mathcal{X}_2} \delta_{x_2, f_1(x_1)} \delta_{x_1, f_2(x_2)} \label{eqn:uniform_winning_prob_nosignalling_game_classical_strategy}
\end{equation}
where we used the fact that the signalling game has only one correct answer $(x_2, x_1)$ corresponding to each question $(x_1, x_2)$. It can be seen from Eq.~\eqref{eqn:uniform_winning_prob_nosignalling_game_classical_strategy} that, for achieving maximum winning probability, the function $f_1$ must be able to invert the action of the function $f_2$ or vice-versa.

If $|\mathcal{X}_2| \leq |\mathcal{X}_1|$, then the set $f_2(\mathcal{X}_2)$ can cover at most $|\mathcal{X}_2|$ elements of $\mathcal{X}_1$. Subsequently, $\delta_{x_1, f_2(x_2)} \neq 0$ for at most $|\mathcal{X}_2|$ elements of $\mathcal{X}_1$. We can then infer from Eq.~\eqref{eqn:uniform_winning_prob_nosignalling_game_classical_strategy} that
\begin{equation*}
    p_W^{(D)} \leq \frac{|\mathcal{X}_2|}{d} = \frac{1}{|\mathcal{X}_1|}
\end{equation*}
since $d = |\mathcal{X}_1| |\mathcal{X}_2|$. Using a similar reasoning when $|\mathcal{X}_1| \leq |\mathcal{X}_2|$, 
\begin{equation*}
    \omega^{\text{cl}}(G_{\text{s}}) =  \frac{1}{\max(|\mathcal{X}_1|, |\mathcal{X}_2|)}.
\end{equation*}
In particular, we have $d \to \infty$ if either of $|\mathcal{X}_1|$ or $|\mathcal{X}_2|$ diverges, while the winning probability $\omega^{\text{cl}}(G) \to 0$. Subsequently, $C(\mathcal{N}_{G_{\text{s}}}) \to \infty$ but $S(\mathcal{N}_{G_{\text{s}}}) \to 0$. Therefore, we get an arbitrarily large separation between $C(\mathcal{N}_{G_{\text{s}}})$ and $S(\mathcal{N}_{G_{\text{s}}})$.

In fact, we verify through numerical simulations that the situation is equally bad for the no-signalling assisted sum rate. To that end, we compute the maximum winning probability $\omega^{\text{NS}}(G_{\text{s}})$ numerically. We show in Prop.~\ref{prop:NS_strategies_convex_polytope} that the set of no-signalling strategies for $N$-player games is a compact and convex set (specifically, a convex polytope). Therefore, computing $\omega^{\text{NS}}(G_{\text{s}})$ amounts to solving a linear program (this fact is well-known for $2$-player games~\cite{toner2009monogamy}). For $2 \leq |\mathcal{X}_1|, |\mathcal{X}_2| \leq 10$, we verify that the numerically computed value for $\omega^{\text{NS}}(G_{\text{s}})$ matches $\omega^{\text{cl}}(\mathcal{N}_{G_{\text{s}}}) = 1/\max(|\mathcal{X}_1|, |\mathcal{X}_2|)$. Thus, we expect $\omega^{\text{cl}}(\mathcal{N}_{G_{\text{s}}}) = \omega^{\text{Q}}(\mathcal{N}_{G_{\text{s}}}) = \omega^{\text{NS}}(\mathcal{N}_{G_{\text{s}}})$ for the signalling game. This would imply that $S_{\text{Q}}(\mathcal{N}_{G_{\text{s}}}), S_{\text{NS}}(\mathcal{N}_{G_{\text{s}}}) \to 0$ as $d \to \infty$ but $C(\mathcal{N}_{G_{\text{s}}}) \to \infty$. In other words, even with quantum or no-signalling assistance, the sum rate is far lower than the bound given by the relaxed sum capacity.

This example highlights the importance of finding better methods to upper bound the sum capacity of MACs. In the next part, we take a step in this direction by defining and studying a class of global optimization problems with relevance in information theory. This class of optimization problems is motivated from the non-convex problem encountered in sum capacity computation. In the third part of this study, we will show how to apply these algorithms for computing the sum capacity of arbitrary two-sender MACs.

\section{Optimization of Lipschitz-like functions\label{secn:Lipschitzlike_optimization}}
\subsection{Lipschitz-like functions}\label{secn:lipschitz-like-functions}
The main object of our study is a Lipschitz-like function. Such functions are a generalization of Lipschitz-continuous and H\"older continuous functions. Recall that a function $f\colon \mathcal{D} \to \mathcal{E}$ between two subsets of normed vector spaces is said to be Lipschitz continuous if for all $x, x' \in \mathcal{D}$, we have $\norm{f(x) - f(x')}_{\mathcal{E}} \leq L \norm{x - x'}_{\mathcal{D}}$ for some constant $L > 0$. The function $f$ is said to be H\"older continuous if $\norm{f(x) - f(x')}_{\mathcal{E}} \leq L \norm{x - x'}_{\mathcal{D}}^\gamma$ for some constants $L, \gamma > 0$~\cite{deklerk2008complexity}. H\"older continuity is a more general notion than Lipschitz continuity since $\gamma = 1$ gives the definition of Lipschitz continuity. Even so, this definition is not general enough to capture the continuity properties of entropic quantities. Shannon and von Neumann entropies, for example, satisfy a different continuity bound. Specifically, if $p, q \in \Delta_d$ are discrete probability distributions, then the Shannon entropy satisfies
\begin{equation}
    |H(p) - H(q)| \leq \frac{1}{2} \log(d - 1) \norm{p - q}_1 + h\left(\frac{1}{2} \norm{p - q}_1\right), \label{eqn:entropy_continuity_bound}
\end{equation}
where $h(x) = -x\log(x) - (1 - x)\log(1 - x)$ is the binary entropy function~\cite{zhang2007estimating}. Similarly, von Neumann entropy satisfies the Fannes-Audenaert inequality~\cite{audenaert2007sharp}. To encapsulate such behaviour of entropic quantities, we define Lipschitz-like functions as follows.

\begin{definition}[Lipschitz-like function]
Let $\beta\colon \mathbb{R}_+ \to \mathbb{R}$ be a non-negative, continuous, monotonically increasing function such that $\beta(0) = 0$. Let $\mathcal{D}$ and $\mathcal{E}$ be subsets of a normed vector space. Then a function $f\colon \mathcal{D} \to \mathbb{E}$ is said to be Lipschitz-like or $\beta$-Lipschitz-like if it satisfies
\begin{equation}
    \norm{f(x) - f(x')}_{\mathcal{E}} \leq \beta(\norm{x - x'}_{\mathcal{D}}) \quad \forall\ x, x' \in \mathcal{D}. \label{eqn:Lipschitz_like_function}
\end{equation}
\end{definition}

Some remarks about this definition are in order. The definition of Lipschitz-like functions can be generalized to metric spaces in a straightforward manner. The reason we require $\beta$ to be monotonically increasing is because it makes Lipschitz-like functions behave similar to Lipschitz continuous functions in the sense that the bound on $\norm{f(x) - f(x')}_{\mathcal{E}}$ tightens or loosens with the value of $\norm{x - x'}_{\mathcal{D}}$. Moreover, this assumption helps with optimization of Lipschitz-like functions. Similarly, we require continuity of $\beta$ for simplicity, but this can be relaxed to right-continuity at $0$. Since $\beta$ is (right-)continuous at $0$, it follows from Eq.~\eqref{eqn:Lipschitz_like_function} that $f$ is a continuous function. If $\beta(x) = L x$ for some $L > 0$, then $f$ is a Lipschitz continuous function with Lipschitz constant $L$, and if $\beta(x) = L x^\gamma$ for $L, \gamma > 0$, then $f$ is a H\"older continuous function with constant $L$. Lipschitz-like functions are therefore a generalization of Lipschitz and H\"older continuous functions.

With a slight modification of the right-hand side of Eq.~\eqref{eqn:entropy_continuity_bound}, we can show that entropy is a Lipschitz-like function. To that end, we define the modified binary entropy as follows.
\begin{equation}
    \overline{h}(x) = \begin{cases} -x\log(x) - (1 - x)\log(1 - x) \hspace{0.1cm} \text{ if } x \leq \frac{1}{2} \\
                                    \log(2)                        \hspace{4.1cm} \text{ if } x \geq \frac{1}{2}
                      \end{cases} \label{eqn:modified_binary_entropy}
\end{equation}
Observe that $\overline{h}$ is a non-negative, continuous, and monotonically increasing function that satisfies $\overline{h}(0) = 0$, and furthermore, we have $h(x) \leq \overline{h}(x)$ for all $x \in [0, 1]$. Thus, defining $\beta_H(x) = \log(d - 1) x/2 + \overline{h}(x/2)$, we obtain
\begin{equation*}
    |H(p) - H(q)| \leq \beta_H\left(\norm{p - q}_1\right).
\end{equation*}
From this, we can conclude that Shannon entropy (and similarly, von Neumann entropy) is a Lipschitz-like function.

Moreover, linear combinations and compositions of Lipschitz-like functions is again a Lipschitz-like function. We summarize this observation in the following result.
\begin{proposition}
    \begin{enumerate}
        \item If $f_1\colon \mathcal{D} \to \mathcal{E}$ and $f_2\colon \mathcal{D} \to \mathcal{E}$ are $\beta_1$-Lipschitz-like and $\beta_2$-Lipschitz-like functions respectively, then the linear combination $\alpha_1 f_1 + \alpha_2 f_2$ is a $(|\alpha_1| \beta_1 + |\alpha_2| \beta_2)$-Lipschitz-like function, where $\alpha_1, \alpha_2$ are scalars.
        \item If $f_1\colon \mathcal{D} \to \mathcal{E}$ and $f_2\colon \mathcal{E} \to \mathcal{F}$ are $\beta_1$-Lipschitz-like and $\beta_2$-Lipschitz-like functions respectively, then the composition $f_2 \circ f_1$ is a $(\beta_2 \circ \beta_1)$-Lipschitz-like function.
    \end{enumerate}
\end{proposition}
\begin{proof}
    1. The function $g = \alpha_1 f_1 + \alpha_2 f_2$ satisfies $\norm{g(x) - g(y)}_{\mathcal{E}} \leq |\alpha_1| \norm{f_1(x) - f_1(y)}_{\mathcal{E}} + |\alpha_2| \norm{f_2(x) - f_2(y)}_{\mathcal{E}}$ for $x, y \in \mathcal{D}$ by triangle inequality. Then, using the fact that $f_1, f_2$ are Lipschitz-like, we obtain $\norm{g(x) - g(y)}_{\mathcal{E}} \leq |\alpha_1| \beta_1(\norm{x - y}_{\mathcal{D}}) + |\alpha_2| \beta_2(\norm{x - y}_{\mathcal{D}})$. Since $\beta = |\alpha_1| \beta_1 + |\alpha_2| \beta_2$ is non-negative, continuous and monotonically increasing with $\beta(0) = 0$, we can conclude that $g$ is $\beta$-Lipschitz-like.

    2. The function $g = f_2 \circ f_1$ satisfies $\norm{g(x) - g(y)}_{\mathcal{F}} \leq \beta_2(\norm{f_1(x) - f_1(y)}_{\mathcal{E}})$ for $x, y \in \mathcal{D}$ by using the fact that $f_2$ is $\beta_2$-Lipschitz-like. Then, since $f_1$ is $\beta_1$-Lipschitz-like and $\beta_2$ is monotonically increasing, we obtain $\norm{g(x) - g(y)}_{\mathcal{F}} \leq \beta_2(\beta_1(\norm{x - y}_{\mathcal{D}}))$. Since $\beta = \beta_2 \circ \beta_1$ is non-negative, continuous and monotonically increasing with $\beta(0) = 0$, we can conclude that $g$ is $\beta$-Lipschitz-like.
\end{proof}

In light of the above result, we can conclude that several entropic quantities derived from Shannon entropy and von Neumann entropy are Lipschitz-like functions. This lends support to our claim that the techniques developed in this section can potentially be useful for non-convex optimization problems in information theory.

For the purposes of optimization, we take the co-domain to be the real line, i.e., $\mathcal{E} = \mathbb{R}$ with the usual norm. Motivated by applications in information theory, we will mainly be focusing on the case where the domain $\mathcal{D} = \Delta_d$ is the standard simplex in $\mathbb{R}^d$ and $\norm{\cdot}_{\mathcal{D}} = \norm{\cdot}_1$. Nevertheless, our techniques work more generally with any norm (for example, using equivalence of norms in finite dimensions).

In this section, we will develop algorithms for optimizing Lipschitz-like functions over a closed interval and a standard simplex. We also discuss techniques for optimizing Lipschitz-like functions over arbitrary compact and convex domains in App.~\ref{app:Lipschitzlike_optimization_compact_convex_domain}. We approach this problem by showing how to extend a Lipschitz-like function from a compact and convex domain to all of Euclidean space such that the global optimum is not affected, which may be of independent interest. Some of these algorithms we discuss are generalizations of existing algorithms for optimization of Lipschitz continuous functions. For this reason, we present a brief overview of some existing algorithms and our approach to generalizing them. In the next section, we focus on outlining the high-level ideas without delving into the technical details. Subsequently, we will give detailed results on the optimization algorithms, along with some numerical examples.

\subsection{Overview of the optimization algorithms}\label{secn:overview-optimization-algorithms}
All the algorithms we propose in this study are designed to optimize any
$\beta$-Lipschitz-like function~$f$. The function $\beta$ is assumed to be
known beforehand, but the function $f$ is unknown and we can only query it at a
specified point. The algorithms can then use the knowledge of the domain, the
function $\beta$, the queried points and corresponding values of the objective
function $f$ to approximate the maximum of $f$ to an additive precision $\epsilon > 0$.
Such a setting is commonly used to study the performance of optimization
algorithms~\cite{deklerk2008complexity}.

In order to make concrete statements about the complexity of optimization, we will assume that $\beta$ does not explicitly depend on the dimension. Our algorithms and convergence analysis are valid even if this assumption does not hold.
By an efficient algorithm for optimization, we mean an algorithm that computes the optimum of $f$ to a given additive precision $\epsilon > 0$ in time polynomial in the dimension and inverse precision $1/\epsilon$.
We sometimes informally use the phrase ``practically efficient" and variations thereof to mean that the algorithm runs reasonably fast in practice, e.g., to exclude situations where the scaling of runtime with dimension is too large (for example, $O(d^{10})$).

Before discussing technical details of the algorithms presented in this study,
we give a high-level overview of the main ideas. We begin by presenting an
algorithm that can optimize any $\beta$-Lipschitz-like function $f$ when the
domain $\mathcal{D} = [a, b]$ is a closed interval. Our
algorithm is a generalization of the Piyavskii-Shubert
algorithm~\cite{piyavskii1972algorithm, shubert1972sequential}, and it focuses
on constructing successively better upper-bounding functions by using the
Lipschitz-like property of the objective function. At each iteration, the
maximum of the upper-bounding function is computed, which is an easier problem
because we only need to maximize the function $\beta$ that is known to be
monotonically increasing. The computed maximum of the upper-bounding function
at each iteration generates a sequence of points that partitions the interval.
When the distance between any two of these points becomes sufficiently small,
one can show that the maximum of the upper-bounding function is close to the
maximum of the original objective function. Our algorithm is guaranteed to
converge to the optimal solution within a precision of $\epsilon > 0$ in
$\lceil(b - a)/\delta\rceil$ time steps in the worst-case, where $\delta > 0$
is the largest number satisfying $\beta(\delta) \leq \epsilon/2$ (see
Prop.~\ref{prop:Lipshitz_like_maximization_analysis}). We refer to this
algorithm as modified Piyavskii-Shubert algorithm.
We remark that several extensions of the Piyavskii-Shubert algorithm have been presented in the literature (see, for
example, Ref.~\cite{lera2002global, vanderbei1999extension,
mladineo1986algorithm, rahal2008new}).
It would be interesting to undertake a more detailed study comparing such algorithms with our method in the future.

In higher dimensions, we focus on the case where the objective $f$ is a $\beta$-Lipschitz-like
function over the standard simplex $\mathcal{D} = \Delta_d$. For this case, we present two algorithms for finding the global optimum of $f$.
For the first algorithm, we resort to a straightforward grid
search. Using the results of Ref.~\cite{deklerk2008simplexcomplexity}, one can
show that a grid of size $O(d^{\lceil 1/\delta^2 \rceil})$ suffices to converge to a
specified precision of $\epsilon > 0$, where $\delta$ is the largest number
satisfying $\beta(\delta) \leq \epsilon/2$ (see
Prop.~\ref{prop:simplex_grid_search_Lipschitzlike_optimization}).\footnote{Note the different scaling in $\delta$ in Prop.~\ref{prop:simplex_grid_search_Lipschitzlike_optimization} compared to that of $O(d^{\lceil 1/\sqrt{\delta} \rceil})$ mentioned in Ref.~\cite{deklerk2008simplexcomplexity}, which we believe to be erroneous.}
These results stand to demonstrate that one can, in principle, compute the
maximum of $f$ in polynomial time for a fixed precision.
Moreover, a simple observation about ordering the elements of the grid allows
for efficient construction of the grid, along with the possibility of
parallelizing the grid search.
Despite the possibility of polynomial
complexity (in dimension) and numerical improvements, grid search is still too inefficient to
be of practical use except for very small dimensions.

Another strategy we propose is to construct a ``dense" Lipschitz continuous
curve that gets close to each point of $\mathcal{D}$ to within some specified
distance. Such a strategy was adopted by Ref.~\cite{ziadi2001global} for
optimizing Lipschitz continuous functions over a hypercube, and is referred to
as Alienor method in the literature.
When $\mathcal{D}$ is the standard simplex in $d$ dimensions, we give a time
and memory efficient algorithm to construct such a curve. This allows us to
reduce the $d$-dimensional problem of optimizing $f$ over $\mathcal{D}$ to the
one-dimensional problem of optimizing it over an interval using the generated
curve. If $\alpha > 0$ is the largest number satisfying $\beta(\alpha) \leq
\epsilon/2$, this method takes $O(\alpha^{1 - d}/d)$ iterations in the worst
case for large dimensions. While this is much worse than a grid search in large
dimensions, in small dimensions this takes fewer iterations to
converge than grid search when the tolerance $\epsilon$ is small.
The dense curve algorithm also has the advantage that we can find an upper bound on the maximum by running the algorithm for a fixed number of time steps.
Furthermore, similar to grid search, the dense curve algorithm can be parallelized.
Nevertheless, we remark that both the grid search and dense curve algorithms are impractical for even moderately large dimensions. We detail these methods here in the hope that these engender the development of more practical algorithms for optimizing Lipschitz-like functions over the standard simplex in higher dimensions.

We also study the more general case when $\mathcal{D}$ is a compact and convex set in $\mathbb{R}^d$. To handle optimization in this general case, we seek to reduce it to a case that can be solved using known techniques. To that end, we show how to extend the objective function $f$ from the domain $\mathcal{D}$ to the full Euclidean space $\mathbb{R}^d$ while retaining the Lipschitz-like property when $\beta$ is itself Lipschitz-like (see Def.~\ref{def:Lipschitzlike_function_extension} and Prop.~\ref{prop:Lipschitzlike_function_extension}). The extended function $\overline{f}$ has the property that its optimum coincides with the optimum of $f$ when optimized over any set containing $\mathcal{D}$. Since we have the freedom to choose which set to optimize $\overline{f}$ over, there are different algorithms one can potentially use for this optimization. In this study, we find the maximum value of $f$ by optimizing its extension $\overline{f}$ over a hypercube containing $\mathcal{D}$. To perform this optimization, we resort to using dense curves because the convergence analysis is very similar to the case of the standard simplex $\mathcal{D} = \Delta_d$. The effective problem, as before, is one-dimensional and can be solved using the modified Piyavskii-Shubert algorithm. In general, this algorithm needs exponentially many iterations (with the dimension) to find the optimal solution to within a specified precision of $\epsilon > 0$ (see Prop.~\ref{prop:Lipschitzlike_optimization_dense_curve_compact_convex_set} for a precise statement). This exponential scaling with the dimension stems from the fact that we do not use any structure of $\mathcal{D}$ to construct the dense curve and instead rely on a curve generated for a hypercube. The $O(1/\epsilon^d)$ complexity cannot be improved in general without additional assumptions on $\mathcal{D}$ or the class of functions we optimize (see Prop.~\ref{prop:Lipschitzlike_optimization_dense_curve_compact_convex_set}). However, it might be possible to improve the other factors that scale exponentially in the algorithm as noted in App.~\ref{app:Lipschitzlike_optimization_compact_convex_domain}. An advantage of using this algorithm is that one can specify a fixed number of iterations to obtain an upper bound on the maximum.
Because optimizing $f$ over an arbitrary compact and convex domain is not directly relevant to our study of MACs, we relegate this discussion to App.~\ref{app:Lipschitzlike_optimization_compact_convex_domain}.

We now dive into details of the proposed algorithms.

\subsection{Optimizing Lipschitz-like functions over an interval using modified Piyavskii-Shubert algorithm\label{secn:Lipschitzlike_optimization_1D}}
We begin our study by presenting an algorithm for computing the maximum of any $\beta$-Lipschitz-like function $f$ over a closed interval $\mathcal{D} = [a, b]$. A pseudocode for this algorithm is given in Alg.~\ref{alg:maximize_Lipschitz_like_function_1D}.
\begin{algorithm}[ht]
    \begin{algorithmic}[1]
        \Function{maximize\_Lipshcitz-like\_function\_1D}{$\epsilon$}
            \State Initialize $q^{(0)} = a$
            \State Set $F_0(q) = f(q^{(0)}) + \beta(|q - q^{(0)}|)$ for $q \in [a, b]$
            \State Set $F \gets F_0$ and $q^* \gets b$
            \State Set $q^{(1)} \gets q^*$, $k \gets 1$
            \While{$F(q^*) - f(q^*) > \epsilon$}
                \State Sort $\{q^{(0)}, \dotsc, q^{(k)}\}$ from smallest to largest and relabel the points in ascending order.
                \State Define $F_i(q) = f(q^{(i)}) + \beta(|q - q^{(i)}|)$ for $0 \leq i \leq k$ \label{algln:bounding_function}
                \For{$i = 0, \dotsc, k - 1$}
                    \State Set $g_i(q) = F_i(q) - F_{i + 1}(q)$ \label{algln:root_finding1}
                    \State Find $\overline{q}^{(i)} \in [q^{(i)}, q^{(i + 1)}]$ such that $g_i(\overline{q}^{(i)}) = 0$ using any root finding method. \label{algln:root_finding2}
                \EndFor
                \State Pick an index $m \in \argmax_{0 \leq i \leq k - 1} F_i(\overline{q}^{(i)})$ and set $q^* = \overline{q}^{(m)}$.
                \State Update $F \gets F_m$
                \State Set $q^{(k + 1)} \gets q^*$, $k \gets k + 1$
            \EndWhile
            \State \textbf{return} $f(q^*)$
        \EndFunction
    \end{algorithmic}
    \caption{Computing the maximum of a function $f$ satisfying Eq.~\eqref{eqn:Lipschitz_like_function} for $\mathcal{D} = [a, b]$, given $\epsilon > 0$}
    \label{alg:maximize_Lipschitz_like_function_1D}
\end{algorithm}

We refer to the function $F_i$ defined in line~\ref{algln:bounding_function} of Alg.~\ref{alg:maximize_Lipschitz_like_function_1D} as a bounding function, since $f(q) \leq F_i(q)$ for all $q \in [a, b]$ and all $i$. Note that $F_i(q^{(i)}) = f(q^{(i)})$ for all $i$. The bounding function $F_i$ depends non-trivially on the argument $q$ only through the function $\beta(|q - q^{(i)}|)$. We compute the optimum of the function $f$ by maximizing these bounding functions, which is an easier problem because $\beta$ is continuous and monotonic. Essentially, the algorithm does the following. Suppose that at the $K$th time step, we have the iterates $q^{(0)}, \dotsm, q^{(K)}$ sorted in the ascending order $a = q^{(0)} \leq q^{(1)} \leq \dotsm \leq q^{(K)} = b$. Then in each of the intervals $[q^{(k)}, q^{(k + 1)}]$ for $k \in \{0, \dotsc, K - 1\}$ we compute the point $\overline{q}^{(k)}$ attaining the maximum of the bounding function $\min\{F_k, F_{k + 1}\}$ (see lines~\ref{algln:root_finding1}, \ref{algln:root_finding2} in Alg.~\ref{alg:maximize_Lipschitz_like_function_1D}). Our next iterate $q^{(K + 1)}$ is chosen to be in $\argmax\{\overline{q}^{(0)}, \dotsc, \overline{q}^{(K - 1)}\}$. In the next iterate, since $F_{K + 1}(q^{(K + 1)}) = f(q^{(K + 1)})$, we have essentially tightened the upper bound on the function $f$. Proceeding this way, one can verify that the algorithm eventually approximates the function $f$ from above well enough. A schematic of this procedure is shown in Fig.~\ref{fig:lipschitzlike_opt_schematic}.

\begin{figure}[ht]
    \begin{center}
        \includegraphics[width=0.85\textwidth]{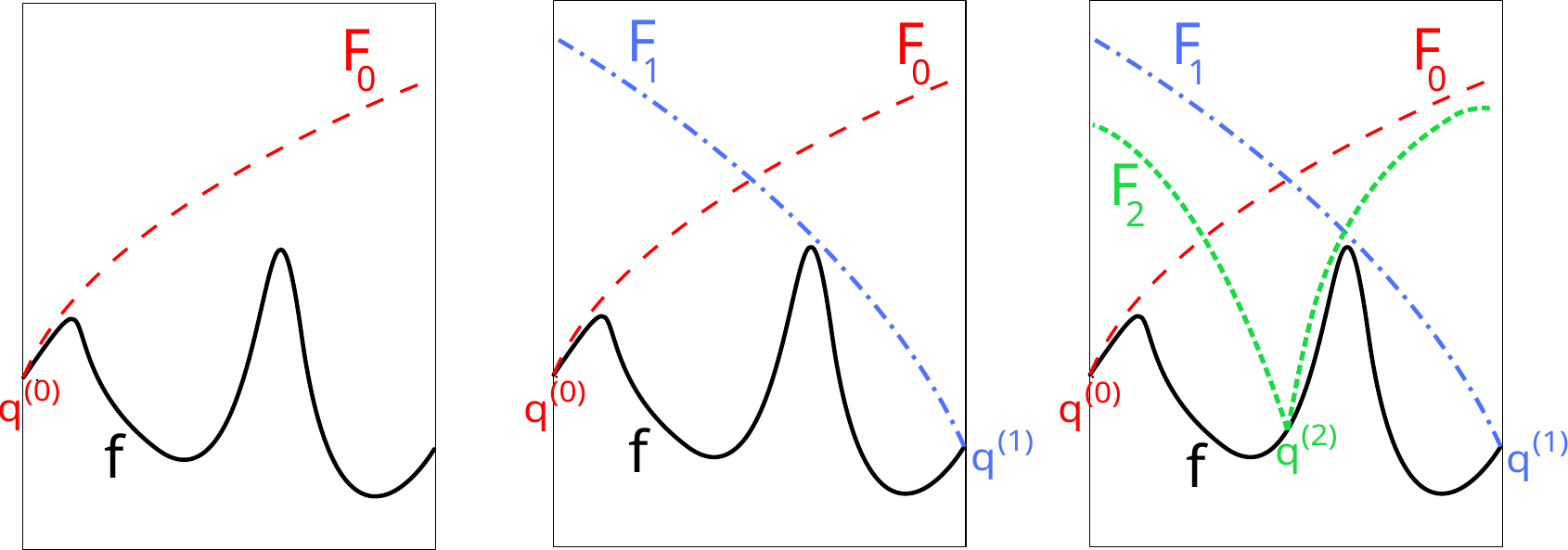}
    \end{center}
    \caption{A schematic of iterations of Alg.~\ref{alg:maximize_Lipschitz_like_function_1D} for optimizing a Lipschitz-like function $f$.}
    \label{fig:lipschitzlike_opt_schematic}
\end{figure}

Indeed, we show in Prop.~\ref{prop:Lipshitz_like_maximization_analysis} that Alg.~\ref{alg:maximize_Lipschitz_like_function_1D} is guaranteed to converge to the global maximum within an error of $\epsilon$. This convergence takes at most $\lceil(b - a)/\delta\rceil$ in the worst-case, where $\delta > 0$ is the largest number satisfying $\beta(\delta) \leq \epsilon/2$. Since we find successively better upper bounds on the objective, Alg.~\ref{alg:maximize_Lipschitz_like_function_1D} can be modified so that it accepts a fixed number of iterations instead of a precision, and outputs the upper bound $F(q^*)$ on the maximum of $f$. This upper bound has an error of at most $F(q^*) - f(q^*)$.

When the domain $\mathcal{D} = \Delta_2$ is the standard simplex in $2$-dimensions, we can parameterize any $x \in \Delta_2$ as $x = (q, 1 - q)$ with $q \in [0, 1]$. Furthermore, we have $\norm{x - y}_1 = 2 |p - q|$ for $x, y \in \Delta_2$ parametrized as $x = (q, 1 - q)$ and $y = (p, 1 - p)$. Thus, taking $a = 0$, $b = 1$, and replacing $|q - q^{(i)}|$ with $2|q - q^{(i)}|$ in Alg.~\ref{alg:maximize_Lipschitz_like_function_1D}, we get an algorithm for optimizing $f$ over $\Delta_2$.

Next, we present algorithms to optimize Lipschitz-like functions in higher dimensions. Our main focus will be on optimizing Lipschitz-like functions over the standard simplex $\Delta_d$ for $d \geq 3$.

\subsection{Optimizing Lipschitz-like function over the standard simplex\label{secn:Lipschitzlike_optimization_simplex}}
We present two algorithms to optimize Lipschitz-like functions over $\mathcal{D} = \Delta_d$. The first algorithm is a simple grid search, whereas the second algorithm uses dense curves to fill $\mathcal{D}$. We describe these algorithms, corresponding convergence guarantees and practical implementation in detail.

\subsubsection{Optimization using grid search\label{secn:Lipschitzlike_optimization_simplex_grid_search}}
Our goal in this section is to optimize a $\beta$-Lipschitz-like function $f$ over the standard simplex $\mathcal{D} = \Delta_d$. Based on the results of Ref.~\cite{deklerk2008simplexcomplexity}, we present a grid search method for finding the maximum of $f$ over $\mathcal{D}$. We begin by defining the integer grid
\begin{equation}
    \mathcal{I}_{d, N} = \left\{n \in \mathbb{N}^d \mid \sum_{i = 1}^d n_i = N\right\}, \label{eqn:integer_grid_simplex}
\end{equation}
where $d \in \mathbb{N}_+$ denotes the dimension and $N \in \mathbb{N}_+$ denotes the size of the integer grid. This grid has
\begin{equation}
    N_{\text{grid}} = \binom{N + d - 1}{d - 1} \label{eqn:simplex_grid_cardinality}
\end{equation}
elements because each element of $\mathcal{I}_{d, N}$ can be obtained by arranging $N$ ones into $d$ coordinates. From $\mathcal{I}_{d, N}$, we obtain the grid
\begin{equation}
    \Delta_{d, N} = \left\{\frac{n}{N} \mid n \in \mathcal{I}_{d,N}\right\} \label{eqn:simplex_grid}
\end{equation}
for the standard simplex $\Delta_d$. Note that $\Delta_{d, N}$ is just a rescaled version of $\mathcal{I}_{d, N}$, i.e., $\Delta_{d, N} = \mathcal{I}_{d, N}/N$.

The authors of Ref.~\cite{deklerk2008simplexcomplexity} propose to search the grid $\Delta_{d, N}$ in order to optimize H\"older continuous functions over the standard simplex. Their results are based on approximations of the function $f$ using Bernstein polynomials, and are in fact general enough to handle the optimization of Lipschitz-like functions (see Prop.~\ref{prop:simplex_grid_search_Lipschitzlike_optimization} for details). Such an approach to compute the maximum of $f$ to a specified precision $\epsilon > 0$ is summarized in Alg.~\ref{alg:maximize_Lipschitz_like_function_simplex_grid_search}.
\begin{algorithm}[H]
    \begin{algorithmic}[1]
        \Function{maximize\_Lipshcitz-like\_function\_simplex}{$d$, $\beta$, $\epsilon$}
            \State Set $N = \lceil 1/\delta^2 \rceil$, where $\delta$ is the largest number satisfying $\beta(\delta) \leq \epsilon/2$
            \State Construct the grid $\Delta_{d, N}$
            \State \textbf{return} $f^* = \max\{f(x) \mid x \in \Delta_{d, N}\}$
        \EndFunction
    \end{algorithmic}
    \caption{Computing the maximum of a function $f$ satisfying Eq.~\eqref{eqn:Lipschitz_like_function} for $\mathcal{D} = \Delta_d$, given $\epsilon > 0$}
    \label{alg:maximize_Lipschitz_like_function_simplex_grid_search}
\end{algorithm}

In Prop.~\ref{prop:simplex_grid_search_Lipschitzlike_optimization}, we show that Alg.~\ref{alg:maximize_Lipschitz_like_function_simplex_grid_search} computes the maximum of $f$ to a precision $\epsilon > 0$ in $N_{\text{grid}} = \binom{N + d - 1}{d - 1}$ time steps, where $N = \lceil 1/\delta^2 \rceil$ and $\delta$ is the largest number satisfying $\beta(\delta) \leq \epsilon/2$. For fixed $\epsilon > 0$ and $d \gg N$, this amounts to $O(d^{\lceil 1/\delta^2 \rceil})$ iterations. In other words, we can find the optimum of $f$ in polynomial time for a fixed precision. We note that if $\beta$ achieves the value $\epsilon/2$, we can obtain $\delta$ by solving $\beta(\delta) = \epsilon/2$ using bisection (or any other root finding method) because $\beta$ is continuous and monotonically increasing.

The crucial step in implementing the above algorithm is computing the grid $\Delta_{d, N}$ efficiently. For that purpose, note that we can write any element $x \in \Delta_{d, N}$ as $x = (N - n_{d - 1}, n_{d - 1} - n_{d - 2}, \dotsc, n_2 - n_1, n_1)/N$, where $0 \leq n_i \leq n_{i + 1} \leq N$ for $i \in [d - 2]$ (see Prop.~\ref{prop:simplex_grid} for a proof). Thus, the elements of $\Delta_{d, N}$ can be computed iteratively, with the total number of iterations equalling $N_{\text{grid}}$. This also allows for parallelizing the search over the grid. The exact algorithm we use to query the elements of the grid orders the elements such that every consecutive element is equidistant with respect to $l_1$-norm. This approach is explained in the next section and allows for easy parallelization.

\subsubsection{Optimization using dense curves\label{secn:Lipschitzlike_optimization_simplex_dense_curve}} 
Next, we outline a method to optimize Lipschitz-like functions over $\mathcal{D} = \Delta_n$ by filling $\mathcal{D}$ with an $\alpha$-dense Lipschitz curve. We propose this method as a way to reduce the total number of iterations required for finding the optimum to a small additive error in practice compared to grid search. Such a strategy of using dense curves was outlined in Ref.~\cite{ziadi2001global} to optimize Lipschitz continuous functions over a hypercube.

\begin{definition}[$\alpha$-dense curve]
\label{def:dense_curve}
Given a number $\alpha > 0$, numbers $a, b \in \mathbb{R}$, and a nonempty set $\mathcal{S} \subseteq \mathbb{R}^n$, a function $\gamma\colon [a, b] \to \mathcal{S}$ is said to be an $\alpha$-dense curve in the norm $\norm{\cdot}$ if for any $x \in \mathcal{S}$, we can find some $\theta \in [a, b]$ such that $\norm{\gamma(\theta) - x} \leq \alpha$~\cite{ziadi2001global}.

The curve $\gamma$ is said to be $\beta_\gamma$-Lipschitz-like if there is some non-negative, continuous, monotonically increasing function $\beta_\gamma\colon \mathbb{R}_+ \to \mathbb{R}$ with $\beta_\gamma(0) = 0$ such that for any $\theta, \theta' \in [a, b]$, we have $\norm{\gamma(\theta) - \gamma(\theta')} \leq \beta_\gamma(|\theta - \theta'|)$.
\end{definition}
The numbers $a, b \in \mathbb{R}$ are the end points of the interval over which the curve is defined. In this section, we will focus on constructing an $\alpha$-dense curve for the standard simplex $\mathcal{S} = \Delta_d$ in $l_1$-norm, where $\alpha = 2(d - 1)/N$. Here, $d$ is the dimension and $N$ is a positive integer that controls the value of $\alpha$. We achieve this by finding a way to efficiently connect the points of the grid $\mathcal{I}_{d, N}$ defined in Eq.~\eqref{eqn:integer_grid_simplex} in the previous section. Towards this end, we define the following ordering of the grid $\mathcal{I}_{d, N}$.
\begin{definition}[Equidistant ordering of $\mathcal{I}_{d, N}$]
    \label{def:equidistant_ordering_simplex_grid}
    Given $d, N \in \mathbb{N}_+$ with $d \geq 2$, let $\mathcal{I}_{d, N}$ be the grid defined in Eq.~\eqref{eqn:simplex_grid}. For $d = 2$, define the forward ordering $\mathcal{I}_{2, N} = \{(N, 0), (N - 1, 1), \dotsc, (1, N - 1), (0, N)\}$ and the reverse ordering $\mathcal{I}_{2, N} = \{(0, N), (1, N - 1), \dotsc, (N - 1, 1), (N, 0)\}$. For $d \geq 3$, define the forward ordering inductively as follows. Start with $\mathcal{I}_{d, N} = \varnothing$. For each $n_{d - 1} \in [N]$, forward order the elements of $\mathcal{I}_{d - 1, n_{d - 1}}$ if $n_{d - 1}$ is odd, and reverse order them if $n_{d - 1}$ is even. Append the elements $(N - n_{d - 1}, n_{d - 1} - n_{d - 2}, \dotsc, n_2 - n_1, n_1)$ with $(n_{d - 1} - n_{d - 2}, \dotsc, n_2 - n_1 n_1) \in \mathcal{I}_{d - 1, n_{d - 1}}$ as ordered above to the (ordered) set $\mathcal{I}_{d, N}$. Reverse ordering of $\mathcal{I}_{d, N}$ corresponds to writing the forward ordered set in the reverse order.
\end{definition}

We remark that $d = 2$ is just a special case of the definition showing the basis step for induction. It can be seen that the first element of $\mathcal{I}_{d, N}$ is $(N, 0, \dotsc, 0)$ and the last element is $(0, \dotsc, 0, N)$ if $\mathcal{I}_{d, N}$ is forward ordered (and the opposite is true if the elements are reverse ordered). In Prop.~\ref{prop:simplex_grid}, we show that $\mathcal{I}_{d, N}$ can be ordered in this manner and the distance between any two consecutive elements of $\mathcal{I}_{d, N}$ according to this ordering is $2$ measured in $l_1$-norm. By default, we work with forward ordering unless specified otherwise. As an example, the (forward) ordering of $\mathcal{I}_{3, 3}$ is given by
\begin{equation*}
  \mathcal{I}_{3, 3} = \{(3, 0, 0), (2, 1, 0), (2, 0, 1), (1, 0, 2), (1, 1, 1), (1, 2, 0), (0, 3, 0), (0, 2, 1), (0, 1, 2), (0, 0, 3)\}.
\end{equation*}
It can be seen that any two consecutive elements have an $l_1$-norm distance of $2$. Since the grid $\Delta_{d, N}$ defined in Eq.~\eqref{eqn:simplex_grid} is just a scaled version of $\mathcal{I}_{d, N}$, this ordering also applies to $\Delta_{d, N}$.

To generate a dense curve filling $\Delta_d$, we make use of the grid $\Delta_{d, N}$ along with the ordering described above. However, $\Delta_{d, N}$ has $N_{\text{grid}} = \binom{N + d - 1}{d - 1}$ elements, and therefore, it can get very expensive to store this grid in memory even if $N$ and $d$ are only moderately large (for example, $N = d = 20$ gives $N_{\text{grid}} \approx 7 \times 10^{10}$). For this purpose, we develop an algorithm that can efficiently query the elements of ordered $\Delta_{d, N}$ without explicitly constructing the set. We skip the details of this algorithm in the study, but include an implementation of the algorithm on Github (see Sec.~\ref{secn:code_availability}). Owing to such an approach, the resulting construction is both time and memory efficient, i.e.,  $\gamma(\theta)$ can be computed efficiently for a curve $\gamma$. We describe the construction of the curve below.
\begin{algorithm}[H]
    \begin{algorithmic}[1]
        \Function{construct\_dense\_curve\_standard\_simplex}{$d$, $N$, $\theta$}
            \State Compute the index $k = \lfloor \theta N/2 \rfloor$
            \State Compute $t = 1 + k - \theta N / 2$
            \State Obtain the grid points $x_k, x_{k + 1} \in \Delta_{d, N}$, where the elements of $\Delta_{d, N}$ are ordered as per Def.~\eqref{def:equidistant_ordering_simplex_grid}
            \State \textbf{return} $t x_k + (1 - t) x_{k + 1}$
        \EndFunction
    \end{algorithmic}
    \caption{Construct an $\alpha$-dense curve that fills the simplex $\Delta_d$ for $\alpha = 2(d - 1)/N$, $N \in \mathbb{N}_+$ and $\theta \in [0, L_{\text{curve}}]$, where $L_{\text{curve}}$ is given in Eq.~\eqref{eqn:length_dense_curve_simplex}}
    \label{alg:construct_simplex_dense_curve}
\end{algorithm}
We now explain the reasoning behind Alg.~\ref{alg:construct_simplex_dense_curve}. We construct the $(2(d - 1)/N)$-dense curve $\gamma$ by joining the consecutive points of the ordered grid $\Delta_{d, N}$. Since every two consecutive points in the ordered grid $\Delta_{d, N}$ are $2/N$ distance apart in $l_1$-norm, the total length of this curve measured in $l_1$-norm is
\begin{equation}
    L_{\text{curve}} = \frac{2}{N} N_{\text{grid}} = \frac{2}{N} \binom{N + d - 1}{d - 1}.\label{eqn:length_dense_curve_simplex}
\end{equation}
In order to obtain $\gamma(\theta)$ efficiently for a given $\theta \in [0, L_{\text{curve}}]$, we first compute the grid points $x_k$ and $x_{k + 1}$ for which $\gamma(\theta)\in [x_k, x_{k + 1}]$. In particular, we have $\gamma(\theta) = t x_k + (1 - t) x_{k + 1}$ for some $t \in [0, 1]$. In order to compute $t$, we note that $\norm{\gamma(\theta) - x_k}_1 = \theta - 2k/N$, where the RHS is the length of the curve at $\theta$ minus the length of the curve at grid point $x_k$ (which is $2k/N$). Since $\norm{\gamma(\theta) - x_k}_1 = (1 - t) \norm{x_{k + 1} - x_k}_1 = 2 (1 - t)/N$, we obtain $t = 1 + k - \theta N / 2$. Since $\theta N/2 - 1 \leq k \leq \theta N /2$, we have $0 \leq t \leq 1$. In Prop.~\ref{prop:Lipschitzlike_optimization_dense_curve_simplex}, we show that the curve constructed using Alg.~\ref{alg:construct_simplex_dense_curve} is $(2(d - 1)/N)$-dense and satisfies the property
\begin{equation}
    \norm{\gamma(\theta) - \gamma(\theta')}_1 \leq \min\{|\theta - \theta'|, 2\}. \label{eqn:simplex_dense_curve_Lipschitz}
\end{equation}
That is, the curve $\gamma$ is $\beta_\gamma$-Lipschitz-like with $\beta_\gamma(x) = \min\{x, 2\}$. In particular, $\gamma$ is Lipschitz continuous with Lipschitz constant $1$.

The essential idea behind the optimization of Lipschitz-like functions using an $\alpha$-dense curve is as follows. If the objective function $f\colon \mathcal{D} \to \mathbb{R}$ is Lipschitz-like and $\gamma\colon [a, b] \to \mathcal{D}$ is an $\alpha$-dense Lipschitz-like curve, then the function $f \circ \gamma\colon [a, b] \to \mathbb{R}$ is also Lipschitz-like. Therefore, one can optimize the function $f \circ \gamma$ using Alg.~\ref{alg:maximize_Lipschitz_like_function_1D}. In order to converge to the optimum within a precision of $\epsilon > 0$, the constant $\alpha$ must be chosen appropriately. The exact procedure is outlined in Alg.~\ref{alg:maximize_Lipschitz_like_function_simplex_dense_curve}.
\begin{algorithm}[H]
    \begin{algorithmic}[1]
        \Function{maximize\_Lipshcitz-like\_function\_simplex}{$d$, $\beta$, $\epsilon$}
            \State Compute the largest number $\alpha > 0$ such that $\beta(\alpha) \leq \epsilon/2$
            \State Set $N = \lceil 2(d - 1)/\alpha \rceil$
            \State Construct the $(2(d - 1)/N)$-dense curve $\gamma\colon [0, L_{\text{curve}}] \to \mathcal{D}$ as per Alg.~\ref{alg:construct_simplex_dense_curve}
            \State Compute the maximum $g^*$ of $g = f \circ \gamma$ over $[0, L_{\text{curve}}]$ to a precision of $\epsilon/2$ using Alg.~\ref{alg:maximize_Lipschitz_like_function_1D}
            \State \textbf{return} $g^*$
        \EndFunction
    \end{algorithmic}
    \caption{Computing the maximum of a $\beta$-Lipschitz-like function $f$ satisfying Eq.~\eqref{eqn:Lipschitz_like_function} for $\mathcal{D} = \Delta_d$, given $\epsilon > 0$}
    \label{alg:maximize_Lipschitz_like_function_simplex_dense_curve}
\end{algorithm}
In Prop.~\ref{prop:Lipschitzlike_optimization_dense_curve_simplex}, we show that Alg.~\ref{alg:maximize_Lipschitz_like_function_simplex_dense_curve} is guaranteed to converge to the maximum of $f$ to within a precision of $\epsilon > 0$. This takes $\lceil 2 \binom{N + d - 1}{d - 1} / N\alpha \rceil$ iterations in the worst case. For large dimensions and fixed precision, this number scales as $O(\alpha^{1 - d}/d)$ with the dimension.

One can parallelize the above algorithm by breaking the interval $[0, L_{\text{curve}}]$ into finitely many sub-intervals and optimizing the function over each interval separately. The final maximum can be obtained by taking the maximum of the maximum values computed for each sub-interval. Since we only break the interval into finitely many sub-intervals, the guarantees given by Prop.~\ref{prop:Lipschitzlike_optimization_dense_curve_simplex} remain valid.

\subsubsection{Numerical examples}
For verifying the performance of the grid search algorithm (Alg.~\ref{alg:maximize_Lipschitz_like_function_simplex_grid_search}) and the algorithm based on dense curves (Alg.~\ref{alg:maximize_Lipschitz_like_function_simplex_dense_curve}), we present two numerical examples in Tab.~\ref{tab:Lipschitzlike_maximization_simplex_examples} where the maximum over the simplex can be computed exactly. The functions considered in both these examples are Lipschitz continuous with respect to the $l_1$-norm. It can be seen that both these methods compute the maximum to within the specified precision. For small dimensions and small values of tolerance $\epsilon$, we find that the algorithm based on dense curves requires far fewer iterations than grid search. As a result, the dense curve algorithm converges much faster than grid search for the examples presented in Tab.~\ref{tab:Lipschitzlike_maximization_simplex_examples}. On the other hand, for larger values of tolerance and higher dimensions, grid search can converge faster since its complexity scales polynomially with the dimension whereas the complexity of dense curve algorithm scales exponentially in the worst case. Even so, the power of the polynomial in complexity of grid search can be large. Therefore, as the dimension $d$ increases, both these algorithms become infeasible to implement in practice.

\begin{table}[ht]
    \begin{center}
        \begin{tabular}{ccccc}
            \toprule
            \textbf{function} & \textbf{algorithm} & \textbf{true maximum} & \textbf{computed maximum} & \textbf{iterations} \\
            \midrule
            \multirow{2}{*}{$f(x) = \sin(\norm{x}_2)$} & grid search & 0.841 & 0.841 & 16290 \\
                                                       & dense curve & 0.841 & 0.836 & 155 \\
            \midrule
            \multirow{2}{*}{$f(x) = -\frac{\norm{x}_2^3}{6} + \frac{\norm{x}_2^2}{4} - \frac{\norm{x}_2}{6\pi}$} & grid search & 0.033 & 0.033 & 19900 \\
                                                                                                                 & dense curve & 0.033 & 0.033 & 480 \\
            \bottomrule
        \end{tabular}
    \end{center}
    \caption{Computing the maximum of Lipschitz-continuous functions $f$ over the standard simplex $\Delta_d \subseteq \mathbb{R}^d$ in dimension $d = 3$ to a precision of $\epsilon = 0.15$. Grid search (Alg.~\ref{alg:maximize_Lipschitz_like_function_simplex_grid_search}) and the algorithm based on dense curves (Alg.~\ref{alg:maximize_Lipschitz_like_function_simplex_dense_curve}) are used to compute the maximum numerically. Numerical values are rounded to three decimal places.}
    \label{tab:Lipschitzlike_maximization_simplex_examples}
\end{table}

Details on how one can use $\alpha$-dense curves to optimize a Lipschitz-like function over any compact and convex domain is given in App.~\ref{app:Lipschitzlike_optimization_compact_convex_domain}.

\section{Sum capacity computation of two-sender MACs\label{secn:sum_capacity_computation_algorithm}}
\subsection{Sum capacity computation as optimization of a Lipschitz-like function\label{secn:sum_capacity_computation_Lipschitz_like}}
Our focus in this part is computing the sum capacity of arbitrary $2$-sender MACs. As noted in the preliminary Sec.~\ref{secn:preliminaries_MAC}, the optimization involved in computing the sum capacity is non-convex. Unfortunately, the approach via convex relaxation commonly pursued is, in general, not a viable strategy for bounding the sum capacity of MACs. The family of two-sender MACs constructed in Sec.~\ref{secn:relaxedALcapacity_sumcapacity_separation} provides a striking example of this claim, where we show that the difference between the convex relaxation and the actual sum capacity can be made arbitrarily large. In light of this result, developing algorithms to compute, or at least better approximate, the sum capacity becomes important, and we undertake this task here.

As discussed in Sec.~\ref{secn:preliminaries_MAC}, a two-sender MAC with
input alphabets $\BC_1$ and $\BC_2$ and output alphabet $\ZC$ is described by
a transition matrix $\NC$. Each entry of this matrix, $\NC(z|b_1,b_2)$, where
$z \in \ZC, b_1 \in \BC_1$, and $b_2 \in \BC_2$, represents the probability
that the channel output takes a value $z$ when the first and second channel
inputs take values $b_1$ and $b_2$, respectively.
We introduce new notation for this sub-section. Let $d_1, d_2$, and $\dout$
denote $|\BC_1|, |\BC_2|,$ and $|\ZC|$ respectively. We write $p(b_1) \in
\Delta_{d_1}$, but refer to $p(b_2)$ by $q(b_2)$ where $q(b_2) \in
\Delta_{d_2}$.
In this notation, the sum capacity~\eqref{eqn:sum_capacity} of the two-sender
MAC $\NC$ takes the form
\begin{equation}
    S(\NC) = \max_{p(b_1, b_2)} I(B_1, B_2; Z) \quad \text{such that} \quad p(b_1,b_2) = p(b_1) q(b_2),
    \label{eq:twoSenderSumCap}
\end{equation}
where $B_1, B_2$, and $Z$ are random variables that describe the channel's
first input, second input, and output respectively.
For fixed $\NC(z|b_1, b_2)$, the mutual information $I(B_1, B_2; Z)$ function
is concave in the argument $p(b_1, b_2)$. On the other hand, the set of joint
distributions which satisfy the product constraint, $p(b_1, b_2) = p(b_1)
p(b_2)$, is not convex.
This lack of convexity turns the maximization
in Eq.~\eqref{eq:twoSenderSumCap} into a non-convex problem.

Our approach to solving the non-convex problem~\eqref{eq:twoSenderSumCap} is
to move the non-convexity from the constraint to the objective function. Instead
of maximizing over the set of product distributions, $p(b_1, b_2) = p(b_1)
q(b_2)$, we maximize sequentially, that is, we write
\begin{equation}
    S(\mathcal{N}) = \max_{q \in \Delta_{d_2}} \max_{p \in \Delta_{d_1}} I(B_1, B_2; Z).
    \label{eq:twosendersumcapre}
\end{equation}
Now, both the inner and outer maximization in Eq.~\eqref{eq:twosendersumcapre} are
carried out over convex sets $\Delta_{d_1}$ and $\Delta_{d_2}$, respectively. 
To carry out these optimizations we derive certain convenient expressions. The
output probability distribution $p^Z$ over $\ZC$ can be written as
\begin{align}
    p^Z(z) &= \sum_{b_1 \in \BC_1} A_q(z, b_1) p(b_1) \label{eqn:outputprob_twosenderMAC_productinput}
    \intertext{where}
    A_q(z, b_1) &= \sum_{b_2 \in \BC_2} \mathcal{N}(z | b_1, b_2) q(b_2) \label{eqn:outputprob_twosenderMAC_productinput_Aq}
\end{align}
for $z \in \ZC$ and $b_1\in\BC_1$. Note that $A_q$ can be considered as a left stochastic matrix
of size $\dout \times d_1$, i.e., every entry of $A_q$ is non-negative and the
columns sum to $1$.  One can view
Eq.~\eqref{eqn:outputprob_twosenderMAC_productinput} as a vector equation $p^Z =
A_q p$, where $p \in \Delta_{d_1}$ and $p^Z \in \Delta_{\dout}$. The mutual
information $I(B_1, B_2; Z)$ can be written as
\begin{equation*}
    I(B_1, B_2; Z) = H(Z) - \sum_{b_1, b_2} p(b_1) q(b_2) H(Z | B_1 = b_1, B_2 = b_2).
\end{equation*}
To express the mutual information in terms of vectors and matrices, we define a $d_1$-dimensional vector $b_q$ with non-negative components $b_q(b_1)$ for $b_1 = 1,\dots,d_1$, where
\begin{equation}
    b_q(b_1) = -\sum_{b_2 \in \mathcal{B}_2} q(b_2) \sum_{z \in \mathcal{Z}} \mathcal{N}(z | b_1, b_2) \log(\mathcal{N}(z | b_1, b_2)).
    \label{eqn:outputprob_twosenderMAC_productinput_bq}
\end{equation}
This allows us to express the mutual information compactly as
\begin{equation}
    I(p, q) \equiv I(B_1, B_2; Z) = H(A_q p) - \ip{b_q, p}, \label{eqn:mutualinfo_twosenderMAC_productinput}
\end{equation}
and the sum capacity as
\begin{equation}
    S(\mathcal{N}) = \max_{q \in \Delta_{d_2}} \max_{p \in \Delta_{d_1}} \left\lbrace H(A_q p) - \ip{b_q, p}\right\rbrace.
    \label{eqn:sum_capacity_twosenderMAC}
\end{equation}
The inner optimization over $p \in \Delta_{d_1}$ is a convex optimization
problem because $I(p, q)$ is a concave function of $p$ for fixed $q$.
Since $I(p, q)$ is not jointly concave over $(p, q)$, the function
\begin{equation}
    I^*(q) = \max_{p \in \Delta_{d_1}} \left(H(A_q p) - \ip{b_q, p}\right) 
    \label{eqn:outer_optimization_objective_twosenderMAC}
\end{equation}
is not concave in general. Therefore, the outer optimization of $I^*(q)$ over
$q \in \Delta_{d_2}$, i.e., $S(\mathcal{N}) = \max_{q \in \Delta_{d_2}} I^*(q)$, is in general a nonconvex problem.

Nevertheless, we show that the non-concave function $I^*(q)$ is a Lipschitz-like function as defined in Eq.~\eqref{eqn:Lipschitz_like_function}.
To elaborate, this means there is some real-valued function $\beta_I$ that is non-negative, continuous, and monotonically increasing with $\beta_I(0) = 0$, such that $|I^*(q) - I^*(q')| \leq \beta_I(\norm{q - q'}_1)$.
We show that $I^*(q)$ indeed satisfies such a property by proving an appropriate continuity bound for $I^*(q)$.
This is summarized in the following result.

\begin{proposition}
    \label{prop:mutual_information_maximum_Lipschitz_like}
    Let $\mathcal{N}$ be any two-sender MAC with input alphabets $\mathcal{B}_1$, $\mathcal{B}_2$ of size $d_1$, $d_2$, and output alphabet $\mathcal{Z}$ of size $\dout$. Assume that $d_1, d_2, \dout \geq 2$. Given input probability distributions $p$ over $\mathcal{B}_1$ and $q$ over $\mathcal{B}_2$, the mutual information between the inputs and the output of the MAC can be written as
    \begin{equation*}
        I(p, q) = H(A_q p) - \ip{b_q, p}
    \end{equation*}
    where the matrix $A_q$ and the vector $b_p$ are defined in Eq.~\eqref{eqn:outputprob_twosenderMAC_productinput_Aq} and Eq.~\eqref{eqn:outputprob_twosenderMAC_productinput_bq} respectively. Define the function
    \begin{equation*}
        I^*(q) = \max_{p \in \Delta_{d_1}} I(p, q).
    \end{equation*}
    Then, for any $p \in \Delta_{d_1}$, and any $q, q' \in \Delta_{d_2}$, we have
    \begin{equation*}
        |I(p, q) - I(p, q')| \leq \beta_I\left(\norm{q - q'}_1\right),
    \end{equation*}
    and subsequently,
    \begin{equation}
        |I^*(q) - I^*(q')| \leq \beta_I\left(\norm{q - q'}_1\right). \label{eq:L-Like}
    \end{equation}
    The function $\beta_I$ is defined as
    \begin{align}
        \beta_I(x) &= \left(\frac{1}{2} \log(\dout - 1) + H_{\mathcal{N}}^{\max}\right) x + \overline{h}\left(\frac{x}{2}\right), \label{eqn:beta_twosenderMAC}
        \intertext{where}
        H_{\mathcal{N}}^{\max} &= \max_{a_1 \in \mathcal{A}_1, a_2 \in \mathcal{A}_2} \left\lbrace -\sum_{z \in \mathcal{Z}} \mathcal{N}(z | a_1, a_2) \log(\mathcal{N}(z | a_1, a_2))\right\rbrace, \label{eqn:HNmax_twosenderMAC}
        \intertext{and}
        \overline{h}(x) &= \begin{cases} -x\log(x) - (1 - x)\log(1 - x) \hspace{0.1cm} \text{ if } x \leq \frac{1}{2} \\
                                         \log(2)                        \hspace{4.1cm} \text{ if } x \geq \frac{1}{2}
                           \end{cases} \nonumber
    \end{align}
    is the modified binary entropy defined in Eq.~\eqref{eqn:modified_binary_entropy}.
\end{proposition}
\begin{proof}
    See Appendix~\ref{proof:mutual_information_maximum_Lipschitz_like}.
\end{proof}

This observation is important in the sense that it allows for off-the-shelf use of algorithms developed in Sec.~\ref{secn:Lipschitzlike_optimization} for optimizing any Lipschitz-like function in order to compute the sum capacity of two-sender MACs. Following this line of approach, we will begin by developing an efficient algorithm for computing the sum capacity of any two-sender MAC where one of the input alphabets is of size $2$. As a result, we can efficiently compute the sum capacity of a large family of MACs that includes all binary MACs. Next, we will develop two algorithms for computing the sum capacity of an arbitrary two-sender MAC. The first algorithm is important from a theoretical standpoint, and shows that sum capacity of an arbitrary two-sender MAC can be computed in quasi-polynomial time. However, it can be costly to implement in practice. Our second algorithm can be faster to run in practice, at least when one of the MAC input alphabet sizes is small, but it suffers from exponential complexity in the worst case.

\subsection{Computing the sum capacity of a two-sender MAC with one input alphabet of size 2\label{secn:sum_capacity_computation_algorithm_inputsize2}}
As before, we take $\mathcal{N}$ to be a two-sender MAC with input alphabets $\mathcal{B}_1$, $\mathcal{B}_2$ of sizes $d_1$, $d_2$ and an output alphabet $\mathcal{Z}$ of size $\dout$. In this section, we focus on the case where at least one of $d_1$ or $d_2$ is equal to $2$. For concreteness, take $d_2 = 2$ (this choice is inconsequential for the algorithm).

For this case, any probability distribution $q_s \in \Delta_{d_2}$ can be expressed as $q_s = (s, 1 - s)$ for some $0 \leq s \leq 1$. Thus, the maximization over $q_s \in \Delta_2$ in computing the sum capacity $S(\mathcal{N}) = \max_{q_s \in \Delta_{d_2}} I^*(q_s)$ is essentially one-dimensional. In other words, considering the objective $I^*(s) \coloneqq I^*(q_s)$ as a function of $s$, we can write the sum capacity for any MAC $\mathcal{N}$ with $d_2 = 2$ as
\begin{equation*}
    S(\mathcal{N}) = \max_{s \in [0, 1]} I^*(s).
\end{equation*}
We will show that the mutual information $I^*(s)$ considered as a function of $s$ is still a Lipschitz-like function. First, observe that for $q_s = (s, 1 - s)$, we have $\norm{q_s - q_{s'}}_1 = 2|s - s'|$. Then, $|I^*(s) - I^*(s')| = |I^*(q_s) - I^*(q_{s'})| \leq \beta_I(\norm{q_s - q_{s'}}_1) = \beta_I(2 |s - s'|)$. Therefore, taking $\beta(x) = \beta_I(2 x)$, we find that $I^*$ as a function of $s$ is $\beta$-Lipschitz-like. Subsequently, we can use modified Piyavskii-Shubert algorithm developed in Sec.~\ref{secn:Lipschitzlike_optimization_1D} to compute the maximum $\max_{s \in [0, 1]} I^*(s)$ to any given precision.

For the convenience of the reader, we rewrite Alg.~\ref{alg:maximize_Lipschitz_like_function_1D} specifically for the purpose of computing the sum capacity of a two-sender MAC $\mathcal{N} $when $d_2 = 2$. This sum capacity computation algorithm is summarized in Alg.~\ref{alg:compute_sum_capacity_twosenderMACs_inputsize2} below. We note that similar to Alg.~\ref{alg:maximize_Lipschitz_like_function_1D}, one can modify Alg.~\ref{alg:compute_sum_capacity_twosenderMACs_inputsize2} so that it accepts a fixed number of iterations and outputs the value $F(s^*)$, which is an upper bound on the sum capacity. This upper bound exceeds the sum capacity by at most $F(s^*) - I^*(s^*)$.
\begin{algorithm}[ht]
    \begin{algorithmic}[1]
        \Function{compute\_sum\_capacity}{$\mathcal{N}, d_1, \epsilon$}
        \State Initialize $s^{(0)} = 0$
        \State Define $F_0(s) = I^*(s^{(0)}) + \beta_I(2|s - s^{(0)}|)$ for $s \in [0, 1]$
            \State Set $F \gets F_0$ and $s^* \gets 1$
            \State Set $s^{(1)} \gets s^*$, $k \gets 1$
            \While{$F(s^*) - I^*(s^*) > \epsilon$}
                \State Sort $\{s^{(0)}, \dotsc, s^{(k)}\}$ from smallest to largest and relabel the points in ascending order.
                \State Set $F_i(s) = I^*(s^{(i)}) + \beta_I(2|s - s^{(i)}|)$ for $0 \leq i \leq k$
                \For{$i = 0, \dotsc, k - 1$}
                    \State Set $g_i(s) = F_i(s) - F_{i + 1}(s)$
                    \State Find $\overline{s}_i \in [s^{(i)}, s^{(i + 1)}]$ such that $g_i(\overline{s}_i) = 0$ using any root finding method.
                \EndFor
                \State Pick an index $m \in \argmax_{0 \leq i \leq k - 1} F_i(\overline{s}_i)$ and set $s^* = \overline{s}_m$.
                \State Update $F \gets F_m$
                \State Set $s^{(k + 1)} \gets s^*$, $k \gets k + 1$
            \EndWhile
            \State \textbf{return} $I^*(s^*)$
        \EndFunction
    \end{algorithmic}
    \caption{Computing the sum capacity of a two-sender MAC $\mathcal{N}$ to a precision $\epsilon > 0$, when one of the input alphabets has size $2$.}
    \label{alg:compute_sum_capacity_twosenderMACs_inputsize2}
\end{algorithm}

We make a technical remark that is important to perform the optimization as per Alg.~\ref{alg:compute_sum_capacity_twosenderMACs_inputsize2}. The objective function $I^*$ appearing in the optimization $S(\mathcal{N}) = \max_{s \in [0, 1]} I^*(x)$ is non-trivial to compute. More precisely, in order to compute the value of $I^*(s)$ for any given $s$, we need to solve another optimization problem. This follows from the definition of $I^*$ given in Eq.~\eqref{eqn:outer_optimization_objective_twosenderMAC}. Fortunately, this optimization problem is convex, and hence, it can be solved by standard convex optimization techniques.

In Prop.~\ref{prop:max_iter_sum_capacity_alg_twosenderMAC_inputsize2}, we show that the number of iterations required by the \texttt{while} loop in Alg.~\ref{alg:compute_sum_capacity_twosenderMACs_inputsize2} to converge to the sum capacity within a tolerance of $0 < \epsilon \leq 3$ is bounded above as $O(\log(\dout)/\epsilon)$, where $\dout \geq 2$ is the size of the output alphabet and $\epsilon$ is a fixed constant. Note that this bound does not account for the number of iterations required to compute $I^*$, for sorting or for root-finding. We show in Prop.~\ref{prop:max_iter_sum_capacity_alg_twosenderMAC_inputsize2} that the total cost involved is at most polynomial in $d_1$, $\dout$, and $1/\epsilon$.

Next, we show how to compute the sum capacity of any two-sender MAC.

\subsection{Computing the sum capacity of any two-sender MAC\label{secn:sum_capacity_computation_algorithm_2senderMAC}}
Let an arbitrary two-sender MAC $\mathcal{N}$ with input alphabet sizes $d_1, d_2$ and output alphabet size $\dout$. Without loss of generality, we assume that $d_2 \leq d_1$. Then, the sum capacity computation can be expressed as $S(\mathcal{N}) = \max_{q \in \Delta_{d_2}} I^*(q)$. This amounts to optimizing the Lipschitz-like function $I^*(q)$ over the standard simplex $\Delta_{d_2}$. As a result, the algorithms developed in Sec.~\ref{secn:Lipschitzlike_optimization_simplex_grid_search} and Sec.~\ref{secn:Lipschitzlike_optimization_simplex_dense_curve} can both be used to compute the sum capacity of $\mathcal{N}$. 
For the convenience of the reader, we present these algorithms here again adapted specifically for sum capacity computation.

The first algorithm that we discuss is grid search. This algorithm is helpful in proving complexity results, but not as helpful from a practical implementation standpoint.
\subsubsection{Sum capacity computation using grid search}
Our goal is to perform the optimization of $I^*(q)$ over the simplex $q \in \Delta_{d_2}$. We perform this optimization by computing the maximum over the grid
\begin{equation*}
    \Delta_{d, N} = \left\{\frac{n}{N} \mid n \in \mathcal{I}_{d,N}\right\}
\end{equation*}
defined in Eq.~\eqref{eqn:simplex_grid}. Thus, the grid search algorithm can be described as follows.

\begin{algorithm}[H]
    \begin{algorithmic}[1]
        \Function{compute\_sum\_capacity\_grid\_search}{$\mathcal{N}$, $d_1$, $d_2$, $\dout$, $\epsilon$}
            \State Set $N = \lceil 1/\delta^2 \rceil$, where $\delta$ is the largest number satisfying $\beta_I(\delta) \leq \epsilon/2$
            \State Construct the grid $\Delta_{d_2, N}$
            \State \textbf{return} $S^* = \max\{I^*(q) \mid q \in \Delta_{d_2, N}\}$
        \EndFunction
    \end{algorithmic}
    \caption{Computing the sum capacity of a two-sender MAC $\mathcal{N}$ to a precision $\epsilon > 0$ using grid search. $d_1, d_2$ are input alphabet sizes with $d_2 \leq d_1$ and $\dout$ is the output alphabet size. The function $\beta_I$ is defined in Eq.~\eqref{eqn:beta_twosenderMAC}.}
    \label{alg:compute_sum_capacity_twosenderMACs_simplex_grid_search}
\end{algorithm}
As noted in Sec.~\ref{secn:Lipschitzlike_optimization_simplex_grid_search}, optimization of a $\beta$-Lipschitz-like functions using grid search can be done in polynomial time if the function $\beta$ does not depend on the dimension. For the case of sum capacity computation, the function $\beta$ is $\beta_I$ given in Eq.~\eqref{eqn:beta_twosenderMAC}. $\beta_I$ depends on the size of the output alphabet of the MAC, and therefore, the complexity analysis is more involved. Furthermore, for any given $q \in \Delta_{d_2}$, the function $I^*(q)$ needs to be computed using convex optimization. Thus, the cost of computing $I^*$ also needs to be included in the complexity analysis.

In Prop.~\ref{prop:max_iter_sum_capacity_alg_twosenderMAC_grid_search}, we show that the total cost of computing the sum capacity to a fixed precision $0 < \epsilon \leq 1$ is roughly bounded above by $\text{poly}(d_1, \dout, 1/\epsilon) O(d_2^{96 \log^2(\dout)/\epsilon^2 + 2})$ (see Eq.~\eqref{eqn:grid_search_complexity_sum_capacity_twosenderMAC} for a more precise statement).
Therefore, the sum capacity can be computed to a fixed precision $\epsilon > 0$ in quasi-polynomial time.
This quasi-polynomial behaviour comes from the fact that $\beta_I$ depends on $\log(\dout)$.
From the proof of Prop.~\ref{prop:max_iter_sum_capacity_alg_twosenderMAC_grid_search}, we can also infer that if one fixes the output alphabet size $\dout$, then the sum capacity can be computed in polynomial time.

While these results are useful from a theoretical standpoint, grid search can be slow in practice.
The reason is two-fold. One, the power of the polynomial appearing in time complexity can be large even for moderately small precision and dimensions (i.e., the size of the grid becomes fairly large). Two, for each point in the grid, we need to solve a convex optimization problem for computing $I^*$. Together, these factors make grid search not ideal for practical implementations. We note that one can nevertheless parallelize grid search to get some improvements in the speed, though eventually even this will be too costly. With this in mind, we study another algorithm to compute the sum capacity using dense curves.

\subsubsection{Sum capacity computation using dense curves}
In this section, we will study the application of algorithm developed in Sec.~\ref{secn:Lipschitzlike_optimization_simplex_dense_curve} for computing the sum capacity. The idea is to fill the simplex with a curve that comes within a distance of $\alpha$ to any given point on the simplex. Such a curve is therefore called an $\alpha$-dense curve. Using such a curve, we reduce a high-dimensional optimization problem to a one-dimensional optimization problem over an interval. The one-dimensional optimization problem can be solved using the modified Piyavskii-Shubert algorithm that underlies Alg.~\ref{alg:compute_sum_capacity_twosenderMACs_inputsize2}.

We refer the reader to Alg.~\ref{alg:construct_simplex_dense_curve} which shows how to construct an $\alpha$-dense curve for filling the standard simplex for $\alpha = 2(d_2 - 1) / N$, where $N$ is a positive integer that controls the value of $\alpha$. Below, we show how to use Alg.~\ref{alg:maximize_Lipschitz_like_function_simplex_dense_curve} to compute the sum capacity of a two-sender MAC.
\begin{algorithm}[H]
    \begin{algorithmic}[1]
        \Function{compute\_sum\_capacity\_dense\_curve}{$\mathcal{N}$, $d_1$, $d_2$, $\dout$, $\epsilon$}
            \State Compute the largest number $\alpha > 0$ such that $\beta_I(\alpha) \leq \epsilon/2$
            \State Set $N = \lceil 2(d_2 - 1)/\alpha \rceil$
            \State Construct the $(2(d_2 - 1)/N)$-dense curve $\gamma\colon [0, L_{\text{curve}}] \to \mathcal{D}$ as per Alg.~\ref{alg:construct_simplex_dense_curve}
            \State Compute the maximum $S^*$ of the function $I^* \circ \gamma$ over $[0, L_{\text{curve}}]$ to a precision of $\epsilon/2$ using Alg.~\ref{alg:maximize_Lipschitz_like_function_1D}
            \State \textbf{return} $S^*$
        \EndFunction
    \end{algorithmic}
    \caption{Computing the sum capacity of a two-sender MAC $\mathcal{N}$ to a precision $\epsilon > 0$ using a dense curve. $d_1, d_2$ are input alphabet sizes with $d_2 \leq d_1$ and $\dout$ is the output alphabet size. The function $\beta_I$ is defined in Eq.~\eqref{eqn:beta_twosenderMAC}.}
    \label{alg:compute_sum_capacity_twosenderMACs_simplex_dense_curve}
\end{algorithm}
From Prop.~\ref{prop:Lipschitzlike_optimization_dense_curve_simplex}, we can infer that the above algorithm is guaranteed to compute the sum capacity to within a precision of $\epsilon > 0$. Prop.~\ref{prop:Lipschitzlike_optimization_dense_curve_simplex} shows that the dense curve algorithm can take an exponential time in the worst case to converge. For this reason, we avoid doing a detailed complexity analysis of this algorithm for sum capacity computation, as better theoretical guarantees can be obtained using grid search. Instead, we focus on more practical gains that may be obtained using Alg.~\ref{alg:compute_sum_capacity_twosenderMACs_simplex_dense_curve} over Alg.~\ref{alg:compute_sum_capacity_twosenderMACs_simplex_grid_search}.

Using Alg.~\ref{alg:compute_sum_capacity_twosenderMACs_simplex_dense_curve} gives another advantage. Instead of computing the sum capacity to a given precision, one can run Alg.~\ref{alg:compute_sum_capacity_twosenderMACs_simplex_dense_curve} for a fixed number of time steps to get an upper bound on the sum capacity. Thus, it can also be used for bounding the sum capacity instead of computing it to a good precision.

\subsubsection{Comparing grid search and dense curve algorithm performance}
We compare the performance of the grid search algorithm and the dense curve algorithm for computing the sum capacity. For this purpose, we will consider a randomly constructed two-sender MAC with input alphabet sizes $d_1, d_2$ and output alphabet size $\dout$. We fix $d_1 = 10$ and $\dout = 20$ and consider the cases $d_2 = 2$ and $d_2 = 3$. Note that the smaller dimension $d_2$ determines the dimension for the non-convex optimization. Hence, the value of $d_2$ effectively decides the overall computation time. The numerical simulations are run on single core of a personal computer, and the code is implemented in Python (see Sec.~\ref{secn:code_availability} for more details). All the run times reported are average time taken over three repetitions of the simulation.

For $d_2 = 2$, the dense curve algorithm can be simplified to modified Piyavskii-Shubert algorithm, as noted in Sec.~\ref{secn:sum_capacity_computation_algorithm_inputsize2}. Thus, for $d_2 = 2$, we compare the performance of modified Piyavskii-Shubert Alg.~\ref{alg:compute_sum_capacity_twosenderMACs_inputsize2} with that of grid search Alg.~\ref{alg:compute_sum_capacity_twosenderMACs_simplex_grid_search} for computing the sum capacity to a tolerance of $\epsilon = 0.15$. For the example we consider, we find that Alg.~\ref{alg:compute_sum_capacity_twosenderMACs_inputsize2} takes about $0.34 \text{ s}$ to run, whereas Alg.~\ref{alg:compute_sum_capacity_twosenderMACs_simplex_grid_search} takes about $1.51 \text{ mins}$ to run. Thus, for this example, modified Piyavskii-Shubert algorithm is more than $250$ times faster than grid search. For a tolerance of $\epsilon = 0.65$, the run time of grid search Alg.~\ref{alg:maximize_Lipschitz_like_function_simplex_grid_search} improves greatly to $1.06 \text{ s}$, whereas the run time of modified Piyavskii-Shubert Alg.~\ref{alg:compute_sum_capacity_twosenderMACs_inputsize2} improves to $0.11 \text{ s}$. In this case, the modified Piyavskii-Shubert algorithm is almost $10$ ten times faster than grid search. We remark that a tolerance of $\epsilon = 0.65$ is not acceptable in practice because the computed sum capacity is about $0.19$ nats. We therefore use such a large value of $\epsilon$ only for benchmarking purposes.

Next, we consider the higher dimensional case of $d_2 = 3$. For this case, we compare the dense curve Alg.~\ref{alg:compute_sum_capacity_twosenderMACs_simplex_dense_curve} with grid search Alg.~\ref{alg:compute_sum_capacity_twosenderMACs_simplex_grid_search}. For a tolerance of $\epsilon = 0.65$, we find that the dense curve Alg.~\ref{alg:compute_sum_capacity_twosenderMACs_simplex_dense_curve} takes $3.55 \text{ s}$ whereas grid search Alg.~\ref{alg:compute_sum_capacity_twosenderMACs_simplex_grid_search} takes $3.65 \text{ min}$. In this case, dense curve algorithm is about $65$ times faster, though as noted previously, a tolerance of $\epsilon = 0.65$ cannot be used in practice. For $\epsilon = 0.15$, the dense curve Alg.~\ref{alg:compute_sum_capacity_twosenderMACs_simplex_dense_curve} takes about $1.72 \text{ mins}$, whereas grid search takes too long to complete. We estimate that our current implementation of grid search algorithm will take more than $3$ days to complete for $\epsilon = 0.15$ on the hardware the code was executed. This shows that for the example under consideration, grid search is not a practical option.

We remark that there can be some cases where grid search performs as good as or better than the dense curve algorithm. For small dimensions, grid search scales poorly with the tolerance in comparison with dense curve and modified Piyavskii-Shubert algorithms, as evidenced from the above numerical simulations. As the dimension $d_2$ becomes large and $d_1, \dout$, and $\epsilon$ are fixed, the grid search algorithm scales polynomially with $d_2$ whereas dense curve algorithm scales exponentially. Despite this, the polynomial power can be so large in practice that grid search is still impractical. Thus, our numerical simulations suggest that for small dimensions and small tolerance, modified Piyavskii-Shubert algorithm (for $d_2 = 2$) and the dense curve algorithm (for $d_2 > 2$) perform better than the grid search algorithm in practice.

As noted previously, both the grid search algorithm and the dense curve algorithm (as well as modified Piyavskii-Shubert algorithm) can be parallelized to get better performance. It remains to see if these algorithms can be improved further so as to make them practical for larger values of $d_2$. Finally, we show how our algorithms compare with the relaxed sum capacity.

\subsection{Comparison with relaxed sum capacity\label{secn:sum_capacity_relaxed_sum_capacity_comparison}}
In order to compare our algorithms for computing the sum capacity with the relaxed sum capacity, we construct a family of binary MACs parametrized by a tunable parameter $t \in [0, 1]$. These MACs have the property that for $t = 0$, the relaxed sum capacity is equal to the actual sum capacity, whereas for $t = 1$, the relaxed sum capacity is twice the actual sum capacity. By computing the sum capacity using our algorithms as well as the relaxed sum capacity for as a function of $t$, we can compare how well our algorithms do in relation to the convex relaxation approach.

For generating examples of such binary MACs, we construct a family of MACs that we call the noise-free subspace MAC. To that end, let $\mathcal{A}$ and $\mathcal{B}$ be input alphabets and let $\mathcal{Z}$ be the output alphabet of the noise-free subspace MAC. We assume that $\mathcal{A}$, $\mathcal{B}$ and $\mathcal{Z}$ are finite sets. The noise-free subspace corresponds to a set $\mathcal{W} \subseteq \mathcal{A} \times \mathcal{B}$. Consider the mapping $n_F\colon \mathcal{W} \to \mathcal{Z}$ that determines the symbol that is deterministically output by the channel when the input is in the noise-free subspace. Then, given the tuple $(\mathcal{A}, \mathcal{B}, \mathcal{Z}, \mathcal{W}, n_F)$, the MAC $\mathcal{N}_F$ has the probability transition matrix
\begin{equation}
    \mathcal{N}_F(z | a, b) = \begin{cases}
                                      \delta_{z, n_F(a, b)}   & (a, b) \in \mathcal{W} \\
                                      \frac{1}{|\mathcal{Z}|} & (a, b) \notin \mathcal{W}
                                  \end{cases} \label{eqn:NF_MAC}
\end{equation}
When the input is in the noise-free subspace, the channel deterministically outputs the symbol selected by $n_F$. On the other hand, when the input is not in the noise-free subspace, the channel outputs a symbol uniformly at random. Therefore, a noise-free subspace MAC can be thought of as a generalization of nonlocal games MAC studied in Part I (see Eq.~\eqref{eqn:MAC_nonlocal_game}).

Now, for constructing the parametrized family of binary MACs mentioned at the beginning of the section, we construct two examples of the MAC $\mathcal{N}_F$ corresponding to the ``extremities" of the parametrized family. For both of these examples, we consider the alphabets $\mathcal{A} = \{a_1, a_2\}$, $\mathcal{B} = \{b_1, b_2\}$, and $\mathcal{Z} = \{z_1, z_2\}$. The examples are labelled $0$ and $1$.

For the first example, we take $\mathcal{W} = \{(a_1, b_1)\}$ and $n_F(a_1, b_1) = z_1$. That is, only a single input $(a_1, b_1)$ is transmitted noise-free. For this example, we can write the probability transition matrix as
\begin{equation*}
    \mathcal{N}_F^{(0)} = \begin{pmatrix}
                              1 & 0.5 & 0.5 & 0.5 \\
                              0 & 0.5 & 0.5 & 0.5
                          \end{pmatrix}
\end{equation*}
where the rows correspond to $z \in \mathcal{Z}$, while the columns correspond to $(a, b) \in \mathcal{A} \times \mathcal{B}$. In App.~\ref{app:noise_free_subspace_MAC}, we show that the sum capacity in nats is
\begin{equation*}
    S(\mathcal{N}_F^{(0)}) = h\left(\frac{4}{5}\right) - \frac{2}{5} \ln(2) \approx 0.223
\end{equation*}
where $h$ is the binary entropy measured in nats. We furthermore show that the relaxed sum capacity in this case is also equal to $C(\mathcal{N}_F^{(0)}) = 0.223$ nats. Thus, for our first example, we have $C(\mathcal{N}_F^{(0)}) = S(\mathcal{N}_F^{(0)})$.

For the second example, we take $\mathcal{W} = \{(a_1, b_1), (a_2, b_2)\}$ as well as $n_F(a_1, b_1) = z_1$ and $n_F(a_2, b_2) = z_2$. That is, the inputs $(a_1, b_1)$ and $(a_2, b_2)$ are transmitted noise-free. The probability transition matrix in this case can be written as
\begin{equation*}
    \mathcal{N}_F^{(1)} = \begin{pmatrix}
                              1 & 0.5 & 0.5 & 0 \\
                              0 & 0.5 & 0.5 & 1
                          \end{pmatrix}
\end{equation*}
where the rows correspond to $z \in \mathcal{Z}$, while the columns correspond to $(a, b) \in \mathcal{A} \times \mathcal{B}$. In App.~\ref{app:noise_free_subspace_MAC}, we show that the sum capacity of the MAC $\mathcal{N}_F^{(1)}$ in nats is equal to
\begin{equation*}
    S(\mathcal{N}_F^{(1)}) = 0.5 \ln(2) \approx 0.3466.
\end{equation*}
On the other hand, we show that in this example the relaxed sum capacity takes the maximum possible value $C(\mathcal{N}_F^{(1)}) = \ln(2)$ nats, thus significantly overestimating the sum capacity. For this example, we have $C(\mathcal{N}_F^{(1)}) = 2 S(\mathcal{N}_F^{(1)})$.

Using these examples, we construct the parametric family of binary MACs as a convex combination of the MACs in above examples. That is,
\begin{equation}
    \mathcal{N}_F^{(t)} = (1 - t) \mathcal{N}_F^{(0)} + t \mathcal{N}_F^{(1)} = \begin{pmatrix}
                                                                                    1 & 0.5 & 0.5 & 0.5 (1 - t) \\
                                                                                    0 & 0.5 & 0.5 & 0.5 (1 + t)
                                                                                \end{pmatrix} \label{eqn:parametric_family_NF_MACs}
\end{equation}
where $t \in [0, 1]$. Observe that for $t = 0$, we get the first example $\mathcal{N}_F^{(0)}$, whereas for $t = 1$, we get the second example $\mathcal{N}_F^{(1)}$. From the above results, we know that $C(\mathcal{N}_F^{(0)}) = S(\mathcal{N}_F^{(0)})$ whereas $C(\mathcal{N}_F^{(1)}) = 2 S(\mathcal{N}_F^{(1)})$. The MAC $\mathcal{N}_F^{(t)}$ basically interpolates between these two cases.

In Fig.~\ref{fig:sum_capacity_relaxed_sum_capacity_NF_MACs}, we plot $S(\mathcal{N}_F^{(t)})$ and $C(\mathcal{N}_F^{(t)})$ as a function of $t$. The sum capacity $S(\mathcal{N}_F^{(t)})$ is computed using Alg.~\ref{alg:compute_sum_capacity_twosenderMACs_inputsize2} with a tolerance of $\epsilon = 0.01$, whereas relaxed sum capacity $C(\mathcal{N}_F^{(t)})$ is computed using standard techniques in convex optimization. We can see from the figure that at $t = 0$ and $t = 1$, the numerically computed values agree with the analytical results. Furthermore, we observe that $C(\mathcal{N}_F^{(t)})$ becomes a progressively worse bound on $S(\mathcal{N}_F^{(t)})$ as $t$ ranges from $0$ to $1$. Thus, we demonstrate that our algorithm for computing the sum capacity does better than the relaxed sum capacity.
\begin{figure}[ht]
    \begin{center}
        \includegraphics[width=0.85\textwidth]{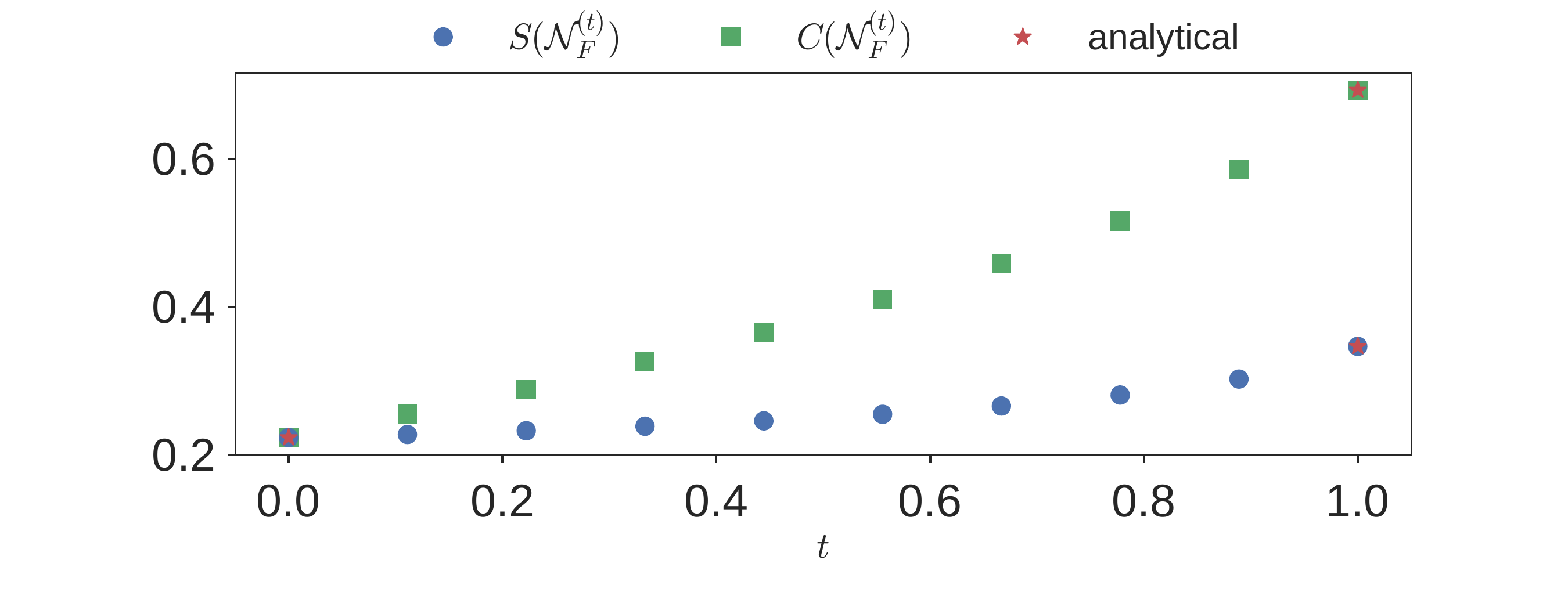}
    \end{center}
    \caption{Plot of the sum capacity $S(\mathcal{N}_F^{(t)})$ (in nats) and the relaxed sum capacity $C(\mathcal{N}_F^{(t)})$ (in nats) as a function of the tunable parameter $t$. The parameterized family of binary MACs $\mathcal{N}_F^{(t)}$ is defined in Eq.~\eqref{eqn:parametric_family_NF_MACs}. The sum capacity is computed to within a tolerance of $\epsilon = 0.01$ using Alg.~\ref{alg:compute_sum_capacity_twosenderMACs_inputsize2}. We find that the relaxed sum capacity gives a progressively worse bound on the actual sum capacity. We also find that the analytically computed values of $S(\mathcal{N}_F^{(t)})$ and $C(\mathcal{N}_F^{(t)})$ for $t = 0$ and $t = 1$ show good agreement with the numerically computed values.}
    \label{fig:sum_capacity_relaxed_sum_capacity_NF_MACs}
\end{figure}

\section{Conclusion and future directions of research}
Computing the sum capacity of a multiple access channel is a nonconvex optimization problem. For MACs obtained from nonlocal games, we obtained an analytical upper bound on the sum capacity that depends only on the number of question tuples in the game and the maximum winning probability of the game when the questions are drawn uniformly at random. Our formula is an upper bound on the achievable sum rate even when the senders of the MAC can share an arbitrary set of correlations. Using this formula, we found a separation between the sum capacity and the entanglement-assisted sum rate for the $2$-sender MAC obtained from the Magic Square game that is larger than the previously reported value. We also obtained separations in some other relevant scenarios using the CHSH game and multiparty parity game.

Furthermore, we studied the performance of the upper bound on the sum capacity obtained by relaxing the nonconvex problem to a convex optimization problem. By constructing the signalling game, we showed that one can obtain an arbitrarily large separation between the sum capacity and the relaxed sum capacity. With the help of numerical simulations, we argued that this separation holds even when the senders are allowed to share no-signalling correlations. These results indicate that the relaxed sum capacity can be a very poor upper bound on the sum capacity. In a recent work, Fawzi \& Ferm\'e~\cite{fawzi2022beating} compute the no-signalling assisted sum rate for MACs, allowing feedback. In their study, they pose the question as to whether the no-signalling assisted sum rate (allowing feedback) is equal to the relaxed sum capacity. It would be interesting to see if the MAC obtained from the signalling game can be used to answer this question in the negative.

In response to the above observations, we studied algorithms to compute the sum capacity. First, we identified that the mutual information occurring in the computation of the sum capacity satisfies a Lipschitz-like property. We subsequently proposed a few algorithms one can use to optimize such functions, by appropriately modifying and generalizing existing algorithms for optimizing Lipschitz-continuous functions. Using this, we were able to show that the sum capacity of any two-sender MAC can be computed to a fixed precision in quasi-polynomial time. Our algorithms are practically efficient for computing the sum capacity of a family of two-sender MACs that have at least one input alphabet of size two.

We remark that some other entropic quantities also satisfy a Lipschitz-like property. This is also true in the quantum setting, for example for the von Neumann entropy. Therefore, further investigation of algorithms to optimize such Lipschitz-like functions might be helpful in solving nonconvex problems in both classical and quantum information theory. In particular, it would be interesting to see if there are other practically relevant problems in information theory which would benefit from such an approach. In any case, the algorithms we present in this study for optimizing Lipschitz-like functions over the standard simplex suffer from the drawback that they are not scalable, i.e., the optimization is very costly to perform in practice as the dimension increases. Therefore, finding more practical algorithms to perform this optimization is an interesting direction for future research. In particular, the algorithm for optimization using dense curves has a scope for improvement because we only prove sub-optimal convergence guarantees.

Another avenue for performing such nonconvex optimization is to devise randomized algorithms, which might allow for faster convergence (with high probability). We remark that using randomized algorithms to optimize Lipschitz-like functions over an arbitrary compact and convex domain is not expected to give significant improvements over the deterministic algorithm we presented in this study.
This is because the number of iterations needed for convergence will scale exponentially with the dimension in both the deterministic and stochastic setting~\cite{malherbe2017global}. It is therefore important to use information about the domain (for example, standard simplex) or impose some additional restrictions on the objective functions while designing such algorithms.
It would also be interesting to see if there is a general method to determine a priori if the convex relaxation gives a good or a bad bound on the optimum of the non-convex problem.
This would allow one to use the more efficiently computable convex relaxation in relevant scenarios.

\section{Acknowledgements}
AS would like to thank Stephen Becker for helpful discussions and for suggesting to look into Lipschitz grid search algorithms for nonconvex optimization.
The authors would like to thank Mohammad Alhejji for helpful discussions and comments on the manuscript.
This work is partially funded by NSF grants CCF 1652560, PHY 1915407, QuIC-TAQS 2137984, and QuIC-TAQS 2137953.
Akshay Seshadri acknowledges support from the Professional Research Experience Program (PREP) operated jointly by NIST and the University of Colorado. 

\section{Code availablility\label{secn:code_availability}}
An open source implementation of the codes used in this study can be found at\\
\href{https://github.com/akshayseshadri/sum-capacity-computation}{https://github.com/akshayseshadri/sum-capacity-computation}.

This repository contains code to perform the following tasks: 1)~Maximize Lipschitz-like functions over the interval (Alg.~\ref{alg:maximize_Lipschitz_like_function_1D}) and standard simplex using grid search (Alg.~\ref{alg:maximize_Lipschitz_like_function_simplex_grid_search}) and dense curves (Alg.~\ref{alg:maximize_Lipschitz_like_function_simplex_dense_curve}), 2)~Compute the sum capacity of two-sender MACs by implementing the aforementioned algorithms, and 3)~Find the maximum winning probability of an $N$-player non-local game using no-signalling strategies and a corresponding optimal NS strategy.

\printbibliography[heading=bibintoc]

\appendix
\section{Bounding the correlation-assisted achievable sum rate of MACs from nonlocal games}
\begin{proof}[\textbf{Proof of Proposition \ref{prop:deterministic_strategy_maximum}}]
    \label{proof:deterministic_strategy_maximum}
    We solve the minimization problem $\min_{\pi \in \Delta_d} -\mathscr{I}_{\bm{w}}(\pi)$, which is equivalent to the given maximization problem. Note that since the entries of $\bm{w}$ are either $0$ or $1$, we can write $\mathcal{K} = \{i \in [d] \mid w_i \neq 0\}$. Then we can write the Lagrangian for this minimization problem as
\begin{align}
    \mathcal{L}(\pi; \lambda, \nu) &= -\mathscr{I}_{\bm{w}}(\pi) - \ip{\lambda, \pi} + \nu \left(\sum_{i = 1}^d \pi_i - 1\right) \nonumber \\
                                   &= \sum_{i = 1}^d p_i \ln p_i - \sum_{i \in \mathcal{K}} \pi_j w_j \ln d + \ln d - \sum_{i = 1}^d \lambda_i \pi_i + \nu \left(\sum_{i = 1}^d \pi_i - 1\right), \label{eqn:Lagrangian_deterministic}
    \intertext{where}
    p_i &= (\overline{W} \pi)_i = \pi_i w_i + \frac{1}{d} - \frac{1}{d} \sum_{j = 1}^d \pi_j w_j.
\end{align}
The expression for the probability distribution $p$, which is the output of the channel $\mathcal{N}_G$, is obtained from Eqs.~\eqref{eqn:W_matrix_elts} and \eqref{eqn:outputprob}. The variables $\lambda_1, \dotsc, \lambda_d$ are the dual variables corresponding to the inequality constraint $\pi_i \geq 0$ for all $i \in [d]$, and $\nu$ is the dual variable corresponding to the equality constraint $\sum_{i = 1}^d \pi_i = 1$. We also write $K = |\mathcal{K}|$.

\noindent\underline{Case 1}: $0 < K < d$ \newline
We can write
\begin{equation}
    p_i = \begin{cases}
              \pi_i + \frac{1}{d} - \frac{1}{d} \sum_{j \in \mathcal{K}} \pi_j & \text{for }i \in \mathcal{K} \\
              \frac{1}{d} - \frac{1}{d} \sum_{j \in \mathcal{K}} \pi_j & \text{for }i \notin \mathcal{K}.
          \end{cases} \label{eqn:outputprob_deterministic}
\end{equation}
Since $p_i \geq 0$ for all $i$, we must have
\begin{equation*}
    \sum_{j \in \mathcal{K}} \pi_j \leq 1,
\end{equation*}
Then, we can consider two sub-cases: either (a) $\sum_{j \in \mathcal{K}} \pi_j = 1$ or (b) $\sum_{j \in \mathcal{K}} \pi_j < 1$ at the optimum.

\noindent (a) If $\sum_{j \in \mathcal{K}} \pi_j = 1$ at the optimum, then from Eq.~\eqref{eqn:outputprob_deterministic}, we get
\begin{equation}
    p_i = \begin{cases}
              \pi_i & \text{for }i \in \mathcal{K} \\
               0    & \text{for }i \notin \mathcal{K}.
          \end{cases} \label{eqn:outputprob_deterministic_case1a}
\end{equation}
Furthermore, since $\sum_{j \in \mathcal{K}} \pi_j = 1$, we also have $\pi_i = 0$ for $i \notin \mathcal{K}$. It follows from Eq.~\eqref{eqn:mutualinfo} and Eq.~\eqref{eqn:outputprob_deterministic_case1a} that $\mathscr{I}_{w}(\pi) = H(p)$, where $p$ is as given in Eq.~\eqref{eqn:outputprob_deterministic_case1a}. Thus, $p$ can be taken as a probability distribution on the indices $\mathcal{K}$, and subsequently, we obtain
\begin{equation}
    \max_\pi \mathscr{I}_{w}(\pi) = \ln K. \label{eqn:deterministic_maximization_casea}
\end{equation}

\noindent (b) If $\sum_{j \in \mathcal{K}} \pi_j < 1$ at the optimum, we can infer from Eq.~\eqref{eqn:outputprob_deterministic} that $p_i > 0$ for all $i \in [d]$ at the optimum. Thus, the entropy $H(p)$ is differentiable at the optimum, and consequently, we can differentiate the Lagrangian given in Eq.~\eqref{eqn:Lagrangian_deterministic}. The gradient of the Lagrangian is given as
\begin{equation*}
    \frac{\partial \mathcal{L}}{\partial \pi_j} = \begin{cases}
                                                       \ln p_j - \frac{1}{d} \sum_{i = 1}^d \ln p_i - \ln d - \lambda_j + \nu & \text{for }j \in \mathcal{K} \\
                                                       -\lambda_j + \nu & \text{for }j \notin \mathcal{K}.
                                                  \end{cases}
\end{equation*}
Note that we are solving a convex optimization problem and Slater's condition holds because the constraint set is a simplex~\cite{boyd2004convex}. Therefore, the KKT conditions are necessary and sufficient for optimality~(see Ch.~$5.5$ in Ref.~\cite{boyd2004convex}). Subsequently, the optimal distribution $\pi$ satisfies~\cite{boyd2004convex}
\begin{itemize}[label={}]
    \item (Primal feasibility) $\pi \in \Delta_d$,
    \item (Dual feasibility) $\lambda_i \geq 0$ for all $i \in [d]$,
    \item (Complementary slackness) $\lambda_i \pi_i = 0$ for all $i \in [d]$, and
    \item (Stationarity) $\nabla_\pi \mathcal{L} = 0$.
\end{itemize}

\noindent For $j \notin \mathcal{K}$, the condition $\nabla_\pi \mathcal{L} = 0$ gives
\begin{equation}
    \lambda_j = \nu. \label{eqn:Lagrangian_deterministic_stationarity_outisdeK}
\end{equation}
Since $\sum_{j \in \mathcal{K}} \pi_j < 1$ by assumption and $\pi \in \Delta_d$, we must have $\pi_{j'} > 0$ for some $j' \notin \mathcal{K}$. Then, by complementary slackness, we obtain $\lambda_{j'} = 0$. Using this in Eq.~\eqref{eqn:Lagrangian_deterministic_stationarity_outisdeK}, we get $\nu = 0$.

\noindent On the other hand, for $j \in \mathcal{K}$ the condition $\nabla_\pi \mathcal{L} = 0$ gives
\begin{equation*}
    \ln p_j - \frac{1}{d} \sum_{i = 1}^d \ln p_i = \lambda_j + \ln d,
\end{equation*}
where we used the fact that $\nu = 0$. To simplify this equation further, we make the observation that
\begin{equation*}
    p_i = \begin{cases}
                \pi_i + \frac{p_L}{d} & \text{for }i \in \mathcal{K} \\
                \frac{p_L}{d} & \text{for }i \notin \mathcal{K}
          \end{cases}
\end{equation*}
with $p_L > 0$, which follows from Eq.~\eqref{eqn:losingprob} and Eq.~\eqref{eqn:outputprob_deterministic}. Therefore, 
\begin{align*}
	\sum_{i = 1}^d \ln p_i = \sum_{i \in \mathcal{K}} \ln p_i + (d - K) \ln(p_L/d),
\end{align*} 
where $d - K = |[d] \setminus \mathcal{K}|$. Thus, for $j \in \mathcal{K}$, we obtain
\begin{equation*}
    \ln p_j - \frac{1}{d} \sum_{i \in \mathcal{K}} \ln p_i = \lambda_j + \ln d + \left(1 - \frac{K}{d}\right) \ln \frac{p_L}{d}.
\end{equation*}
If we label the indices in $\mathcal{K}$ as $j_1, \dotsc, j_K$, we can write the above as the following matrix equation:
\begin{equation}
    \begin{pmatrix}
        1 - \frac{1}{d} & -\frac{1}{d} & \dotsb & -\frac{1}{d} \\
            \vdots      &              & \dotsb & \vdots \\
           -\frac{1}{d} & -\frac{1}{d} & \dotsb & 1 - \frac{1}{d}
    \end{pmatrix}
    \begin{pmatrix} \ln p_{j_1} \\ \vdots \\ \ln p_{j_{K}} \end{pmatrix}
    = \begin{pmatrix}
        \lambda_{j_1} + \ln d + \left(1 - \frac{K}{d}\right) \ln \frac{p_L}{d} \\
          \vdots \\
          \lambda_{j_K} + \ln d + \left(1 - \frac{K}{d}\right) \ln \frac{p_L}{d}
      \end{pmatrix} \label{eqn:outputprob_deterministic_case2a}
\end{equation}
Let $o = \begin{pmatrix} 1 & \dotsb & 1 \end{pmatrix}^T$ denote the vector of ones. Then, the matrix appearing in the LHS of \eqref{eqn:outputprob_deterministic_case2a} can be expressed as $\id - o o^T / d$. By the matrix determinant lemma (see Cor.~(18.1.3) in Ref.~\cite{harville1998matrix}), we know that this matrix is invertible iff $1 - o^T o / d = 1 - K / d \neq 0$, which is always true by our assumption that $K < d$. Its inverse is given by the Sherman-Morrison formula~\cite{sherman1950adjustment, bartlett1951inverse} as follows:
\begin{equation*}
    \left(\id - \frac{1}{d} o o^T\right)^{-1} = \id + \frac{1}{d - K} oo^T.
\end{equation*}

\noindent Using this inverse in Eq.~\eqref{eqn:outputprob_deterministic_case2a}, for $j \in \mathcal{K}$, we obtain
\begin{align*}
    p_j &= E_j \frac{p_L}{d},
    \intertext{where}
    E_j &= \exp\left(\lambda_j + \frac{1}{d - K} \sum_{i \in \mathcal{K}} \lambda_i + \frac{d}{d - K} \ln d\right).
\end{align*}
Since $\lambda_i \geq 0$ for all $i \in [d]$ due to dual feasibility, we have $E_j > 1$ for all $j \in \mathcal{K}$.

\noindent Using $p_j = \pi_j + p_L / d$ for $j \in \mathcal{K}$, we obtain
\begin{equation}
    \pi_j = (E_j - 1) \frac{p_L}{d}. \label{eqn:deterministic_maximization_pij_Ej}
\end{equation}
Because $E_j > 1$, we have $\pi_j > 0$ for each $j \in \mathcal{K}$. Then, by complementary slackness, we have $\lambda_j = 0$ for all $j \in \mathcal{K}$. Therefore, we obtain
\begin{equation}
    E_j = d^{\frac{d}{d - K}}\ \forall j \in \mathcal{K}. \label{eqn:deterministic_maximization_variableEj}
\end{equation}

\noindent Now, we will solve for $\pi_j$ for $j \in \mathcal{K}$. For this purpose, note that
\begin{equation*}
    e_j^{-1} = \frac{d}{E_j - 1}
\end{equation*}
is a positive number for all $j \in \mathcal{K}$. Then, using $p_L = 1 - \sum_{i \in \mathcal{K}} \pi_i$ and $\pi_j = e_j p_L$ for $j \in \mathcal{K}$ (see Eq.~\eqref{eqn:deterministic_maximization_pij_Ej}), we can write
\begin{align*}
    e_j^{-1} \pi_j + \sum_{i \in \mathcal{K}} \pi_i &= 1,
\end{align*}
which can be written in matrix form as
\begin{align*}
    \begin{pmatrix}
        e_{j_1}^{-1} + 1 & 1 & \dotsb & 1 \\
        \vdots           &   & \dotsb & \vdots \\
        1                & 1 & \dotsb & e_{j_{K}}^{-1} + 1
    \end{pmatrix}
    &\begin{pmatrix}
        \pi_{j_1} \\
        \vdots \\
        \pi_{j_{K}}
    \end{pmatrix}
    = \begin{pmatrix}
           1 \\
           \vdots \\
           1
       \end{pmatrix}.
\end{align*}
The matrix appearing in the LHS of the above equation can be written as $\text{diag}(e_j^{-1}) + o o^T$, which is invertible iff $1 + \sum_{j \in \mathcal{K}} e_j \neq 0$. This condition always holds because $e_j > 0$ for each $j \in \mathcal{K}$. Denoting $A = \text{diag}(e_j)$, we can use the Sherman-Morrison formula~\cite{sherman1950adjustment, bartlett1951inverse} to write
\begin{equation}
    \begin{pmatrix}
        \pi_{j_1} \\
        \vdots \\
        \pi_{j_{K}}
    \end{pmatrix}
    = A o - \frac{o^T A o}{1 + o^T A o} A o
    = \begin{pmatrix}
           \frac{e_{j_1}}{1 + \sum_{j \in \mathcal{K}} e_j} \\
           \vdots \\
           \frac{e_{j_{K}}}{1 + \sum_{j \in \mathcal{K}} e_j}
       \end{pmatrix}. \label{eqn:deterministic_maximization_optimum_pi_e}
\end{equation}

\noindent Then, using $e_j = (E_j - 1) / d$ for $j \in \mathcal{K}$ along with Eq.~\eqref{eqn:deterministic_maximization_variableEj} and Eq.~\eqref{eqn:deterministic_maximization_optimum_pi_e}, we obtain
\begin{align*}
    \pi_j &= \frac{d^{\frac{d}{d - K}} - 1}{d + K(d^{\frac{d}{d - K}} - 1)}\ \forall j \in \mathcal{K} \\
    \intertext{and}
    p_L &= \frac{d}{d + K(d^{\frac{d}{d - K}} - 1)}.
\end{align*}
Subsequently, we obtain
\begin{align}
    \mathscr{I}^*_K \equiv \max_{\pi \in \Delta_d} \mathscr{I}_{w}(\pi) &= -\frac{K d^{\frac{d}{d - K}}}{d + K(d^{\frac{d}{d - K}} - 1)} \ln\left(\frac{d^{\frac{d}{d - K}}}{d + K(d^{\frac{d}{d - K}} - 1)}\right) \nonumber \\
                                                       &\hspace{0.3cm}- \frac{(d - K)}{d + K(d^{\frac{d}{d - K}} - 1)} \ln\left(\frac{1}{d + K(d^{\frac{d}{d - K}} - 1)}\right) - \frac{d}{d + K(d^{\frac{d}{d - K}} - 1)} \ln d.
                                                                                                                                                                            \label{eqn:deterministic_maximization_caseb}
\end{align}

\noindent While the expression for $\mathscr{I}^*_K$ looks complicated, it can be greatly simplified. Denoting
\begin{equation*}
    \delta_K = \frac{d - K}{d^{\frac{d}{d - K}}},
\end{equation*}
one can rearrange terms in $\mathscr{I}^*_K$ to obtain
\begin{equation}
    \mathscr{I}^*_K = \ln\left(K + \delta_K\right). \label{eqn:deterministic_maximization_caseb_simplified}
\end{equation}

\noindent From Eq.~\eqref{eqn:deterministic_maximization_casea} and Eq.~\eqref{eqn:deterministic_maximization_caseb} and the fact that $\delta_K > 0$ for $0 < K < d$, we can infer that
\begin{equation*}
    \max_{\pi \in \Delta_d} \mathscr{I}_{\bm{w}}(\pi) = \mathscr{I}^*_K.
\end{equation*}

\noindent\underline{Case 2}: $K = 0$ or $K = d$

$K = 0$ implies $w_i = 0$ for all $i \in [d]$. Thus, $\ip{\pi, \bm{w}} = 0$ for all $\pi \in \Delta_d$, i.e., we will always lose the game. In that case, it can be verified that we must have $\mathscr{I}_{\bm{w}}(\pi) = 0$ for all $\pi \in \Delta_d$.

On the other hand, $K = d$ implies $w_i = 1$ for all $i \in [d]$, meaning that the strategy is perfect. In particular, we have $\ip{\pi, \bm{w}} = 1$ for all $\pi \in \Delta_d$, i.e., $p_L = 0$ for any distribution on questions. It can be verified that $\mathscr{I}_{\bm{w}}(\pi_U) = \ln d$ for the uniform distribution $\pi_U \in \Delta_d$.
\end{proof}

\begin{proof}[\textbf{Proof of Proposition \ref{prop:arbitrary_strategy_maximum_nearperfectgame}}]
    \label{proof:arbitrary_strategy_maximum_nearperfectgame}
    1) Let $\bm{w}^{(D)} \in \{0, 1\}^d$ denote a binary deterministic strategy that can answer $d - 1$ of the $d$ questions correctly. Then, there is exactly one index $k \in [d]$ such that $w^{(D)}_k = 0$. We show that we can achieve $\max_\pi \mathscr{I}_{\bm{w}}(\pi)$ using the deterministic strategy $\mathscr{I}_{\bm{w}^{(D)}}(\widetilde{\pi})$ for some appropriate distribution $\widetilde{\pi} \in \Delta_d$. Note that the winning vector $\bm{w}^{(D)}$ may not be permitted by the game, but it still gives an upper bound on $\max_\pi \mathscr{I}_{\bm{w}}(\pi)$.

    Since by assumption $K^* < d$, we can find an injective function $\sigma\colon \mathcal{K}^* \to [d] \setminus \{k\}$. The map $\sigma$ will be used to label the indices. We construct the distribution $\widetilde{\pi} \in \Delta_d$ as follows:
    \begin{align}
        \widetilde{\pi}_i &= w_{\sigma^{-1}(i)} \pi^*_{\sigma^{-1}(i)}, &&\text{for }i \in \sigma(\mathcal{K}^*) \nonumber \\
        \widetilde{\pi}_i &= 0, &&\text{for }i \notin \sigma(\mathcal{K}^*) \setminus \{k\} \nonumber \\
        \widetilde{\pi}_k &= 1 - \sum_{i \in \sigma(\mathcal{K}^*)} \widetilde{\pi}_i.
    \end{align}
    Since $\sigma$ is an injective function on $\mathcal{K}^*$, it is invertible when its co-domain is restricted to its range $\sigma(\mathcal{K}^*)$. Thus, $\sigma^{-1}(i)$ is well-defined for $i \in \sigma(\mathcal{K}^*)$.

    \noindent Since $w^{(D)}_j = 1$ for $j \in [d] \setminus \{k\}$ and $w^{(D)}_k = 0$, we have
    \begin{equation}
        \ip{\bm{w}^{(D)}, \widetilde{\pi}} = \sum_{i \in \sigma(\mathcal{K}^*)} w_{\sigma^{-1}(i)} \pi^*_{\sigma^{-1}(i)}
                                           = \sum_{j \in \mathcal{K}^*} w_j \pi^*_j
                                           = \ip{\bm{w}, \pi^*}. \label{eqn:strategy_improvement_winningprob}
    \end{equation}
    That is, the winning probabilities obtained from $\bm{w}$, $\pi^*$ and $\bm{w}^{(D)}$, $\widetilde{\pi}$ are the same.

    Next, we will look at the output probability of the channel in the two cases. For $\bm{w}$, $\pi^*$, the output probability is given as (see Eq.~\eqref{eqn:outputprob})
    \begin{align*}
        p^*_i &= \begin{cases}
                     w_i \pi^*_i + \frac{1}{d} - \frac{1}{d} \sum_{j \in \mathcal{K}^*} w_j \pi^*_j & \text{for }i \in \mathcal{K}^* \\[0.4cm]
                     \frac{1}{d} - \frac{1}{d} \sum_{j \in \mathcal{K}^*} w_j \pi^*_j & \text{for }i \notin \mathcal{K}^*,
                  \end{cases}
        \intertext{whereas for $w^{(D)}$, $\widetilde{\pi}$, the output probability is given as}
        \widetilde{p}_i &= \begin{cases}
                               w_{\sigma^{-1}(i)} \pi^*_{\sigma^{-1}(i)} + \frac{1}{d} - \frac{1}{d} \sum_{j \in \mathcal{K}^*} w_j \pi^*_j & \text{for }i \in \sigma(\mathcal{K}^*) \\[0.4cm]
                               \frac{1}{d} - \frac{1}{d} \sum_{j \in \mathcal{K}^*} w_j \pi^*_j & \text{for }i \notin \sigma(\mathcal{K}^*).
                           \end{cases}
    \end{align*}
    Here we used Eq.~\eqref{eqn:strategy_improvement_winningprob} to obtain the expression for $\widetilde{p}$. Since $|\mathcal{K}^*| = |\sigma(\mathcal{K}^*)|$, the probability $\widetilde{p}$ is just a permutation of $p^*$. Furthermore, because Shannon entropy is invariant under permutation of the entries of the probability distribution, we have $H(\widetilde{p}) = H(p^*)$. Then from Eq.~\eqref{eqn:outputprob}, Eq.~\eqref{eqn:mutualinfo}, and Eq.~\eqref{eqn:strategy_improvement_winningprob}, we have
    \begin{equation*}
        \max_{\pi \in \Delta_d} \mathscr{I}_{\bm{w}}(\pi) = \mathscr{I}_{\bm{w}}(\pi^*) = \mathscr{I}_{\bm{w}^{(D)}}(\widetilde{\pi})
                                                                                        \leq \max_{\pi \in \Delta_d} \mathscr{I}_{\bm{w}^{(D)}}(\pi)
                                                                                        = \mathscr{I}^*_{d - 1}
    \end{equation*}
    where $\mathscr{I}^*_{d - 1}$ is given by Eq.\eqref{eqn:deterministic_strategy_maximum_mutualinfo}.

    2) We focus on the case where $K^* = K = d$. This implies $w_i \neq 0$ and $\pi^*_i \neq 0$ for all $i \in [d]$. Our goal is to solve the optimization problem $\max_{\pi \in \Delta_d} \mathscr{I}_{\bm{w}}(\pi)$. Following Eq.~\eqref{eqn:Lagrangian_deterministic}, we write the Lagrangian for the problem
    \begin{align}
        \mathcal{L}(\pi; \lambda, \nu) &= \sum_{i = 1}^d p_i \ln p_i - \sum_{i = 1}^d \pi_j w_j \ln d + \ln d - \sum_{i = 1}^d \lambda_i \pi_i + \nu \left(\sum_{i = 1}^d \pi_i - 1\right),
        \intertext{where}
        p_i &= \pi_i w_i + \frac{1}{d} - \frac{1}{d} \sum_{j = 1}^d \pi_j w_j. \label{eqn:arbitrary_strategy_prob}
    \end{align}
    Since $K^* = d$, the probabilities $p_i$ are non-zero at the optimum, and hence we can differentiate the entropy at the optimum.

    \noindent Note that we are solving a convex optimization problem and Slater's condition holds~\cite{boyd2004convex}. Therefore, KKT conditions are necessary and sufficient for optimality~\cite{boyd2004convex}.
    From KKT conditions, it follows that $\nabla_\pi \mathcal{L} = 0$, which gives
    \begin{equation*}
        \ln p_j - \frac{1}{d} \sum_{i = 1}^d \ln p_i = \frac{\lambda_j}{w_j} - \frac{\nu}{w_j} + \ln d, \quad j \in [d].
    \end{equation*}
    Since by assumption $\pi_i \neq 0$ at the optimum for all $i \in [d]$, we have $\lambda_i = 0$ by complementary slackness. This gives
    \begin{equation}
        p_j = d \left(\prod_{i = 1}^d p_i\right)^{1/d} \exp\left(-\frac{\nu}{w_j}\right)\label{eqn:arbitrary_strategy_prob_eqn}
    \end{equation}
    for $j \in [d]$. By multiplying Eq.~\eqref{eqn:arbitrary_strategy_prob_eqn} for $j = 1, \dotsc, d$, we can infer that
    \begin{equation}
        \nu = d w_{\text{eff}} \ln d \label{eqn:arbitrary_strategy_prob_eqn_nu}
    \end{equation}
    since $p_i \neq 0$ at the optimum. Next, summing over $j \in [d]$ in Eq.~\eqref{eqn:arbitrary_strategy_prob_eqn}, we obtain
    \begin{equation}
        d \left(\prod_{i = 1}^d p_i\right)^{1/d} = \frac{1}{\sum_{j = 1}^d \exp\left(-\frac{\nu}{w_j}\right)}.
    \end{equation}

    \noindent Combining this with Eq.~\eqref{eqn:arbitrary_strategy_prob_eqn} and Eq.~\eqref{eqn:arbitrary_strategy_prob_eqn_nu}, we obtain
    \begin{equation}
        p_i = \frac{\exp\left(-\frac{d w_{\text{eff}} \ln d}{w_i}\right)}{\sum_{j = 1}^d \exp\left(-\frac{d w_{\text{eff}} \ln d}{w_j}\right)}. \label{eqn:boltzmannprob_opt}
    \end{equation}
    Therefore, at the optimum, we have a Boltzmann distribution for the probability of outcomes of the channel.

    Now, instead of solving for $\pi$ from Eq.~\eqref{eqn:boltzmannprob_opt}, we find a suitable expression for the objective function $\mathscr{I}_w(\pi)$ in terms of $p$ and $w$ (instead of $\pi$ and $w$). To this end, subtracting Eq.~\eqref{eqn:arbitrary_strategy_prob} corresponding to two different indices $i$ and $j$ and summing over the index $i$, we obtain
    \begin{equation*}
        \sum_{i = 1}^d \pi_i w_i = 1 + d(\pi_j w_j - p_j).
    \end{equation*}
    Dividing this equation by $w_j$ and summing over $j$, we get
    \begin{align*}
        \sum_{j = 1}^d \frac{1}{w_j} \sum_{i = 1}^d \pi_i w_i &= \sum_{j = 1}^d \frac{1}{w_j} + d\left(1 - \sum_{j = 1}^d \frac{p_j}{w_j}\right) \\
        \sum_{i = 1}^d \pi_i w_i = 1 &- d w_{\text{eff}} \left(\sum_{j = 1}^d \frac{p_j}{w_j} - 1\right),
    \end{align*}
    where $w_{\text{eff}}$ is as defined in Eq.~\eqref{eqn:weff}. Substituting this in Eq.~\eqref{eqn:mutualinfo}, we find that the mutual information $\mathscr{I}_{\bm{w}}(\pi)$ can be written as a function of the outcome probability $p$ as follows:
    \begin{equation}
        \mathscr{I}_{\bm{w}}(\pi) = H(p) - d w_{\text{eff}} \left(\sum_{j = 1}^d \frac{p_j}{w_j} - 1\right) \ln d. \label{eqn:arbit_strategy_mutualinfo_outputprob}
    \end{equation}

    To obtain the value of $\mathscr{I}_{\bm{w}}(\pi)$ at the optimum, we substitute Eq.~\eqref{eqn:boltzmannprob_opt} in Eq.~\eqref{eqn:arbit_strategy_mutualinfo_outputprob}. Denoting $\max_{\pi \in \Delta_d} \mathscr{I}_{\bm{w}}(\pi) = \mathscr{I}^*(\bm{w})$, one can rearrange terms to obtain
    \begin{align}
        \mathscr{I}^*(\bm{w}) &= d w_{\text{eff}} \ln d + \ln\left(\sum_{j = 1}^d \exp\left(-\frac{d w_{\text{eff}} \ln d}{w_j}\right)\right) \\
                              &= \ln\left(\sum_{j = 1}^d \exp\left[d w_{\textnormal{eff}} \ln d \left(1 - \frac{1}{w_j}\right)\right]\right),
    \end{align}
    where we get the last equation by noting that $d w_{\text{eff}} \ln d = \ln \exp(d w_{\text{eff}} \ln d)$.
\end{proof}

\begin{proof}[\textbf{Proof of Proposition~\ref{prop:strategy_coeff_maximization_upper_bound}}]
    \label{proof:strategy_coeff_maximization_upper_bound}
    To obtain an upper bound on
    \begin{equation*}
        \sup_{\bm{w} \in \overline{\winvec_\corr}, \bm{w} > 0} \mathscr{I}^*(\bm{w})
    \end{equation*}
    given in Eq.~\eqref{eqn:strategy_coeeff_maximization_upper_bound}, we solve a relaxation of this optimization problem. For this purpose, note that the set over which we optimize can be written as
    \begin{equation*}
        \overline{\winvec_\corr}_> = \left\{\bm{w} \in [0, 1]^d \mid w > 0,\ \frac{\sum_{i = 1}^d w_i}{d} \leq \omega^{\mathcal{C}}(G)\right\}.
    \end{equation*}

    Let $\HM(w_1, \dotsc, w_d)$ be the harmonic mean and $\AM(w_1, \dotsc, w_d)$ be the arithmetic mean of $w_1, \dotsc, w_d$, respectively. Observe that
    \begin{equation*}
        d w_{\text{eff}} = \HM(w_1, \dotsc, w_d).
    \end{equation*}
    Then, since $\HM(w_1, \dotsc, w_d) \leq \AM(w_1, \dotsc, w_d)$ for $w_1, \dotsc, w_d > 0$ (with equality iff $w_1 = \dotsb = w_d$), we have
    \begin{equation}
        d w_{\text{eff}} \leq \frac{\sum_{i = 1}^d w_i}{d} \leq \omega^{\strat_\corr}(G). \label{eqn:AM_HM_ineq}
    \end{equation}
    
    \noindent With this in mind, we define the set
    \begin{equation*}
        \overline{\winvec_\corr}_{\text{eff}} = \left\{w \in [0, 1]^d \mid w > 0,\ d w_{\text{eff}} \leq \omega^{\strat_\corr}(G)\right\}.
    \end{equation*}
    Then, from Eq.~\eqref{eqn:AM_HM_ineq}, it follows that $\overline{\winvec_\corr}_> \subseteq \overline{\winvec_\corr}_{\text{eff}}$.
    Therefore, we have
    \begin{equation*}
        \sup_{\bm{w} \in \overline{\winvec_\corr}_>} \mathscr{I}^*(\bm{w}) \leq \sup_{\bm{w} \in \overline{\winvec_\corr}_{\text{eff}}} \mathscr{I}^*(\bm{w}).
    \end{equation*}
    
    \noindent Subsequently, we solve the optimization problem $\sup_{\bm{w} \in \overline{\winvec_\corr}_{\text{eff}}} \mathscr{I}^*(\bm{w})$. To that end, we make the change of variables
    \begin{equation*}
        t_i = \frac{1}{w_i}
    \end{equation*}
    for $i \in [d]$. The optimization problem then becomes
    \begin{align}
        \sup \quad &\ln\left(\sum_{i = 1}^d \exp\left[\frac{d}{\sum_{j = 1}^d t_j} \ln d\ \left(1 - t_i\right)\right]\right) \nonumber \\
        \text{s.t.} \quad &\bm{t} \geq 1 \nonumber \\
                          &\frac{\sum_{i = 1}^d t_i}{d} \geq \frac{1}{\omega^{\strat_\corr}(G)}. \label{eqn:strategy_coeff_maximization_relaxed_eff}
    \end{align}

    \noindent We solve Eq.~\eqref{eqn:strategy_coeff_maximization_relaxed_eff} by splitting it into two maximizations. To that end, define
    \begin{align}
        s &= \frac{\sum_{i = 1}^d t_i}{d},
        \intertext{so that the objective of the optimization in Eq.~\eqref{eqn:strategy_coeff_maximization_relaxed_eff} can be written as}
        f(\bm{t}, s) &= f(\bm{t}, s) = \text{LSE}\left(-\frac{\ln d}{s} t_1, \dotsc, -\frac{\ln d}{s} t_d\right) + \frac{\ln d}{s}, \label{eqn:strategy_coeff_objective_t_s}
    \end{align}
    where $\text{LSE}(y_1, \dotsc, y_d) = \ln(\sum_{j = 1}^d e^{y_j})$ is the log-sum-exp function. Since LSE is a convex function, we can conclude that $f(\bm{t}, s)$ is convex in $\bm{t}$ for a fixed $s$. Then, Eq.~\eqref{eqn:strategy_coeff_maximization_relaxed_eff} can be expressed as
    \begin{equation*}
        \sup_{\bm{t} \in \overline{\winvec_\corr}_{\text{eff}}} f(\bm{t}) = \sup_{s \geq (\omega^{\mathcal{C}}(G))^{-1}} \max_{\bm{t} \in \overline{\winvec_\corr}_{\text{eff}}^{(s)}} f(\bm{t}, s),
    \end{equation*}
    where
    \begin{equation*}
        \overline{\winvec_\corr}_{\text{eff}}^{(s)} = \left\{\bm{t} \in \mathbb{R}^d \mid \bm{t} \geq 1,\ \frac{\sum_{i = 1}^d t_i}{d} = s\right\}
    \end{equation*}
    is a translated and scaled simplex. The extreme points of $\overline{\winvec_\corr}_{\text{eff}}^{(s)}$ are given as 
    \begin{equation}
        \bm{t}^{(1)}_s = \begin{pmatrix} 1, & \dotsc, & 1, & 1 + (s - 1) d) \end{pmatrix}^T, \label{eqn:strategy_coeff_set_relaxation_eff_s_extremept}
    \end{equation}
    and its permutations are denoted by $\bm{t}^{(2)}_s, \dotsc, \bm{t}^{(d)}_s$. Note that the constraints in Eq.~\eqref{eqn:strategy_coeff_maximization_relaxed_eff} imply $s \geq 1$.

    Now, we are seeking to maximize the convex function $f(\bm{t}, s)$ over the set $\overline{\winvec_\corr}_{\text{eff}}^{(s)}$ (for a fixed $s$). However, since $f$ is convex and any point $\bm{t} \in \overline{\winvec_\corr}_{\text{eff}}^{(s)}$ can be written as $\bm{t} = \sum_{i = 1}^d \lambda_i \bm{t}^{(i)}_s$, we have
    \begin{equation*}
        f(\bm{t}, s) \leq \sum_{i = 1}^d \lambda_i f(\bm{t}^{(i)}_s, s).
    \end{equation*}
    Then, because $f(\bm{t}, s)$ is invariant under the permutation of the components of $\bm{t}$, we have $f(\bm{t}^{(i)}_s, s) =  f(\bm{t}^{(j)}_s, s)$ for any $i,j\in[d]$. Therefore,
    \begin{equation*}
        \max_{\bm{t} \in \overline{\winvec_\corr}_{\text{eff}}^{(s)}} f(\bm{t}, s) = f(\bm{t}^{(1)}_s, s) = \ln\left(d - 1 + \exp\left[-d \ln d \left(\frac{s - 1}{s}\right)\right]\right)
    \end{equation*}
    where we substituted Eq.~\eqref{eqn:strategy_coeff_set_relaxation_eff_s_extremept} in Eq.~\eqref{eqn:strategy_coeff_objective_t_s} in the last equation. Then, since $f(\bm{t}^{(1)}, s)$ is a decreasing function of $s$, we can infer that
    \begin{equation*}
        \sup_{s \geq (\omega^{\strat_\corr}(G))^{-1}} \max_{\bm{t} \in \overline{\winvec_\corr}_{\text{eff}}^{(s)}} f(\bm{t}, s) = f\left(\bm{t}^{(1)}_s, \frac{1}{\omega^{\strat_\corr}(G)}\right)
        = \ln\left(d - 1 + d^{-(1 - \omega^{\strat_\corr}(G)) d}\right),
    \end{equation*}
    giving the desired bound.
\end{proof}

\section{Some notes on strategies for nonlocal games\label{app:strategy_nonlocal_games}}
We begin by characterizing the extreme points of the set of conditional distributions (over finite sets).
\begin{proposition}
    \label{prop:deterministic_strategy_extreme_point}
    Let $\mathcal{X}$ and $\mathcal{Y}$ be finite sets, and let $\mathcal{C}$ denote the set of conditional probability distributions on $\mathcal{Y}$ given $\mathcal{X}$. Then, the extreme points of $\mathcal{C}$ correspond to conditional probabilities $p^f_{\bm{Y} | \bm{X}}(\bm{y} | \bm{x}) = \delta_{\bm{y}, f(\bm{x})}$ obtained from functions $f\colon \mathcal{X} \to \mathcal{Y}$. In particular, any conditional probability distribution $p_{\bm{Y} | \bm{X}} \in \mathcal{C}$ can be written as a convex combination of $p^f_{\bm{Y} | \bm{X}}$ corresponding to functions $f$.
\end{proposition}
\begin{proof}
    For convenience, let us denote $|\mathcal{X}| = m$ and $|\mathcal{Y}| = n$ for some $m, n \in \mathbb{N}$. We also fix a labelling of elements $\mathcal{X} = \{x_1, \dotsc, x_m\}$ and $\mathcal{Y} = \{y_1, \dotsc, y_n\}$. The set of conditional probability distributions $\mathcal{C}$ can be thought of as the set of functions from $\mathcal{X} \to \Delta_n$. That is, for each $\bm{x} \in \mathcal{X}$, we have a probability distribution over $\mathcal{Y}$. Thus, we identify $\mathcal{C}$ with the set $(\Delta_n)^m$.

    Since $\Delta_n$ is a simplex in $\mathbb{R}^n$, its extreme points are $\bm{e}_1, \dotsc, \bm{e}_n$, where $\bm{e}_i$ is the standard Euclidean basis vector for the $i$th coordinate. Therefore, the extreme points of $\mathcal{C} = (\Delta_n)^m$ are $(\bm{e}_{i_1}, \dotsc, \bm{e}_{i_m})$ for $i_1, \dotsc, i_m \in [n]$. We argue that such extreme points can be obtained from functions $f\colon \mathcal{X} \to \mathcal{Y}$.

    To that end, given an extreme point $(\bm{e}_{i_1}, \dotsc, \bm{e}_{i_m})$ of $\mathcal{C}$, construct the function
    \begin{equation*}
        f(\bm{x}_j) = \bm{y}_{i_j} \qquad j \in [m].
    \end{equation*}
    Observe that the conditional probability distribution determined by the function $f$ is
    \begin{equation*}
        p^f_{\bm{Y} | \bm{X}}(\cdot, \bm{x}_j) = \delta_{\cdot, f(\bm{x}_j)}
                                               = \delta_{\cdot, \bm{y}_{i_j}}
                                               = \bm{e}_{i_j}
    \end{equation*}
    for $j \in [m]$. Since there are a total of $|\mathcal{Y}|^{|\mathcal{X}|}$ extreme points of $\mathcal{C}$, and just as many functions from $\mathcal{X}$ to $\mathcal{Y}$, the above construction gives a bijective mapping between the extreme points and these functions. Since $\mathcal{C}$ is a compact and convex set with a finite number of extreme points, it is generated as the convex hull of its extreme points by the Krein-Milman theorem~\cite{conway2019course}.
\end{proof}

In the following proposition, we show that the set of no-signalling strategies is a compact and convex set. The proof of this proposition is constructive, and it can therefore be used to construct the set of no-signalling distributions numerically.
\begin{proposition}
    \label{prop:NS_strategies_convex_polytope}
    Let $\mathcal{X}_1, \dotsc, \mathcal{X}_N$ denote the question set for $N$ players, and let $\mathcal{Y}_1, \dotsc, \mathcal{Y}_N$ denote the answer set. Let $\strat_{\text{NS}}$ denote the set of no-signalling strategies used by the players. That is, $p_{\bm{Y} | \bm{X}} \in \strat_{\text{NS}}$ iff
    \begin{equation}
        p_{Y_i | \bm{X}}(y_i | x_1, \dotsc, x_N) = p_{Y_i | X_i}(y_i | x_i)\ \forall x_k \in \mathcal{X}_k,\ k \in [N] \setminus \{i\} \label{eqn:no_signalling_strategy_rewritten}
    \end{equation}
    for all $y_i \in \mathcal{Y}_i$, $x_i \in \mathcal{X}_i$, $i \in [N]$. Then, $\strat_{\text{NS}}$ is a compact and convex set. Specifically, $\strat_{\text{NS}}$ is a convex polytope obtained as the intersection of hyperplanes and halfspaces.
\end{proposition}
\begin{proof}
    For convenience, we denote $\mathcal{X} = \mathcal{X}_1 \times \dotsm \times \mathcal{X}_N$ as the question set and $\mathcal{Y} = \mathcal{Y}_1 \times \dotsm \times \mathcal{Y}_N$ as the answer set. The set of all strategies $p_{\bm{Y} | \bm{X}}(y_1, \dotsc, y_N | x_1, \dotsc, x_N)$ can be written as a product of simplices $(\Delta_{|\mathcal{Y}|})^{|\mathcal{X}|}$ (see Prop.~\ref{prop:deterministic_strategy_extreme_point}), and is therefore a compact set. Also note that we can write the no-signalling condition given in Eq.~\eqref{eqn:no_signalling_strategy_rewritten} as
    \begin{multline}
        \sum_{\substack{y_j \in \mathcal{Y}_j\\j \neq i}} p_{\bm{Y} | \bm{X}}(y_1, \dotsc, y_i, \dotsc, y_N | x_1, \dotsc, x_i, \dotsc, x_N) \\
         = \sum_{\substack{y_j \in \mathcal{Y}_j \\j \neq i}} p_{\bm{Y} | \bm{X}}(y_1, \dotsc, y_i, \dotsc, y_N | x_1', \dotsc, x_i, \dotsc, x_N')
            \quad \forall x_k, x_k' \in \mathcal{X}_k,\ k \in [N] \setminus \{i\} \label{eqn:no_signalling_strategy_alternate_def}
    \end{multline}
    for all $y_i \in \mathcal{Y}_i$, $x_i \in \mathcal{X}_i$, $i \in [N]$. In the above equation, the sum over $y_j \in \mathcal{Y}_j$ is a shorthand for the sum over $y_1 \in \mathcal{Y}_1, \dotsc, y_N \in \mathcal{Y}_N$.

    We will show that the set of no-signalling strategies $\strat_{\text{NS}}$ is a closed set. First, note that $\strat_{\text{NS}} \subseteq (\Delta_{|\mathcal{Y}|})^{|\mathcal{X}|}$, i.e., the set of no-signalling strategies is contained in the set of all strategies. Thus, we can write the elements of $\strat_{\text{NS}}$ as vectors $v = (v^{(1)}, \dotsc, v^{(|\mathcal{X}|)})$, where $v^{(i)} \in \Delta_{|\mathcal{Y}|}$ is a probability vector. Note that $v^{(i)}$ is a $|\mathcal{Y}|$-dimensional vector whereas $v$ is a $|\mathcal{X}||\mathcal{Y}|$-dimensional vector. Essentially, the vector $v^{(i)}$ denotes a probability distribution $p_{\bm{Y} | \bm{X}}(y | x^{(i)})$ over $\mathcal{Y}$ for a fixed $x^{(i)} \in \mathcal{X}$. Using this vectorial representation of a strategy, we will write the no-signalling condition given in Eq.~\eqref{eqn:no_signalling_strategy_alternate_def} in matrix form.

    To that end, fix an ordering for the elements of $\mathcal{Y}$. Then, we can index the elements of $v^{(i)}$ as $v^{(i)}_{(y_1, \dotsc, y_N)}$ for $(y_1, \dotsc, y_N) \in \mathcal{Y}$ corresponding to that ordering. For each $y_i \in \mathcal{Y}_i$ ($i \in [d]$), let $s_{y_i}$ denote the $|\mathcal{Y}|$-dimensional vector with $1$ at each index $(y_1', \dotsc, y_N') \in \mathcal{Y}$ with $y_i' = y_i$ and $0$ elsewhere. For example, if $\mathcal{Y}_1 = \{a, b\}$ and $\mathcal{Y}_2 = \{c, d\}$ and we write the elements of $\mathcal{Y} = \{(a, c), (a, d), (b, c), (b, d)\}$ in that order, then, $s_a = (1 , 1 , 0 , 0 )^T$, $s_c = (1, 0, 1, 0)^T$, and so forth. Observe also that, for each $k \in [|\mathcal{X}|]$,
    \begin{equation*}
        s_{y_i}^T v^{(k)} = \sum_{\substack{y_j \in \mathcal{Y}_j\\j \neq i}} p_{\bm{Y} | \bm{X}}(y_1, \dotsc, y_i, \dotsc, y_N | x^{(k)}).
    \end{equation*}

    Similarly, fix an ordering for $\mathcal{X}$. For each $x_i \in \mathcal{X}_i$ with $i \in [N]$, let $\mathcal{I}_{x_i}$ be the set that contains the indices $(x_1', \dotsc, x_N') \in \mathcal{X}$ with $x_i' = x_i$. Then, define the $|\mathcal{I}_{x_i}| \times |\mathcal{X}||\mathcal{Y}|$ matrix $S_{(x_i, y_i)}$ in block form as follows. Imagine each row of $S_{(x_i, y_i)}$ being split into $|\mathcal{X}|$ blocks of $|\mathcal{Y}|$-dimensional vectors, i.e., $\begin{pmatrix} b_1 & \dotsb & b_{|\mathcal{X}|} \end{pmatrix}$  with $b_i$ a $|\mathcal{Y}|$-dimensional (row) vector. We label the rows of $S_{(x_i, y_i)}$ with $\mathcal{I}_{x_i}$. Define the row $k \in \mathcal{I}_{x_i}$ to be the block $\begin{pmatrix} 0_{1 \times |\mathcal{Y}|} & \dotsb & s_{y_i}^T & \dotsb & 0_{1 \times |\mathcal{Y}|} \end{pmatrix}$ with $s_{y_i}^T$ in the $k$-th block.

    Let us label the blocks $v^{(i)}$ of the vector $v = (v^{(1)}, \dotsc, v^{(|\mathcal{X}|)}) \in (\Delta_{|\mathcal{Y}|})^{|\mathcal{X}|}$ as $v^{(x_1, \dotsc, x_N)}$ for $(x_1, \dotsc, x_N) \in \mathcal{X}$ using the ordering of $\mathcal{X}$ that we have fixed. Then, we have $(S_{(x_i, y_i)} v)_{(x_1, \dotsc, x_N)} = s_{y_i}^T v^{(x_1, \dotsc, x_N)}$ for $(x_1, \dotsc, x_N) \in \mathcal{I}_{x_i}$. That is, $S_{(x_i, y_i)} v$ is a $|\mathcal{I}_{x_i}|$-dimensional vector with the entries
    \begin{equation*}
        (S_{(x_i, y_i)} v)_{(x_1, \dotsc, x_N)} = \sum_{\substack{y_j \in \mathcal{Y}_j\\j \neq i}} p_{\bm{Y} | \bm{X}}(y_1, \dotsc, y_i, \dotsc, y_N | x_1, \dotsc, x_i, \dotsc, x_N)
    \end{equation*}
    for $(x_1, \dotsc, x_N) \in \mathcal{I}_{x_i}$. Since the $i$-th component of $(x_1, \dotsc, x_N) \in \mathcal{I}_{x_i}$ is fixed to be $x_i$, the no-signalling condition given in Eq.~\eqref{eqn:no_signalling_strategy_alternate_def} says that all the components of $S_{(x_i, y_i)} v$ are equal.

    Therefore, we enforce the no-signalling condition as follows. Define a $(|\mathcal{I}_{x_i}| - 1) \times |\mathcal{I}_{x_i}|$ matrix
    \begin{equation*}
        D_{x_i} = \begin{pmatrix}
                      1 & -1 & 0 & \dotsb & 0 & 0 \\
                      0 & 1 & -1 & \dotsb & 0 & 0 \\
                        &   &    & \ddots &   &  \\
                      0 & 0 & 0 & \dotsb & 1 & -1
                  \end{pmatrix},
    \end{equation*}
    and observe that if $r = \begin{pmatrix} r_1 & \dotsb & r_{|\mathcal{I}_{x_i}|}\end{pmatrix}^T$ is any $|\mathcal{I}_{x_i}|$-dimensional vector, then $D_{x_i} r$ is the $(|\mathcal{I}_{x_i}| - 1)$-dimensional vector $\begin{pmatrix} r_1 - r_2 & \dotsb & r_{|\mathcal{I}_{x_i}| - 1} - r_{|\mathcal{I}_{x_i}|}\end{pmatrix}^T$. Then, by the preceding remarks, the no-signalling condition can be written as
    \begin{equation}
        D_{x_i} S_{(x_i, y_i)} v = 0 \quad \text{ for all } x_i \in \mathcal{X}_i,\ y_i \in \mathcal{Y}_i,\ i \in [N]. \label{eqn:no_signalling_strategy_matrix_form}
    \end{equation}

    Since $D_{x_i} S_{(x_i, y_i)}$ is a $(|\mathcal{I}_{x_i}| - 1) \times |\mathcal{X}||\mathcal{Y}|$ matrix and $v$ is a $|\mathcal{X}||\mathcal{Y}|$-dimensional vector, the equation $D_{x_i} S_{(x_i, y_i)} v = 0$ encodes $(|\mathcal{I}_{x_i}| - 1)$ hyperplanes (thinking of $v$ as an arbitrary $|\mathcal{X}||\mathcal{Y}|$-dimensional vector). Therefore, the set of no-signalling strategies can be written as
    \begin{equation}
        \strat_{\text{NS}} = \{v \in (\Delta_{|\mathcal{Y}|})^{|\mathcal{X}|} \mid D_{x_i} S_{(x_i, y_i)} v = 0\ \forall x_i \in \mathcal{X}_i,\ y_i \in \mathcal{Y}_i,\ i \in [N]\}.
    \end{equation}
    This is the intersection of the compact set $(\Delta_{|\mathcal{Y}|})^{|\mathcal{X}|}$ with the hyperplanes defined by $D_{x_i} S_{(x_i, y_i)} v = 0$. Since hyperplanes are closed, the set $\strat_{\text{NS}}$ is closed as well.

    Since $(\Delta_{|\mathcal{Y}|})^{|\mathcal{X}|}$ is bounded, $\strat_{\text{NS}}$ is also bounded, so that $\strat_{\text{NS}}$ is compact. The convexity of $\strat_{\text{NS}}$ follows from the fact that $(\Delta_{|\mathcal{Y}|})^{|\mathcal{X}|}$ and hyperplanes are convex.
\end{proof}

\section{Analysis of algorithms for optimization of Lipschitz-like function\label{app:Lipschitzlike_optimization_analysis}}
In this section, we present convergence analysis for algorithms used for Lipschitz-like optimization. We also present an algorithm for performing optimization of Lipschitz-like functions over an arbitrary compact \& convex domain.

\subsection{Optimizing Lipschitz-like functions over an interval using modified Piyavskii-Shubert algorithm}
We first present a convergence analysis for modified Piyavskii-Shubert algorithm that finds the global maximum of Lipschitz-like functions over a closed interval.

\begin{proposition}
    \label{prop:Lipshitz_like_maximization_analysis}
    Let $\beta\colon \mathbb{R}_+ \to \mathbb{R}$ be a non-negative, continuous, monotonically increasing function with $\beta(0) = 0$. Let $\mathcal{D} = [a, b]$ be a closed interval, where $a, b \in \mathbb{R}$. Let $f\colon \mathcal{D} \to \mathbb{R}$ be a real-valued function satisfying
    \begin{equation*}
        |f(q) - f(q')| \leq \beta(|q - q'|)
    \end{equation*}
    Then, for each choice of tolerance $\epsilon > 0$, Alg.~\ref{alg:maximize_Lipschitz_like_function_1D} terminates in a finite number of time steps. If $f(q^*)$ denotes the output of Alg.~\ref{alg:maximize_Lipschitz_like_function_1D} corresponding to a tolerance of $\epsilon > 0$, we have $\max_{q \in \mathcal{D}} f(q) - f(q^*) \leq \epsilon$. The number of time steps required to converge with a tolerance of $\epsilon > 0$ is bounded above by $\lceil (b - a) / \delta \rceil$, where $\delta = \sup\{\delta' > 0 \mid \beta(x) < \epsilon/2\ \forall\ 0 \leq x \leq \delta'\}$.
\end{proposition}
\begin{proof}
    Let $i \in \mathbb{N}$ be a natural number. Given any $q^{(i)} \in [a, b]$, define the function
    \begin{equation}
        F_i(q) = f(q^{(i)}) + \beta(|q - q^{(i)}|) \label{eqn:bounding_function_1D}
    \end{equation}
    for $q \in [a, b]$. Since $f(q^{(i)}) - f(q) \geq -\beta(|q - q^{(i)}|)$ by assumption, we have $f(q) \leq F_i(q)$ for all $q \in [a, b]$. Given that we initialize $q^{(0)} = a$, the maximum of $F_0(q)$ occurs at $b$ because $\beta$ is a monotonically increasing function with $\beta(0) = 0$. This justifies the assignment $q^{(1)} = b$ at the start of the algorithm.

    Now, suppose that Alg.~\ref{alg:maximize_Lipschitz_like_function_1D} terminates at the $K$th time step. Then, the algorithm has generated points $q^{(0)}, \dotsc, q^{(K)} \in [a, b]$ and a point $q^* \in [a, b]$ satisfying $F(q^*) - f(q^*) \leq \epsilon$. The function $F$ and the point $q^*$ are obtained as follows. The points $q^{(0)}, \dotsc, q^{(K)}$ are sorted at the beginning of the $K$th iteration, so that $a = q^{(0)} \leq q^{(1)} \leq \dotsb \leq q^{(K)} = b$. Then, points $\overline{q}^{(i)} \in \argmax_{q \in [q^{(i)}, q^{(i + 1)}]} \min\{F_i(q), F_{i + 1}(q)\}$ are obtained through root finding for $0 \leq i \leq K - 1$. We choose $m \in \argmax_{0 \leq i \leq K - 1} F_i(\overline{q}^{(i)})$ and set $q^* = \overline{q}^{(m)}$ and $F = F_m$.

    Next, we elaborate on how $\overline{q}^{(i)} \in \argmax_{q \in [q^{(i)}, q^{(i + 1)}]} \min\{F_i(q), F_{i + 1}(q)\}$ is computed using root finding. First, consider the function $g_i(q) = F_i(q) - F_{i + 1}(q)$. From the non-negativity and monotonicity of $\beta$, it follows that $F_i(q)$ is a monotonically increasing function in the interval $q \in [q^{(i)}, q^{(i + 1)}]$, whereas $F_{i + 1}(q)$ is a monotonically decreasing function in the interval $q \in [q^{(i)}, q^{(i + 1)}]$. Therefore, $g_i(q)$ is a continuous and monotonically increasing function in the interval $q \in [q^{(i)}, q^{(i + 1)}]$, where the continuity of $g_i$ follows from that of $F_i$ and $F_{i + 1}$. Since $\beta(0) = 0$ and $f \leq F_i$ for all $0 \leq i \leq K$, we have $g_i(q^{(i)}) \leq 0$ and $g_i(q^{(i + 1)}) \geq 0$. Therefore, the function $g_i$ has a root $\overline{q}^{(i)}$ in the interval $[q^{(i)}, q^{(i + 1)}]$, and since $g_i = F_i - F_{i + 1}$, we have $F_i(\overline{q}^{(i)}) = F_{i + 1}(\overline{q}^{(i)})$. From the monotonicity properties of $F_i$ and $F_{i + 1}$, we can infer that $\overline{q}^{(i)}$ maximizes the bounding function $\min\{F_i, F_{i + 1}\}$ in the interval $[q^{(i)}, q^{(i + 1)}]$, and this maximum value is equal to $F_i(\overline{q}^{(i)}) = F_{i + 1}(\overline{q}^{(i)})$.

    Then, because $m \in \argmax_{0 \leq i \leq K - 1} F_i(\overline{q}^{(i)})$, $q^* = \overline{q}^{(m)}$ and $F = F_m$, we can infer that for all $q \in [q^{(i)}, q^{(i + 1)}]$, we have
    \begin{equation*}
        f(q) \leq \min\{F_i(q), F_{i + 1}(q)\} \leq F_i(\overline{q}^{(i)}) \leq F(q^*) \leq f(q^*) + \epsilon
    \end{equation*}
    Since the above equation holds for every $0 \leq i \leq K - 1$, and the intervals $[q^{(0)}, q^{(1)}], \dotsc, [q^{(K - 1)}, q^{(K)}]$ cover $[a, b]$, we can infer that
    \begin{equation*}
        \max_{q \in [a, b]} f(q) \leq f(q^*) + \epsilon
    \end{equation*}

    It remains to show that Alg.~\ref{alg:maximize_Lipschitz_like_function_1D} terminates in a finite number of time steps. To that end, note that $\beta(x)$ is a continuous function of $x \in \mathbb{R}_+$ with $\beta(0) = 0$. Therefore, for any given $\epsilon > 0$, we can find a $\delta > 0$ such that $\beta(x) \leq \epsilon/2$ whenever $0 \leq x \leq \delta$. Let $K \in \mathbb{N}$ denote the current time step. For $0 \leq i \leq K - 1$, let $\overline{q}^{(i)}$ denote a root of $F_i - F_{i + 1}$ in the interval $[q^{(i)}, q^{(i + 1)}]$. Let $m \in \argmax_{0 \leq i \leq K - 1} F_i(\overline{q}^{(i)})$, $q^* = \overline{q}^{(m)}$ and $F = F_m$ as before. Since $f(q^{(m)}) - f(q^*) \leq \beta(|q^* - q^{(m)}|)$, we have $F(q^*) - f(q^*) \leq 2 \beta(|q^* - q^{(m)}|)$. Therefore, when $|q^* - q^{(m)}| \leq \delta$, we have $F(q^*) - f(q^*) \leq \epsilon$. Since $F(q^*) = F_m(q^*) = F_{m + 1}(q^*)$, we can similarly infer that $F(q^*) - f(q^*) \leq \epsilon$ whenever $|q^* - q^{(m + 1)}| \leq \delta$. Therefore, the algorithm terminates if either $|q^* - q^{(m)}| \leq \delta$ or $|q^* - q^{(m + 1)}| \leq \delta$, where $q^* \in [q^{(m)}, q^{(m + 1)}]$.

    As per the procedure outlined in Alg.~\ref{alg:maximize_Lipschitz_like_function_1D}, the point $q^*$ will join the iterates $q^{(0)}, \dotsc, q^{(K)}$ at the $(K + 1)$th time step. When this new iterate $q^*$ is added, the updated intervals include $[q^{(m)}, q^*]$ and $[q^*, q^{(m + 1)}]$. Since $q^* \in [q^{(m)}, q^{(m + 1)}]$, we either have $|q^* - q^{(m)}| \leq |q^{(m + 1)} - q^{(m)}|/2$ or $|q^* - q^{(m + 1)}| \leq |q^{(m + 1)} - q^{(m)}|/2$. Since there are only a finite number of intervals at each time step and the length of one of the sides of the interval where the new iterate falls is at least halved at each time step, at some large enough time step $K$, we will have $|q^* - q^{(m)}| \leq \delta$ or $|q^* - q^{(m + 1)}| \leq \delta$. Thus, the algorithm terminates in a finite number of time steps.

    The worst case scenario corresponds to the situation where $|q^{(i)} - q^{(i + 1)}| = \delta$ for all $0 \leq i \leq K - 1$. In this case, the algorithm terminates at the $K$th time step, wherein $q^*$ falls within one of these intervals. Taking $\delta = \sup\{\delta' > 0 \mid \beta(\delta') \leq \epsilon/2\}$, there is a sequence $\delta'_n \to \delta$ with $\beta(\delta'_n) \leq \epsilon/2\ \forall n$, so that by continuity of $\beta$, we have $\beta(\delta) = \lim_{n \to \infty} \beta(\delta'_n) \leq \epsilon/2$. Therefore, the number of time steps required for the algorithm to terminate is bounded above by the number $\lceil (b - a) / \delta \rceil$, where $\delta > 0$ can be taken as the largest number that satisfies $\beta(x) \leq \epsilon/2$ for $0 \leq x \leq \delta$.
\end{proof}

Next, we present details of constructing and searching over the grid for Lipschitz-like optimization.
\subsection{Lipschitz-like optimization over the standard simplex using grid search}
We make the following observations about the integer grid defined in Eq.~\eqref{eqn:integer_grid_simplex}. For the grid $\Delta_{d, N}$ defined in Eq.~\eqref{eqn:simplex_grid} over the standard simplex, we have $\Delta_{d, N} = \mathcal{I}_{d, N}/N := \{n/N \mid n \in \mathcal{I}_{d, N}\}$. Therefore, the observations noted below also apply to $\Delta_{d, N}$ with appropriate modifications.

\begin{proposition}
    \label{prop:simplex_grid}
    Given $d, N \in \mathbb{N}_+$, let $\mathcal{I}_{d, N}$ denote the integer grid defined in Eq.~\eqref{eqn:integer_grid_simplex}. Then the following hold.
    \begin{enumerate}[label=\arabic*)]
        \item Any element $n \in \mathcal{I}_{d, N}$ can be written as $n = (N - \ell_{d - 1}, \ell_{d - 1} - \ell_{d - 2}, \dotsc, \ell_2 - \ell_1, \ell_1)$ for some integers $0 \leq \ell_i \leq \ell_{i + 1} \leq N$, $i \in [d - 2]$.
        \item The elements of $\mathcal{I}_{n, N}$ can be ordered such that any two consecutive elements are distance $2$ apart in $l_1$-norm. This ordering is constructive and can be implemented algorithmically.
        \item The grid $\Delta_{d, N}$ defined in Eq.~\eqref{eqn:simplex_grid} is a $(2(d - 1)/N)$-net in the $l_1$-norm for the standard simplex in dimension $d$. That is, given any $x \in \Delta_d$, there is some $z \in \Delta_{d, N}$ such that $\norm{x - z}_1 \leq 2(d - 1)/N$.
    \end{enumerate}
\end{proposition}
\begin{proof}
    1) We prove this statement by induction on the dimension. For $d = 2$, any $n \in \mathcal{I}_{d, M}$ satisfies $n_1 + n_2 = M$, and therefore, $n = (M - n_1, n_1)$ holds for all $M \in \mathbb{N}_+$. Now, assume that for any $M \in \mathbb{N}_+$, we can write each $m \in \mathcal{I}_{d - 1, M}$ in dimension $d - 1$ as $m = (M - s_{d - 2}, \dotsc, s_2 - s_1, s_1)$, where $s_1, \dotsc, s_{d - 2} \in \mathbb{N}$ satisfy $0 \leq s_i \leq s_{i + 1} \leq M$ for $i \in [d - 3]$. Then, given any $n \in \mathcal{I}_{d, N}$, we can write $n_d + \sum_{i = 1}^{d - 1} n_i = N$. Denote $\ell_{d - 1} = \sum_{i = 1}^{d - 1} n_i$, so that $n_d = N - \ell_{d - 1}$ and $0 \leq \ell_{d - 1} \leq N$. Then, since $\sum_{i = 1}^{d - 1} = \ell_{d - 1}$, we can find some numbers $\ell_1, \dotsc, \ell_{d - 2}$ such that $0 \leq \ell_i \leq \ell_{i + 1} \leq \ell_{d - 1} \leq N$ for $i \in [d - 3]$ and $(n_1, \dotsc, n_{d - 1}) = (\ell_{d - 1} - \ell_{d - 2}, \dotsc, \ell_2 - \ell_1, \ell_1)$ by assumption. Subsequently, we can write $n = (N - \ell_{d - 1}, \dotsc, \ell_2 - \ell_1, \ell_1)$. Thus, by induction, the statement holds for all dimensions. \newline

    2) We explicitly construct an ordering to prove the assertion. The proof uses induction to obtain the desired result. Suppose that for dimension $d \in \mathbb{N}_+$ and any $N \in \mathbb{N}_+$, the elements of $\mathcal{I}_{d, N}$ can be arranged such that the first element is $(N, 0, \dotsc, 0)$, the last element is $(0, \dotsc, 0, N)$, and the $l_1$-norm distance between any two consecutive elements is $2$. We call this forward ordering of elements. Writing the elements of a forward ordered set gives us reverse ordering, i.e., the first element is $(0, \dotsc, 0, N)$, the last element is $(N, 0, \dotsc, 0)$, and the $l_1$-norm distance between any two consecutive elements is $2$.

    For $d = 2$, the elements of $\mathcal{I}_{2, N}$ can be ordered as either $\{(N, 0), (N - 1, 1), \dotsc, (1, N - 1), (0, N)\}$ (forward ordering) or $\{(0, N), (1, N - 1), \dotsc, (N - 1, 1), (N, 0)\}$ (reverse ordering), so that $1$-norm distance between any two consecutive elements is $2$.

    Assuming that this statement (induction hypothesis) holds for dimension $d$, we show that it also holds for $d + 1$. To that end, we note using the previous result that every element of $\mathcal{I}_{d + 1, N}$ can be written as $(N - s_d, \dotsc, s_2 - s_1, s_1)$ with $0 \leq s_i \leq s_{i + 1} \leq N$ for $i \in [d - 1]$. Then, we (forward) order the elements of $\mathcal{I}_{d + 1, N}$ as follows. Let the first element be $(N, 0, \dotsc, 0)$. Choose the next sequence of elements as follows. For elements of the form $(N - 1, 1 - s_{d - 1}, \dotsc, s_2 - s_1, s_1)$ arrange the elements $(1 - s_{d - 1}, \dotsc, s_2 - s_1, s_1)$ in forward order, which is possible by assumption. For elements of the form $(N - 2, 2 - s_{d - 1}, \dotsc, s_2 - s_1)$ arrange the elements $(2 - s_{d - 1}, \dotsc, s_2 - s_1)$ in reverse order. Continuing this way, given elements of the form $(N - s_d, s_d - s_{d - 1}, \dotsc, s_2 - s_1, s_1)$ for fixed $s_d$, arrange the elements $(s_d - s_{d - 1}, \dotsc, s_2 - s_1, s_1)$ in forward order if $s_d$ is odd, and in reverse order if $s_d$ is even.

    Then, for a fixed $0 \leq s_d \leq N$, if $(N - s_d, s_d - s_{d - 1}', \dotsc, s_1')$ is the element after $(N - s_d, s_d - s_{d - 1}, \dotsc, s_1)$ as per the above ordering, we have $\norm{(N - s_d, s_d - s_{d - 1}', \dotsc, s_1') - (N - s_d, s_d - s_{d - 1}, \dotsc, s_1)}_1 = \norm{(s_d - s_{d - 1}', \dotsc, s_1') - (s_d - s_{d - 1}', \dotsc, s_1')}_1 = 2$ by induction hypothesis. Next, we consider the case when $s_d$ increases by $1$. In this case, the last element of $(N - s_d, \dotsc, s_2 - s_1)$ is $(N - s_d, s_d, 0, \dotsc, 0)$ if $s_d$ is even and it is $(N - s_d, 0, \dotsc, 0, s_d)$ if $s_d$ is odd. Then, the first element of the next sequence $(N - s_d - 1, s_d + 1 - s_{d - 1}, \dotsc, s_1)$ is $(N - s_d - 1, s_d + 1, 0, \dotsc, 0)$ if $s_d$ is even and it is $(N - s_d - 1, 0, \dotsc, 0, s_d + 1)$ if $s_d$ is odd. Therefore, the $l_1$-norm distance between consecutive elements when $s_d$ increases by $1$ is equal to $2$. Therefore, the elements of $\mathcal{I}_{d + 1, N}$ can be ordered as described above. By induction, the result holds for any dimension. \newline

    3) We prove this by induction. Given dimension $d \in \mathbb{N}_+$, assume that for every $x \in \Delta_d$, there is some $z \in \Delta_{d, N}$ such that $\norm{x - z}_1 \leq 2(d - 1)/N$ for all $N \in \mathbb{N}_+$. For $d = 2$, we can write any $x \in \Delta_2$ as $x = (1 - x_1, x_1)$ for some $x_1 \in [0, 1]$. Let $n_1 \in \argmin\{|x_1 - m_1/N| \mid 0 \leq m_1 \leq N\}$, so that have $|x_1 - n_1/N| \leq 1/2N$. Choosing $z = (1 - n_1/N, n_1/N)$, we obtain $\norm{x - z}_1 \leq 1/N \leq 2/N$.

    Now, suppose that the assumption holds for dimension $d$. We show that it also holds for dimension $d + 1$. To that end, let $x \in \Delta_{d + 1}$ be written as $x = (1 - s, x_2, \dotsc, x_{d + 1})$, where $s = \sum_{i = 2}^{d + 1} x_i$. If $s = 0$, then $x = (1, 0, \dotsc, 0)$ and the result follows. Therefore, let $s > 0$ and consider the vector $x' = (x_2, \dotsc, x_{d + 1})/s$, so that $x' \in \Delta_d$. Let $M = \lceil s N \rceil$, and by assumption, there is some $z' \in \Delta_{d, M}$ such that $\norm{x' - z'}_1 \leq 2(d - 1)/M$. Denote $z' = (n_2, \dotsc, n_{d + 1})/M$ and define $n_1 = N - \sum_{i = 2}^{d + 1} n_i = N - M$. Then, $|(1 - s) - n_1/N| = |s N - M|/N \leq 1/N$ since $sN \leq M \leq sN + 1$. Then, denoting $z = (n_1, \dotsc, n_{d + 1})/N$, we have
    \begin{align*}
        \norm{x - z}_1 &= \left|(1 - s) - \frac{n_1}{N}\right| + \norm{s x' - z' M/N}_1 \\
                       &\leq \frac{1}{N} + s \norm{x' - z'}_1 + \left|s - \frac{M}{N}\right| \norm{z'}_1 \\
                       &\leq \frac{2}{N} + 2(d - 1) \frac{s}{M} \\
                       &\leq \frac{2d}{N}
    \end{align*}
    where we used the fact the $s/M \leq 1/N$. Thus, the statement holds by induction.
\end{proof}

Prop.~\ref{prop:simplex_grid}.1 helps in iterative construction of the grid, while Prop.~\ref{prop:simplex_grid}.2,3 will be used in construction of dense curve. In practice, the ordering of elements given in Prop.~\ref{prop:simplex_grid}.2 can be done efficiently for moderately small dimensions. We use this ordering in practice to generate the grid, and this allows for easy parallelization.

Next, based on the results of Ref.~\cite{deklerk2008simplexcomplexity}, we calculate the value of $N$ that needs to be used in the grid search to converge to a precision of $\epsilon > 0$.
\begin{proposition}
    \label{prop:simplex_grid_search_Lipschitzlike_optimization}
    Let $f\colon \mathcal{D} \to \mathbb{R}$ be a $\beta$-Lipschitz-like function satisfying Eq.~\eqref{eqn:Lipschitz_like_function}, where $\mathcal{D} = \Delta_d$ is the standard simplex in $\mathbb{R}^d$. Let $N = \lceil 1/\delta^2 \rceil$, where $\delta = \sup\{\delta' > 0 \mid \beta(\delta') \leq \epsilon/2\}$. Let $\Delta_{d, N}$ be the grid defined in Eq.~\eqref{eqn:simplex_grid}, and let $f^* = \max\{f(z) \mid z \in \Delta_{d, N}\}$. Then, we have $\max_{x \in \Delta_d} f(x) - f^* \leq \epsilon$.
\end{proposition}
\begin{proof}
    This result is implied by techniques developed in Ref.~\cite{deklerk2008simplexcomplexity}. For the sake of completeness, we give a proof here for the specific case of Lipschitz-like functions.

    Given $f\colon \mathcal{D} \to \mathbb{R}$, the Bernstein polynomial approximation of $f$ of order $N \in \mathbb{N}_+$ is defined as
    \begin{equation*}
        B_N(f)(x) = \sum_{n \in \mathcal{I}_{d, N}} f\left(\frac{n}{N}\right) \frac{N!}{n!}\ x^n
    \end{equation*}
    for $x \in \Delta_d$, where we use the multi-index notation $n! = n_1! \dotsm n_d!$, $x^n = x_1^{n_1} \dotsm x_d^{n_d}$ for $n \in \mathcal{I}_{d, N}$~\cite{deklerk2008simplexcomplexity}. The sup-norm of a continuous function $f\colon \mathcal{D} \to \mathbb{R}$ is given by $\norm{f}_\infty = \max_{x \in \mathcal{D}} |f(x)|$, where we use the fact that $\mathcal{D}$ is compact. Next, given a continuous function $f\colon \mathcal{D} \to \mathbb{R}$, the modulus of continuity of $f$ is defined as
    \begin{equation*}
        \omega(f, \delta) = \max_{\substack{x, y \in \mathcal{D}\\ \norm{x - y} \leq \delta}} |f(x) - f(y)|
    \end{equation*}
    with respect to a suitable norm $\norm{\cdot}$ on $\mathbb{R}^n$~\cite{deklerk2008simplexcomplexity}. If $f$ is Lipschitz-like, then we use the norm with respect to which $f$ satisfies the Lipschitz-like property. For a $\beta$-Lipschitz-like function $f$, using Eq.~\eqref{eqn:Lipschitz_like_function}, we have
    \begin{equation}
        \omega(f, \delta) \leq \beta(\delta), \label{eqn:modulus_of_continuity_Lipschitzlike}
    \end{equation}
    where we used the fact that $\beta$ is a monotonically increasing function.

    Then, given a $\beta$-Lipschitz-like function $f\colon \mathcal{D} \to \mathbb{R}$, using Thm.~(3.2) of Ref.~\cite{deklerk2008simplexcomplexity}, we have
    \begin{equation*}
        \norm{B_N(f) - f}_\infty \leq 2 \omega\left(f, \frac{1}{\sqrt{N}}\right) \leq 2 \beta\left(\frac{1}{\sqrt{N}}\right)
    \end{equation*}
    where we used Eq.~\eqref{eqn:modulus_of_continuity_Lipschitzlike} to obtain the last inequality. In particular, this implies that $\max_{x \in \mathcal{D}} B_N(f) \geq \max_{x \in \mathcal{D}} f(x) - 2 \beta(1/\sqrt{N})$. Using this result along with Lemma~(3.1) of Ref.~\cite{deklerk2008simplexcomplexity}, we obtain
    \begin{equation*}
        \max_{x \in \mathcal{D}} f(x) - 2 \beta\left(\frac{1}{\sqrt{N}}\right) \leq \max_{x \in \mathcal{D}} B_N(f)(x) \leq \max_{x \in \Delta_{d, N}} f(x).
    \end{equation*}
    Therefore, choosing
    \begin{equation*}
        \beta\left(\frac{1}{\sqrt{N}}\right) \leq \frac{\epsilon}{2},
    \end{equation*}
    we obtain
    \begin{equation*}
        \max_{x \in \mathcal{D}} f(x) \leq \max_{x \in \Delta_{d, N}} f(x) + \epsilon
    \end{equation*}
    giving us the desired result. In other words, to compute the maximum of $f$ to within a precision of $\epsilon > 0$, it suffices to search the grid $\Delta_{d, N}$ with $N = \lceil 1/\delta^2 \rceil$, where $\delta$ is the largest number satisfying $\beta(\delta) \leq \epsilon/2$.
\end{proof}

\subsection{Lipschitz-like optimization over the standard simplex using dense curves\label{app:Lipschitzlike_optimization_simplex_dense_curve}}
We show how to optimize Lipschitz-like functions over the standard simplex using $\alpha$-dense curves. Such curves get within a distance $\alpha$ of all points in the simplex (see Def.~\eqref{def:dense_curve}). As we show below, by appropriately choosing $\alpha$, one can perform a one-dimensional optimization to obtain the maximum of $f$ to the desired precision.

\begin{proposition}
    \label{prop:Lipschitzlike_optimization_dense_curve_simplex}
    Suppose that $f\colon \Delta_d \to \mathbb{R}$ is a $\beta$-Lipschitz-like function satisfying Eq.~\eqref{eqn:Lipschitz_like_function}. Then the following hold.
    \begin{enumerate}
        \item Given $N \in \mathbb{N})_+$, let $\gamma\colon [0, L_{\text{curve}}] \to \Delta_d$ be a Lipschitz-like curve constructed as per Alg.~\ref{alg:construct_simplex_dense_curve}. Then the curve $\gamma$ is $(2(d - 1)/N)$-dense in the simplex $\Delta_d$ and satisfies Eq.~\eqref{eqn:simplex_dense_curve_Lipschitz}.
        \item Alg.~\ref{alg:maximize_Lipschitz_like_function_simplex_dense_curve} computes the maximum of $f$ to within a precision of $\epsilon > 0$ for any $\beta_\gamma$-Lipschitz-like, $(2(d - 1)/N)$-dense curve $\gamma$. Here, $N = \lceil 2(d - 1)/\alpha \rceil$ with $\alpha = \sup\{\alpha' > 0 \mid \beta(\alpha') \leq \epsilon/2\}$ as noted in Alg.~\ref{alg:maximize_Lipschitz_like_function_simplex_dense_curve}. In the worst case, the algorithm takes
            \begin{equation*}
                \left\lceil \frac{2}{N \delta} \binom{N + d - 1}{d - 1} \right\rceil
            \end{equation*}
            time steps to converge to the maximum within a precision of $\epsilon > 0$, where $\delta = \sup\{\delta' > 0 \mid \beta(\beta_\gamma(\delta')) \leq \epsilon/2\}$. When $\gamma$ is the curve generated by Alg.~\ref{alg:construct_simplex_dense_curve}, $\alpha < 1$ and $d \gg 1$, this amounts to $O(\alpha^{1 - d}/d)$ time steps in the worst case.
    \end{enumerate}
\end{proposition}
\begin{proof}
    1. From Prop.~\ref{prop:simplex_grid}, we know that $\Delta_{d, N}$ is a $(2(d - 1)/N)$-net. Since $\Delta_{d, N} \subseteq \text{Range}(\gamma)$, for any $x \in \Delta_d$, there is some $z \in \text{Range}(\gamma)$ such that $\norm{x - z}_1 \leq 2(d - 1)/N$. In other words, $\gamma$ is a $(2(d - 1)/N)$-dense curve in $\Delta_d$. Next, we show that $\gamma$ is Lipschitz-like function. To that end, given $\theta, \theta' \in \mathbb{R}$, let $k = \lceil \theta N/2 \rceil$ and $k' = \lceil \theta' N/2 \rceil$. Without loss of generality, take $\theta \leq \theta'$. Given $i \in [N_{\text{grid}}]$, where $N_{\text{grid}}$ is defined in Eq.~\eqref{eqn:simplex_grid_cardinality}, denote $x_i \in \Delta_{d, N}$ to be the $i$th element of $\Delta_{d, N}$ ordered as per Eq.~\eqref{def:equidistant_ordering_simplex_grid}. Now, if $k = k'$, then we can write $\gamma(\theta) = t x_k + (1 - t) x_{k + 1}$ and $\gamma(\theta') = t' x_k + (1 - t') x_{k + 1}$, where $t = 1 + k - \theta N/2$ and $t' = 1 + k - \theta' N/2$ (see Alg.~\ref{alg:construct_simplex_dense_curve}). Therefore, $\norm{\gamma(\theta') - \gamma(\theta)}_1 \leq |t - t'| \norm{x_{k + 1} - x_k}_1 = |\theta' - \theta|$, where we used the fact that $\norm{x_{k + 1} - x_k}_1 = 2/N$ (see Prop.~\ref{prop:simplex_grid}). Thus, consider the case $k' > k$, so that
    \begin{equation*}
        \norm{\gamma(\theta') - \gamma(\theta)}_1 \leq \norm{\gamma(\theta') - x_{k'}}_1 + \sum_{i = k + 1}^{k' - 1} \norm{x_{i + 1} - x_i}_1 + \norm{x_{k + 1} - \gamma(\theta)}_1 = |\theta' - \theta|
    \end{equation*}
    where in the last line we used the fact that $\theta$ (and $\theta'$) is defined as the length along obtained by joining consecutive points of the grid until we reach $\theta$ (and $\theta'$ respectively). Since $\norm{x - y}_1 \leq 2$ for all $x, y \in \Delta_d$, Eq.~\eqref{eqn:simplex_dense_curve_Lipschitz} follows. \newline

    2. Consider the real-valued function $g = f \circ \gamma$ defined on the interval $[0, L_{\text{curve}}]$. Here, $\gamma$ is a $(2(d - 1)/N)$-dense $\beta_\gamma$-Lipschitz-like curve with $N = \lceil 2(d - 1)/\alpha \rceil$ and $\alpha = \sup\{\alpha' > 0 \mid \beta(\alpha) \leq \epsilon/2\}$. Since $2(d - 1)/N \leq \alpha$, $\gamma$ is also an $\alpha$-dense curve. Let $\theta^* \in \argmax\{g(\theta) \mid \theta \in [0, L_{\text{curve}}]\}$, and let $\theta^*_\epsilon \in [0, L_{\text{curve}}]$ be the point output by Alg.~\ref{alg:maximize_Lipschitz_like_function_1D}. Since $f$ is $\beta$-Lipschitz-like and $\gamma$ is $\beta_\gamma$-Lipschitz-like, $g = f \circ \gamma$ is $\beta \circ \beta_\gamma$-Lipschitz-like, so by Prop.~\ref{prop:Lipshitz_like_maximization_analysis}, we know that $g(\theta^*) - g(\theta^*_\epsilon) \leq \epsilon/2$. Let $x^* \in \argmax\{f(x) \mid x \in \Delta_d\}$ denote a point achieving the maximum of $f$. Then, since $\gamma$ is an $\alpha$-dense curve, there is some point $\theta_0 \in [0, L_{\text{curve}}]$ such that $\norm{x^* - \gamma(\theta_0)}_1 \leq \alpha$. Noting that $g(\theta_0) \leq g(\theta^*)$, we have
    \begin{align*}
        f(x^*) - g(\theta^*_\epsilon) &\leq f(x^*) - g(\theta^*) + \frac{\epsilon}{2} \\
                                      &\leq f(x^*) - g(\theta_0) + \frac{\epsilon}{2} \\
                                      &\leq \beta(\norm{x^* - \gamma(\theta_0)}_1) + \frac{\epsilon}{2} \\
                                      &\leq \epsilon
    \end{align*}
    To obtain the penultimate inequality we used the Lipschitz-like property of $f$ along with the fact that $g(\theta_0) = f(\gamma(\theta_0))$. The last inequality follows by noting that $\beta(\norm{x^* - \gamma(\theta_0)}_1) \leq \beta(\alpha) \leq \epsilon/2$. The number of iterations needed to compute $\theta^*_\epsilon$ in the worst-case is given by
    \begin{equation*}
        N_{\text{maxiter}} = \left\lceil \frac{L_{\text{curve}}}{\delta} \right\rceil = \left\lceil \frac{2}{N \delta} \binom{N + d - 1}{d - 1} \right\rceil
    \end{equation*}
    where $\delta = \sup\{\delta' > 0 \mid \beta(\beta_\gamma(\delta')) \leq \epsilon/2\}$, which follows from Prop.~\ref{prop:Lipshitz_like_maximization_analysis}. Using the fact that $(n/k)^k \leq \binom{n}{k} \leq (e n/k)^k$ and $2(d - 1)/\alpha \leq N \leq 2(d - 1)/\alpha + 1$, the number of iterations is bounded as
    \begin{equation*}
        \frac{\alpha}{(d - 1) \delta} \left(\frac{(2(d - 1)/\alpha + d - 1)}{d - 1}\right)^{d - 1} \leq N_{\text{maxiter}} \leq \frac{\alpha}{(d - 1) \delta} \left(\frac{e (2(d - 1)/\alpha + d)}{d - 1}\right)^{d - 1}.
    \end{equation*}
    Thus, assuming $\alpha < 1$ and $d \gg 1$, we need $O(\alpha^{2 - d} \delta^{-1}/d)$ time steps for convergence. For the case when $\gamma$ is the curve generated by Alg.~\ref{alg:construct_simplex_dense_curve}, we can take $\beta_\gamma(x) = x$ owing to Eq.~\eqref{eqn:simplex_dense_curve_Lipschitz}, and therefore, the statement of the proposition follows.
\end{proof}

\subsection{Lipschitz-like optimization over a compact \& convex domain\label{app:Lipschitzlike_optimization_compact_convex_domain}}
We now consider the general problem of optimizing a $\beta$-Lipschitz-like function $f\colon \mathcal{D} \to \mathbb{R}$ satisfying Eq.~\eqref{eqn:Lipschitz_like_function}, where $\mathcal{D} \subseteq \mathbb{R}^d$ is a nonempty compact and convex domain. Many techniques have been developed for optimizing Lipschitz or H\"older continuous functions when $\mathcal{D}$ is a specific set such as hypercube, simplex or is the full Euclidean space $\mathbb{R}^n$ as noted in Sec.~\ref{secn:Lipschitzlike_optimization}.
However, when dealing with an arbitrary compact and convex domain, it is not always obvious how to encode the constraints defining the domain in the algorithm.

We circumvent this problem by looking for a function $\overline{f}\colon \mathbb{R}^d \to \mathbb{R}$ satisfying the following properties:
\begin{enumerate}
    \item The function $\overline{f}$ is an extension of $f$ in the sense that $\overline{f}$ is Lipschitz-like and $\overline{f}|_{\mathcal{D}} = f$.
    \item The maximum of $\overline{f}$ and $f$ coincide. More specifically, $\sup_{x \in \mathcal{K}} \overline{f} = \max_{x \in \mathcal{D}} f$ for every set $\mathcal{D} \subseteq \mathcal{K} \subseteq \mathbb{R}^d$.
    \item The function $\overline{f}$ can be efficiently computed whenever $f$ and $\beta$ can be efficiently computed.
\end{enumerate}
The first property ensures that we are able to use existing Lipschitz-like optimization algorithms. The second property ensures that we can compute the maximum of $f$ by computing the maximum of $\overline{f}$ over a convenient set $\mathcal{K}$ containing the domain $\mathcal{D}$. The choice of the set $\mathcal{K}$ will usually depend on the actual algorithm used to perform the optimization. The last property ensures that the function $\overline{f}$ can be efficiently evaluated in practice.

We show that one can construct a function $\overline{f}$ satisfying all the above properties for the case when $\beta$ is itself a Lipschitz-like function. That is, there is some non-negative, continuous, monotonically increasing function $\kappa\colon \mathbb{R}_+ \to \mathbb{R}$ with $\kappa(0) = 0$ such that, for all $x, y \in \mathbb{R}_+$,
\begin{equation}
    |\beta(x) - \beta(y)| \leq \kappa(|x - y|). \label{eqn:beta_also_Lipschitz_like}
\end{equation}
 This assumption is not too restrictive for our purposes because Lipschitz continuous functions (i.e., $\beta(x) = x$), and functions $\beta$ relevant to entropic quantities like the one in Eq.~\eqref{eqn:beta_twosenderMAC} satisfy Eq.~\eqref{eqn:beta_also_Lipschitz_like} (see Prop.~\ref{prop:Lipschitzlike_function_extension}). In this case, we can define the extension of $\overline{f}$ as follows.

\begin{definition}[Lipschitz-like extension]
    \label{def:Lipschitzlike_function_extension}
    Let $f\colon \mathcal{D} \to \mathbb{R}$ be a $\beta$-Lipschitz-like function over a compact and convex domain $\mathcal{D} \subseteq \mathbb{R}^d$. Let $\norm{\cdot}$ be the norm on $\mathbb{R}^d$ with respect to which Eq.~\eqref{eqn:Lipschitz_like_function} holds. Suppose that $\beta$ is $\kappa$-Lipschitz-like. Then, the Lipschitz-like extension of $f$ is the function $\overline{f}\colon \mathbb{R}^d \to \mathbb{R}$ defined as
    \begin{equation}
        \overline{f}(x) = f(\Pi_{\mathcal{D}}(x)) - \beta(\norm{x - \Pi_{\mathcal{D}}(x)}), \label{eqn:Lipschitzlike_function_extension}
    \end{equation}
    where $\Pi_{\mathcal{D}} = \argmin_{z \in \mathcal{D}} \norm{z - x}_2$ is the projection of $x$ onto $\mathcal{D}$ with respect to the Euclidean norm.
\end{definition}
We remark that the choice of norm in Eq.~\eqref{eqn:Lipschitz_like_function}, and subsequently in Eq.~\eqref{eqn:Lipschitzlike_function_extension}, is flexible. Moreover, because all norms on $\mathbb{R}^d$ are equivalent, one can change the function $\beta$ so as to get a Lipschitz-like property with respect to a different norm. Recall that two norms $\norm{\cdot}_a$ and $\norm{\cdot}_b$ on $\mathbb{R}^d$ are said to be equivalent if there exist constants $c_1, c_2 > 0$ (possibly dimension dependent) such that $c_1 \norm{\cdot}_a~\leq~\norm{\cdot}_b~\leq~c_2 \norm{\cdot}_a$. The caveat of using equivalence of norms to change the function $\beta$ is that the modified $\beta$ might end up depending on the dimension. This could negatively impact the convergence rate (for example, polynomial time convergence guarantee for optimization over the simplex using grid search might be lost if $\beta$ depends on the dimension). This observation is also of relevance in Def.~\ref{def:Lipschitzlike_function_extension} because we define the projection with respect to the Euclidean norm, while allow an arbitrary choice of norm in Eq.~\eqref{eqn:Lipschitzlike_function_extension} (also see Prop.~\ref{prop:Lipschitzlike_function_extension}).

Note that $\Pi_{\mathcal{D}}$ is well-defined since $\mathcal{D}$ is compact and convex (see Thm.~(2.5) in Ref.~\cite{conway2019course}). Since $\Pi_{\mathcal{D}}(x) = x$ for $x \in \mathcal{D}$, we can see that $\overline{f}$ is indeed an extension of $f$. Roughly speaking, the extension is also Lipschitz-like because both $f$ and $\beta$ are Lipschitz-like. Furthermore, since $\beta$ is a non-negative monotonically increasing function, the value of $\overline{f}$ decreases as we move away from $\mathcal{D}$. Thus, the maximum of $\overline{f}$ occurs over $\mathcal{D}$. Finally, so long as the projection $\Pi_{\mathcal{D}}$ can be efficiently computed, the function $\overline{f}$ can be efficiently computed (assuming $f$ and $\beta$ can be efficiently computed).

We formalize the observations made above in the following proposition.
\begin{proposition}
    \label{prop:Lipschitzlike_function_extension}
    Let $\beta\colon \mathbb{R}_+ \to \mathbb{R}$ be a non-negative, continuous, monotonically increasing function with $\beta(0) = 0$. Suppose that $\beta$ is $\kappa$-Lipschitz-like in the sense of Eq.~\eqref{eqn:beta_also_Lipschitz_like} for some appropriate $\kappa$. Let $f\colon \mathcal{D} \to \mathbb{R}$ be a $\beta$-Lipschitz-like function over some compact and convex domain $\mathcal{D} \subseteq \mathbb{R}^d$. Let $\overline{f}$ be the extension of $f$ as defined in Eq.~\eqref{eqn:Lipschitzlike_function_extension}. Then the following statements hold.
    \begin{enumerate}
        \item Define the constant $C = c_2/c_1$, where $c_1, c_2 > 0$ are obtained from equivalence of $\norm{\cdot}, \norm{\cdot}_2$ on $\mathbb{R}^d$, i.e., $c_1 \norm{\cdot} \leq \norm{\cdot}_2 \leq c_2 \norm{\cdot}$. Then, the extension $\overline{f}$ is a $\overline{\beta}$-Lipschitz-like function, where
            \begin{equation}
                \overline{\beta}(x) = \beta(C x) + \kappa((C + 1) x). \label{eqn:extended_beta}
            \end{equation}
        \item Given any set $\mathcal{D} \subseteq \mathcal{K} \subseteq \mathbb{R}^d$, we have $\sup_{x \in \mathcal{K}} \overline{f}(x) = \max_{x \in \mathcal{D}} f(x)$. Moreover, any point achieving the maximum of $f$ achieves the maximum of $\overline{f}$.
        \item The modified binary entropy $\overline{h}$ defined in Eq.~\eqref{eqn:modified_binary_entropy} is an $\overline{h}$-Lipschitz-like function, that is,
                \begin{equation*}
                    |\overline{h}(x) - \overline{h}(y)| \leq \overline{h}(|x - y|)
                \end{equation*}
              for all $x, y \in \mathbb{R}_+$. Subsequently, $\beta_I$ defined in Eq.~\eqref{eqn:beta_twosenderMAC} is a $\beta_I$-Lipschitz-like function.
    \end{enumerate}
\end{proposition}
\begin{proof}
    1. We first note that $\Pi_{\mathcal{D}}$ is a non-expansive mapping on $\mathbb{R}^d$, i.e., for all $x, y \in \mathbb{R}^d$, we have $\norm{\Pi_{\mathcal{D}}(x) - \Pi_{\mathcal{D}}(y)}_2 \leq \norm{x - y}_2$~\cite{balestro2019convex}. By equivalence of norms, we can find a (possibly dimension-dependent) constant $C > 0$ noted in the statement of the proposition, such that $\norm{\Pi_{\mathcal{D}}(x) - \Pi_{\mathcal{D}}(y)} \leq C \norm{x - y}$. From this, we can infer the following list of inequalities:
    \begin{align*}
        |\overline{f}(x) - \overline{f}(y)| &\leq |f(\Pi_{\mathcal{D}}(x) - f(\Pi_{\mathcal{D}}(y))| + |\beta(\norm{x - \Pi_{\mathcal{D}}(x)}) - \beta(\norm{y - \Pi_{\mathcal{D}}(y)})| \\
                                            &\leq \beta(\norm{\Pi_{\mathcal{D}}(x) - \Pi_{\mathcal{D}}(y)}) + \kappa(|\norm{x - \Pi_{\mathcal{D}}(x)} - \norm{y - \Pi_{\mathcal{D}}(y)}|) \\
                                            &\leq \beta(C \norm{x - y}) + \kappa(\norm{x - y + \Pi_{\mathcal{D}}(y) - \Pi_{\mathcal{D}}(x)}) \\
                                            &\leq \beta(C \norm{x - y}) + \kappa(\norm{x - y} + \norm{\Pi_{\mathcal{D}}(x) - \Pi_{\mathcal{D}}(y)}) \\
                                            &\leq \beta(C \norm{x - y}) + \kappa((C + 1) \norm{x - y}) \\
                                            &= \overline{\beta}(\norm{x - y})
    \end{align*}
    To obtain the second inequality, we use the fact that $f$ is $\beta$-Lipschitz-like and $\beta$ is $\kappa$-Lipschitz-like. To obtain the third inequality, we use the reverse triangle inequality, the fact that $\Pi_{\mathcal{D}}$ is non-expansive, and that $\beta$ and $\kappa$ are monotonically increasing functions. The last two inequalities follow from similar observations. Finally, we note that $\overline{\beta}(x) = \beta(C x) + \kappa((C + 1) x)$ is a non-negative, continuous, monotonically increasing function with $\overline{\beta}(0) = 0$. \newline

    2. Since $\beta$ is a non-negative function, we have $\overline{f}(x) \leq f(\Pi_{\mathcal{D}}(x))$ for all $x \in \mathbb{R}^d$. Therefore, given any set $\mathcal{D} \subseteq \mathcal{K} \subseteq \mathbb{R}^d$, we have $\sup_{x \in \mathcal{K}} \overline{f}(x) \leq \sup_{x \in \mathcal{K}} f(\Pi_{\mathcal{D}}(x)) = \max_{x \in \mathcal{D}} f(x)$. The last equality follows by noting that $\Pi_{\mathcal{D}}(x) \in \mathcal{D}$ for all $x \in \mathbb{R}^d$, $f$ is continuous and $\mathcal{D}$ is compact. If $x^* \in \mathcal{D}$ achieves the maximum of $f$, then $\sup_{x \in \mathcal{K}} \overline{f}(x) = \max_{x \in \mathcal{D}} f(x) = f(x^*) = \overline{f}(x^*)$, where in the last step, we used the fact that $f = \overline{f}$ on $\mathcal{D}$. \newline

    3. Let $\overline{h}$ be the modified binary entropy defined in Eq.~\eqref{eqn:modified_binary_entropy}. If $x \in [0, 1/2]$ then, $\overline{h}(x) = h(x)$. Furthermore, we have $h(x) = H((x, 1 - x))$ for $x \in [0, 1]$, where $H$ is the Shannon entropy. Therefore, given $x, y \in [0, 1/2]$, we have $|\overline{h}(x) - \overline{h}(y)| = |h(x) - h(y)| = |H((x, 1 - x)) - H((y, 1 - y))| \leq h(|x - y|) = \overline{h}(|x - y|)$, where the inequality follows from the results of Ref.~\cite{zhang2007estimating}. Here, we used the fact that $|x - y| \leq 1/2$ when $x, y \in [0, 1/2]$. For $x, y \geq 1/2$, we have $\overline{h}(x) = \overline{h}(y) = \ln(2)$, and thus, $|\overline{h}(x) - \overline{h}(y)| \leq \overline{h}(|x - y|)$. Finally, we consider the case when $x \in [0, 1/2]$ and $y \geq 1/2$. Here, we have $|\overline{h}(x) - \overline{h}(y)| = \ln(2) - h(x) = h(1/2) - h(x) \leq h(|1/2 - x|) = \overline{h}(|1/2 - x|) \leq \overline{h}(|y - x|)$. In the last step, we used the fact that $|1/2 - x| \leq |y - x|$ and $\overline{h}$ is monotonically increasing. Therefore, we have $|\overline{h}(x) - \overline{h}(y)| \leq \overline{h}(|x - y|)$ for all $x, y \in \mathbb{R}_+$.
\end{proof}

The above extension gives some freedom in determining what algorithm to use to perform the maximization of $\overline{f}$, especially for the case when $f$ is Lipschitz or H\"older continuous. One can, for example, use an unconstrained optimization method developed for optimization of Lipschitz or H\"older continuous function. Alternatively, one can embed $\mathcal{D}$ inside a hypercube $\mathcal{K}$, and use an algorithm that can perform global Lipschitz optimization~\cite{lera2010lipschitz, sergeyev2013introduction, ziadi2001global}. In this study, we generalize the method used in Ref.~\cite{ziadi2001global} using dense curves because the proof of convergence essentially follows Prop.~\ref{prop:Lipschitzlike_optimization_dense_curve_simplex}. In practice, one might get faster convergence by generalizing the global optimization method using Hilbert space-filling curves studied in Ref.~\cite{lera2010lipschitz}.

We use the $\alpha$-dense curve for filling the hypercube $\mathcal{K} = \prod_{i = 1}^d [a_i, b_i]$ used in Ref.~\cite{ziadi2001global}, which we reproduce below for convenience. Let $d \geq 2$, $\eta > 0$, $\eta_1 = 1$ and define 
\begin{equation}
    \eta_i = \left(\frac{\eta}{\pi}\right)^{i - 1} \prod_{j = 2}^i \frac{1}{|a_j| + |b_j|} \label{eqn:dense_curve_hypercube_eta_i}
\end{equation}
for $i = 2, \dotsc, d$. Then, the curve $\gamma\colon [0, \pi/\eta_d] \to \mathcal{K}$ defined as
\begin{equation}
    \gamma_i(\theta) = \frac{a_i - b_i}{2} \cos(\eta_i \theta) + \frac{a_i + b_i}{2} \label{eqn:dense_curve_hypercube}
\end{equation}
for $i \in [d]$ is a $(\sqrt{d - 1} \eta)$-dense curve with respect to the Euclidean norm~\cite{ziadi2001global}. Furthermore, $\gamma$ is Lipschitz continuous (with respect to the Euclidean norm) with Lipschitz constant
\begin{equation}
    L_\gamma = \frac{1}{2} \left(\sum_{i = 1}^n (|a_i| + |b_i|)^2 \eta_i^2\right)^{1/2}. \label{eqn:dense_curve_hypercube_Lipschitz_constant}
\end{equation}

\begin{algorithm}[H]
    \begin{algorithmic}[1]
        \Function{maximize\_Lipshcitz-like\_function\_compact\_convex\_domain}{$d$, $\beta$, $\epsilon$}
            \State Find $\mathcal{K} = \prod_{i = 1}^d [a_i, b_i] \subseteq \mathbb{R}^d$ such that $\mathcal{D} \subseteq \mathcal{K}$
            \State Find $c_1 > 0$ such that $c_1 \norm{\cdot} \leq \norm{\cdot}_2$
            \State Construct the curve $\gamma$ as per Eq.~\eqref{eqn:dense_curve_hypercube} for $\eta = c_1 \alpha/\sqrt{d - 1}$, where $\alpha = \sup\{\alpha' > 0 \mid \overline{\beta}(\alpha') \leq \epsilon/2\}$
            \State Construct the extension $\overline{f}$ as per Eq.~\eqref{eqn:Lipschitzlike_function_extension}
            \State Compute the maximum $\overline{g}^*$ of $\overline{g} = \overline{f} \circ \gamma$ over $[0, \pi/\eta_d]$ to a precision of $\epsilon/2$ using Alg.~\ref{alg:maximize_Lipschitz_like_function_1D}
            \State \textbf{return} $\overline{g}^*$
        \EndFunction
    \end{algorithmic}
    \caption{Computing the maximum of a $\beta$-Lipschitz-like function $f$ satisfying Eq.~\eqref{eqn:Lipschitz_like_function} for $\mathcal{D} = \Delta_d$, given $\epsilon > 0$}
    \label{alg:maximize_Lipschitz_like_function_dense_curve_compact_convex_set}
\end{algorithm}

Below, we show that the above algorithm will converge to the maximum within precision $\epsilon > 0$. The proof essentially adapts the ideas used to prove Prop.~\ref{prop:Lipschitzlike_optimization_dense_curve_simplex}.
\begin{proposition}
    \label{prop:Lipschitzlike_optimization_dense_curve_compact_convex_set}
    Let $f\colon \mathcal{D} \to \mathbb{R}$ be a $\beta$-Lipschitz-like function with respect to the norm $\norm{\cdot}$ on $\mathbb{R}^d$, where $\mathcal{D} \subseteq \mathbb{R}^d$ is a compact \& convex set. Suppose that $\beta$ is $\kappa$-Lipschitz-like. Let $\mathcal{K} \subseteq \mathbb{R}^d$ be a bounded set containing $\mathcal{D}$. Let $\gamma$ be any $\alpha$-dense, $\beta_\gamma$-Lipschitz-like curve from the interval $[a, b]$ into the set $\mathcal{K}$. Let $\overline{f}$ be a $\overline{\beta}$-Lipschitz-like extension of $f$ defined in Eq.~\eqref{eqn:Lipschitzlike_function_extension}, where $\overline{\beta}$ is given in Eq.~\eqref{eqn:extended_beta}. Suppose that we compute the maximum of $f$ using an appropriate generalization of Alg.~\ref{alg:maximize_Lipschitz_like_function_dense_curve_compact_convex_set} with the curve $\gamma$. Then, for $\alpha = \sup\{\alpha' > 0 \mid \overline{\beta}(\alpha') \leq \epsilon/2\}$, the algorithm computes the maximum of $f$ to within a precision of $\epsilon > 0$. In the worst case, this algorithm takes $\lceil (b - a) / \delta \rceil$ time steps to converge, where $\delta = \sup\{\delta' > 0 \mid \overline{\beta}(\beta_\gamma(\delta')) \leq \epsilon/2\}$.

    In particular, if $\mathcal{K}$ is a hypercube containing $\mathcal{D}$ and $\gamma$ is the curve generated as per Eq.~\eqref{eqn:dense_curve_hypercube}, a value of $\eta = c_1 \alpha/\sqrt{d - 1}$ suffices to converge to the maximum of $f$ within a precision of $\epsilon > 0$. Here, $\alpha$ is as defined above and $c_1 > 0$ is a \textnormal{(}possibly dimension dependent\textnormal{)} constant satisfying $c_1 \norm{\cdot} \leq \norm{\cdot}_2$. The algorithm takes $\lceil \frac{\pi}{\eta_d \delta} \rceil$ time steps to converge in the worst case, where $\eta_d$ is defined in Eq.~\eqref{eqn:dense_curve_hypercube_eta_i} and $\delta$ is as defined above. For the case when $f$ is a Lipschitz continuous function with Lipschitz constant $L$ with respect to the Euclidean norm, Alg.~\ref{alg:maximize_Lipschitz_like_function_dense_curve_compact_convex_set} takes $O((6 \pi L \sqrt{d}/\epsilon)^d)$ iterations to converge to a precision of $\epsilon > 0$. The exponential scaling with $L$ and $\epsilon$ cannot be improved without additional assumptions on the functions or the domain.
\end{proposition}
\begin{proof}
    Since $\mathcal{D} \subseteq \mathbb{R}^d$ is compact, it is bounded, and can therefore be embedded into a bounded set $\mathcal{K} \subseteq \mathbb{R}^d$. Let $\overline{f}$ be the $\overline{\beta}$-Lipschitz-like extension of $f$ defined in Eq.~\eqref{eqn:Lipschitzlike_function_extension}, where $\overline{\beta}$ is given in Eq.~\eqref{eqn:extended_beta}. Let $\gamma\colon [a, b] \to \mathcal{K}$ be an $\alpha$-dense, $\beta_\gamma$-Lipschitz-like curve, where $\alpha = \sup\{\alpha' > 0 \mid \overline{\beta}(\alpha') \leq \epsilon/2\}$. Consider the real-valued function $\overline{g} = \overline{f} \circ \gamma$ defined on the interval $[a, b]$. Let $\theta^* \in \argmax\{\overline{g}(\theta) \mid \theta \in [a, b]\}$, and let $\theta^*_\epsilon \in [a, b]$ be the point output by Alg.~\ref{alg:maximize_Lipschitz_like_function_1D}. Since $\overline{f}$ is $\overline{\beta}$-Lipschitz-like and $\gamma$ is $\beta_\gamma$-Lipschitz-like, $\overline{g} = \overline{f} \circ \gamma$ is $\overline{\beta} \circ \beta_\gamma$-Lipschitz-like. So, by Prop.~\ref{prop:Lipshitz_like_maximization_analysis}, we know that $\overline{g}(\theta^*) - \overline{g}(\theta^*_\epsilon) \leq \epsilon/2$. 

Let $x^* \in \argmax\{f(x) \mid x \in \Delta_d\}$ denote a point achieving the maximum of $f$. From Prop.~\ref{prop:Lipschitzlike_function_extension}, we know that $\overline{f}(x^*) = \sup_{x \in \mathcal{K}} \overline{f}(x) = \sup_{x \in \mathcal{D}} f(x) = f(x^*)$. Then, since $\gamma$ is an $\alpha$-dense curve and $\mathcal{D} \subseteq \mathcal{K}$, there is some point $\theta_0 \in [a, b]$ such that $\norm{x^* - \gamma(\theta_0)}_1 \leq \alpha$. Noting that $\overline{g}(\theta_0) \leq \overline{g}(\theta^*)$, we have
    \begin{align*}
        f(x^*) - \overline{g}(\theta^*_\epsilon) &\leq \overline{f}(x^*) - \overline{g}(\theta^*) + \frac{\epsilon}{2} \\
                                                 &\leq \overline{f}(x^*) - \overline{g}(\theta_0) + \frac{\epsilon}{2} \\
                                                 &\leq \overline{\beta}(\norm{x^* - \gamma(\theta_0)}_1) + \frac{\epsilon}{2} \\
                                                 &\leq \epsilon
    \end{align*}
    To obtain the penultimate inequality we used the Lipschitz-like property of $\overline{f}$ along with the fact that $\overline{g}(\theta_0) = \overline{f}(\gamma(\theta_0))$. The last inequality follows by noting that $\overline{\beta}(\norm{x^* - \gamma(\theta_0)}_1) \leq \overline{\beta}(\alpha) \leq \epsilon/2$. The number of iterations needed to compute $\theta^*_\epsilon$ in the worst-case is given by
    \begin{equation*}
        N_{\text{maxiter}} = \left\lceil \frac{b - a}{\delta} \right\rceil
    \end{equation*}
    where $\delta = \sup\{\delta' > 0 \mid \overline{\beta}(\beta_\gamma(\delta')) \leq \epsilon/2\}$, which follows from Prop.~\ref{prop:Lipshitz_like_maximization_analysis}.

    Now suppose that $\mathcal{K}$ is a hypercube containing $\mathcal{D}$ and $\gamma$ is the curve given in Eq.~\eqref{eqn:dense_curve_hypercube}. By equivalence of norms on $\mathbb{R}^d$, we have $c_1 \norm{\cdot} \leq \norm{\cdot}_2 \leq c_2 \norm{\cdot}$ for some (possibly dimension dependent) constants $c_1, c_2 > 0$. Therefore, $\gamma$ is $(\sqrt{d - 1} \eta/c_1)$-dense and $(L_g/c_1)$-Lipschitz continuous with respect to the norm $\norm{\cdot}$, where $L_g$ is given in Eq.~\eqref{eqn:dense_curve_hypercube_Lipschitz_constant}. By the above results, we can choose $\eta = c_1 \alpha/\sqrt{d - 1}$ for convergence, where $\alpha = \{\alpha' > 0 \mid \overline{\beta}(\alpha') \leq \epsilon/2\}$. In this worst case, this algorithm takes $\lceil \pi/(\eta_d \delta) \rceil$, where $\eta_d$ is defined in Eq.~\eqref{eqn:dense_curve_hypercube_eta_i} and $\delta = \sup\{\delta' > 0 \mid \overline{\beta}(L_\gamma \delta') \leq \epsilon/2\}$.

    Consider the case when $d \geq 2$, $\mathcal{D} = [0, 1]^d$, $\mathcal{K} = \mathcal{D}$, and $f$ is Lipschitz continuous with Lipschitz constant $L$ (independent of the dimension) with respect to the Euclidean norm $\norm{\cdot}_2$. Then, $c_1 = c_2 = 1$ as defined above, and subsequently, $C$ defined in Prop.~\ref{prop:Lipschitzlike_function_extension} is equal to $1$. Then, since $\kappa(x) = \beta(x) = L x$, we have $\overline{\beta}(x) = 3 L x$. Subsequently, we have $\alpha = \epsilon/6L$, and for $\epsilon < 6L$, we have $\alpha < 1$. For the curve $\gamma$, we have $\eta = \alpha/\sqrt{d -1} = \epsilon/6 L \sqrt{d - 1}$, $\eta_d = (\eta/\pi)^{d - 1}$, $\delta = \epsilon/6 L L_\gamma$, and $L_\gamma \leq \sqrt{d}/2$ assuming that $\alpha < 1$. Thus, the Alg.~\ref{alg:maximize_Lipschitz_like_function_dense_curve_compact_convex_set} takes
    \begin{equation*}
        \left\lceil \frac{6\pi^d L L_\gamma}{\eta^{d - 1} \epsilon} \right\rceil \leq \left\lceil \frac{(6 \pi L \sqrt{d})^d }{\epsilon^d} \right\rceil
    \end{equation*}
    iterations to converge. It is known that any algorithm (in the sense of a black box model) needs at least $O((L/2 \epsilon)^d)$ iterations to compute the optimum of a $L$-Lipschitz function over the unit hypercube to a precision of $\epsilon > 0$ (see Ref.~\cite{deklerk2008complexity} for details). Thus, one cannot, in general improve the $1/\epsilon^d$ scaling for Alg.~\ref{alg:maximize_Lipschitz_like_function_dense_curve_compact_convex_set} without additional assumptions on $\mathcal{D}$ or the class of functions that are optimized by the algorithm.
\end{proof}

The above result shows that the exponential scaling with the dimension cannot be improved without additional assumptions on the functions or the domain. It is likely the case that the $\sqrt{d}$ factor is sub-optimal and comes from the choice of the curve $\gamma$ in Eq.~\eqref{eqn:dense_curve_hypercube}, so one might be able to improve that factor by using better constructions for the curve or by resorting to a different approach altogether.

We remark that the above algorithm will perform worse than Alg.~\ref{alg:maximize_Lipschitz_like_function_simplex_grid_search} and Alg.~\ref{alg:maximize_Lipschitz_like_function_simplex_dense_curve} in practice when $\mathcal{D}$ is the standard simplex. This is because we essentially use no information about the domain except for the projection $\Pi_{\mathcal{D}}$ in constructing the extension $\overline{f}$ (see Eq.~\eqref{eqn:Lipschitzlike_function_extension}). Indeed, when using grid search specially designed for $\mathcal{D} = \Delta_d$, we can get polynomial scaling with the dimension for a fixed precision (see Prop.~\ref{prop:simplex_grid_search_Lipschitzlike_optimization}). Thus, it might be preferable to resort to methods designed specifically for the particular domain and function class of interest, as opposed to general algorithms like Alg.~\ref{alg:maximize_Lipschitz_like_function_dense_curve_compact_convex_set}. That said, an advantage of Alg.~\ref{alg:maximize_Lipschitz_like_function_dense_curve_compact_convex_set} is that the maximum is computed by finding successively better approximations. Therefore, one can decide to terminate the computation after some fixed number of iterations in order to obtain an upper bound on the maximum.

\section{Analysis of algorithms for computing the sum capacity for two-sender MACs\label{app:twosender_MAC_algorithm_analysis}}
\begin{proof}[Proof of Prop.~\ref{prop:mutual_information_maximum_Lipschitz_like}]
    \label{proof:mutual_information_maximum_Lipschitz_like}
    For any probability distributions $p^Z_1, p^Z_2 \in \Delta_{\dout}$, we have $|H(p^Z_1) - H(p^Z_2)| \leq \theta \log(\dout - 1) + h(\theta)$, where $\theta = \norm{p^Z_1 - p^Z_2}_1 / 2$ and $h$ is the binary entropy function~\cite{zhang2007estimating}. In particular, if we define $\overline{h}$ as in Eq.~\eqref{eqn:modified_binary_entropy}, then we have
    \begin{equation}
        |H(p^Z_1) - H(p^Z_2)| \leq \frac{1}{2} \log(\dout - 1) \norm{p^Z_1 - p^Z_2}_1 + \overline{h}\left(\frac{1}{2} \norm{p^Z_1 - p^Z_2}_1\right) \label{eqn:entropy_Lipschitzlike_bound}
    \end{equation}
    because $h(x) \leq \overline{h}(x)$ for all $x \in [0, 1]$.

    Now, note that using the same ideas employed in obtaining $I(p, q)$ given in Eq.~\eqref{eqn:mutualinfo_twosenderMAC_productinput}, we can interchange the roles of $p \in \Delta_{d_1}$ and $q \in \Delta_{d_2}$ to obtain
    \begin{equation*}
        I(p, q) = H(A_p q) - \ip{b_p, q}
    \end{equation*}
    where
    \begin{equation*}
        A_p(z, b_2) = \sum_{b_1 \in \mathcal{B}_1} N(z | b_1, b_2) p(b_1)
        \quad \text{and} \quad
        b_p(b_2) = -\sum_{b_1 \in \mathcal{B}_1} p(b_1) \sum_{z \in \mathcal{Z}} N(z | b_1, b_2) \log(N(z | b_1, b_2)).
    \end{equation*}
    Then, for any $q, q' \in \Delta_{d_2}$ and a fixed $p \in \Delta_{d_1}$, we have
    \begin{align}
        |I(p, q) - I(p, q')| &\leq |H(A_p q) - H(A_p q')| + |\ip{b_p, q - q'}| \nonumber \\
                             &\leq \frac{1}{2} \log(\dout - 1) \norm{A_p q - A_p q'}_1 + \overline{h}\left(\frac{1}{2} \norm{A_p q - A_p q'}_1\right) + \norm{b_p}_\infty \norm{q - q'}_1
                                                                                                                                                        \label{eqn:mutualinfo_twosenderMAC_productinput_intermediate_bound}
    \end{align}
    where in the second step, we used Eq.~\eqref{eqn:entropy_Lipschitzlike_bound} and H\"older's inequality.

    Now, we note that $\norm{A_p q - A_p q'}_1 \leq \norm{A_p}_{1 \to 1} \norm{q - q'}_1$, where $\norm{A_p}_{1 \to 1}$ is the induced matrix norm with respect to $l_1$-norm on the domain and the co-domain of $A_p$. This norm is equal to the maximum (absolute value) column sum of $A_p$. Since $A_p$ is a left stochastic matrix, all of its entries are non-negative and all of its columns sum to $1$, and therefore, $\norm{A_p}_{1 \to 1} = 1$. Similarly, we have $0 \leq b_p(b_2) \leq H_{\mathcal{N}}^{\max}$ for all $b_2 \in \mathcal{B}_2$ (see Eq.~\eqref{eqn:HNmax_twosenderMAC} for definition of $H_{\mathcal{N}}^{\max}$). Therefore, we have $\norm{b_p}_\infty \leq H_{\mathcal{N}}^{\max}$. Then, since $\overline{h}$ a monotonically increasing function, from Eq.~\eqref{eqn:mutualinfo_twosenderMAC_productinput_intermediate_bound}, we obtain
    \begin{equation}
        |I(p, q) - I(p, q')| \leq \frac{1}{2} \log(\dout - 1) \norm{q - q'}_1 + \overline{h}\left(\frac{1}{2} \norm{q - q'}_1\right) + H_{\mathcal{N}}^{\max} \norm{q - q'}_1
                             = \beta_I\left(\norm{q - q'}_1\right) \label{eqn:mutualinfo_twosenderMAC_productinput_bound}
    \end{equation}

    \noindent Since Eq.~\eqref{eqn:mutualinfo_twosenderMAC_productinput_bound} holds for every $p \in \Delta_{d_1}$, we have for fixed $q, q' \in \Delta_{d_2}$
    \begin{equation*}
        \max_{p \in \Delta_{d_1}} I(p, q) \leq \max_{p \in \Delta_{d_1}} I(p, q') + \beta_I\left(\norm{q - q'}_1\right) \text{ and }
        \max_{p \in \Delta_{d_1}} I(p, q') \leq \max_{p \in \Delta_{d_1}} I(p, q) + \beta_I\left(\norm{q - q'}_1\right)
    \end{equation*}
    Since $I^*(q) = \max_{p \in \Delta_{d_1}} I(p, q)$, we obtain $|I^*(q) - I^*(q')| \leq \beta_I\left(\norm{q - q'}_1\right)$.
\end{proof}

\begin{proposition}
    \label{prop:max_iter_sum_capacity_alg_twosenderMAC_inputsize2}
    Let $\mathcal{N}$ be any two-sender MAC with input alphabets of size $d_1$, $d_2$, and output alphabet of size $\dout$. Suppose that that $d_1, \dout \geq 2$ and $d_2 = 2$. Then, for $0 < \epsilon \leq 3\log(\dout)$, the number of iterations required by the \textnormal{\texttt{while}} loop in Alg.~\ref{alg:compute_sum_capacity_twosenderMACs_inputsize2} to compute the sum capacity of the MAC $\mathcal{N}$ to a precision $\epsilon > 0$ is bounded above by
    \begin{equation}
        \left\lceil \frac{2 (\mu^2 + 4)}{(\epsilon \mu + 4) - \sqrt{16 + 8 \epsilon \mu - 4 \epsilon^2}} \right\rceil \label{eqn:max_iter_PS_sum_capacity_twosenderMAC_inputsize2}
    \end{equation}
    where $\mu = 3 \log(\dout)$. In particular, for a precision $0 < \epsilon \leq 3$ chosen independent of $\dout$, the number of iterations is bounded above by $O(\log(\dout)/\epsilon)$ for $\log(\dout) \gg 1$.

    The total cost of computing the sum capacity to a precision $\epsilon > 0$, including the costs of performing intermediate steps such as computing $I^*$, sorting and root-finding, is at most polynomial in $d_1, \dout$ and $1/\epsilon$.
\end{proposition}
\begin{proof}
    To compute the total number of iterations required to compute the sum capacity to a precision $\epsilon > 0$, we break the analysis into the following steps.
    \begin{enumerate}
        \item Number of iterations required for the while loop to converge in Alg.~\ref{alg:compute_sum_capacity_twosenderMACs_inputsize2}.
        \item Cost of computing $I^*(s)$ in Alg.~\ref{alg:compute_sum_capacity_twosenderMACs_inputsize2} to a precision $\epsilon_I > 0$.
        \item Cost of sorting the points $s^{(0)}, \dotsc, s^{(k)}$ in Alg.~\ref{alg:compute_sum_capacity_twosenderMACs_inputsize2}.
        \item Cost for finding a root of $g_i(s)$ in Alg.~\ref{alg:compute_sum_capacity_twosenderMACs_inputsize2} to a precision $\epsilon_r > 0$.
    \end{enumerate}
    From these, we can compute the total cost of converging to a precision $\epsilon > 0$.

    1. We begin by computing the number of iterations required for Alg.~\ref{alg:compute_sum_capacity_twosenderMACs_inputsize2} to converge to a precision $\epsilon > 0$. To that end, note that for obtaining Alg.~\ref{alg:compute_sum_capacity_twosenderMACs_inputsize2} from Alg.~\ref{alg:maximize_Lipschitz_like_function_1D}, we used the fact that $\norm{q_s - q_{s'}}_1 = 2|s - s'|$, where given any $s \in [0, 1]$, we define $q_s = (s, 1 - s)$. Therefore, we have an additional factor of $2$ in $\beta_I(2|s - s'|)$ in Alg.~\ref{alg:compute_sum_capacity_twosenderMACs_inputsize2}. Therefore, to make Alg.~\ref{alg:compute_sum_capacity_twosenderMACs_inputsize2} the same as Alg.~\ref{alg:maximize_Lipschitz_like_function_1D}, we define $\beta(x) = \beta_I(2x)$. Then, from Prop.~\ref{prop:Lipshitz_like_maximization_analysis}, we know that the number of iterations required to converge to the maximum within an error of $\epsilon > 0$ is bounded above by $\lceil 1/\delta \rceil$, where we used the fact that $\mathcal{D} = [0, 1]$ for the optimization. The number $\delta > 0$ is chosen such that $\beta(x) \leq \epsilon/2$ whenever $0 \leq x \leq \delta$.

    In the following analysis, we measure entropy in bits. To proceed, we note that the binary entropy $h$ satisfies the inequality $h(x) \leq 2\sqrt{x(1 - x)}$~\cite{topsoe2001bounds}. Further, for any probability transition matrix $\mathcal{N}$, the quantity $H_{\mathcal{N}}^{\max}$ defined in Eq.~\eqref{eqn:HNmax_twosenderMAC} is bounded above by $\log(\dout)$, where $\dout$ is the size of the output alphabet. For $x \leq 1/2$, we have $h(x) = \overline{h}(x)$, where $\overline{h}$ is the modified binary entropy function defined in Eq.~\eqref{eqn:modified_binary_entropy}. Then, from Eq.~\eqref{eqn:beta_twosenderMAC} and the definition $\beta(x) = \beta_I(2x)$, it follows that
    \begin{equation*}
        \beta(x) < \mu x + 2 \sqrt{x (1 - x)}
    \end{equation*}
    for $x \leq 1/2$, where $\mu = 3\log(\dout)$. Then, solving the inequality $\mu x + 2 \sqrt{x (1 - x)} \leq \epsilon/2$, we obtain
    \begin{equation*}
        x \leq \frac{(\epsilon \mu + 4) - \sqrt{16 + 8 \epsilon \mu - 4\epsilon^2}}{2 (\mu^2 + 4)}
    \end{equation*}
    which holds whenever $\epsilon \leq \mu$.
    Therefore, we choose
    \begin{equation*}
        \delta = \frac{(\epsilon \mu + 4) - \sqrt{16 + 8 \epsilon \mu - 4\epsilon^2}}{2 (\mu^2 + 4)}
    \end{equation*}
    which is a positive number for $0 < \epsilon \leq \mu$. Then, from Prop.~\ref{prop:Lipshitz_like_maximization_analysis}, we know that the number of iterations to converge to a precision $\epsilon > 0$ is bounded above by
    \begin{equation*}
        K_{\text{PS}} = \left\lceil \frac{2 (\mu^2 + 4)}{(\epsilon \mu + 4) - \sqrt{16 + 8 \epsilon \mu - 4 \epsilon^2}} \right\rceil
    \end{equation*}
    where the subscript PS stands for Piyavskii-Shubert.

    2. Next, we calculate the cost of computing $I^*(s)$ to a precision $\epsilon_I > 0$. This is a non-trivial cost because is obtained by solving a convex optimization problem. Specifically, $I^*(s) = \max_{p \in \Delta_{d_1}} H(A_{q_s} p) - \ip{b_{q_s}, p}$, where $q_s = (s, 1 - s)$. This cost depends on the algorithm one uses to solve the optimization problem. For example, if one uses an interior point method based on the log-barrier, then one needs at most $O(d_1^3 \sqrt{d_1 + \dout} \log((d_1 + \dout) / \epsilon_I)$ flops to converge to the optimum within a precision of $\epsilon_I > 0$ (see Ch.~(11.5) in Ref.~\cite{boyd2004convex}). We denote the cost of computing $I^*$ as $K_I$.

    3. If we use quicksort to sort the numbers $s^{(0)}, \dotsc, s^{(k)}$, then in the worst case, one needs $O(k^2)$ operations to perform this sorting.

    4. If we compute the root to a precision $\epsilon_r > 0$ using bisection, then one needs at most $O(\log(1/\epsilon_r))$ iterations to find this root.

    The total cost of sorting (including all $K_{\text{PS}}$ iterations) is bounded above by $O(K_{\text{PS}}^3)$. The number of times the root needs to be computed is bounded above by $O(K_{\text{PS}}^2)$ and the number of times the function $I^*$ is calculated in the algorithm is $O(K_{\text{PS}}^2 \log(1/\epsilon_r))$. Since we can only compute $I^*$ to a precision $\epsilon_I$ and find the root to a precision $\epsilon_r$, we need to choose these small enough so that the total error is below $\epsilon$. In particular, if we choose $\epsilon_I = \epsilon_r$ such that $O(K_{\text{PS}}^2 \log(1/\epsilon_I)) \epsilon_I = \epsilon/2$, then using the fact that $-x\ln(x) \geq x - x^2$ for $x > 0$, we can infer that it is sufficient to choose $\epsilon_I = O(\epsilon/K_{\text{PS}}^2)$. Then running the while loop in Alg.~\ref{alg:compute_sum_capacity_twosenderMACs_inputsize2} to a precision $\epsilon/2$ so that the total error is within a tolerance of $\epsilon$, we can infer that the total number of iterations needed is bounded above by
    \begin{equation*}
        O(K_I K_{\text{PS}}^2 \log(K_{\text{PS}}^2/\epsilon))
    \end{equation*}
    where we used the fact that $K_I$ is at least as large as $K_{\text{PS}}$. Therefore, the total cost to converge to the optimum is at most polynomial in $d_1$, $\dout$, and $1/\epsilon$.
\end{proof}

Next, we obtain an upper bound on the number of iterations needed to compute the sum capacity of an arbitrary two-sender MAC to a fixed precision. The algorithm we use to analyze this is grid search explained in Sec.~\ref{secn:Lipschitzlike_optimization_simplex_grid_search}, and many of the calculations follow the ideas given in Prop.~\ref{prop:max_iter_sum_capacity_alg_twosenderMAC_inputsize2}.
\begin{proposition}
    \label{prop:max_iter_sum_capacity_alg_twosenderMAC_grid_search}
    Let $\mathcal{N}$ be any two-sender MAC with input alphabets of size $d_1$, $d_2$, and output alphabet of size $\dout$. Suppose that that $d_1, d_2, \dout \geq 2$ with $d_1 \geq d_2$. Suppose that all the entropies are measured in bits and denote $\mu = 3 \log(\dout)$. Then, for a fixed $0 < \epsilon \leq 3$, the number of iterations required by the grid search in Alg.~\ref{alg:compute_sum_capacity_twosenderMACs_inputsize2} to compute the sum capacity of the MAC $\mathcal{N}$ to a precision $\epsilon > 0$ is bounded above by
    \begin{equation}
        \left(\frac{e \epsilon^2}{\mu^2} d_2 + e\right)^{\frac{4 \mu^2}{\epsilon^2} + 1} \label{eqn:max_iter_sum_capacity_twosenderMAC_grid_search}
    \end{equation}
    when $\mu \geq \max\{(8/\epsilon) + 2\sqrt{(16/\epsilon^2) - 1}, \epsilon/2\}$.

    The total cost of computing the sum capacity to a precision $\epsilon > 0$, including the cost of computing $I^*$, is at most
    \begin{equation}
        \textnormal{poly}\left(d_1, \dout, \frac{1}{\epsilon}\right) \left(\frac{e \epsilon^2}{36} \frac{d_2}{(\log(\dout))^2} + e\right)^{\frac{96 (\log(\dout))^2}{\epsilon^2} + 2}
                                                                                                                                                                        \label{eqn:grid_search_complexity_sum_capacity_twosenderMAC}
    \end{equation}
    for $0 < \epsilon \leq 3$ and $\mu \geq \max\{(16/\epsilon) + 2\sqrt{(64/\epsilon^2) - 1}, \epsilon/4\}$.
\end{proposition}
\begin{proof}
    The sum capacity of a two sender MAC can be expressed as $S(\mathcal{N}) = \max_{q \in \Delta_{d_2}} I^*(q)$, where $I^*$ is defined in Eq.~\eqref{eqn:outer_optimization_objective_twosenderMAC}. Note that $I^*$ is $\beta_I$-Lipschitz-like, where $\beta_I$ is defined in Eq.~\eqref{eqn:beta_twosenderMAC}. By Prop.~\ref{prop:simplex_grid_search_Lipschitzlike_optimization}, we know that the number of iterations needed for the grid search algorithm to compute the sum capacity is equal to 
    \begin{equation*}
        K_{\text{GS}} = \binom{N + d_2 - 1}{d_2 - 1},
    \end{equation*}
    where $N = \lceil 1/\delta^2 \rceil$ and $\delta = \sup\{\delta' > 0 \mid \beta_I(\delta') \leq \epsilon/2\}$. In particular, the number $\delta > 0$ can be chosen such that $\beta_I(x) \leq \epsilon/2$ whenever $0 \leq x \leq \delta$ to obtain an upper bound on the number of iterations.

    As noted in the proof of Prop.~\ref{prop:max_iter_sum_capacity_alg_twosenderMAC_inputsize2}, the binary entropy $h$ (measured in bits) satisfies the inequality $h(x) \leq 2\sqrt{x(1 - x)}$~\cite{topsoe2001bounds}. Further, the quantity $H_{\mathcal{N}}^{\max}$ defined in Eq.~\eqref{eqn:HNmax_twosenderMAC} is bounded above by $\log(\dout)$. For $x \leq 1/2$, we have $h(x) = \overline{h}(x)$, where $\overline{h}$ is the modified binary entropy function defined in Eq.~\eqref{eqn:modified_binary_entropy}. Then, from Eq.~\eqref{eqn:beta_twosenderMAC}, it follows that
    \begin{equation*}
        \beta_I(x) < \frac{\mu}{2} x + \sqrt{x (2 - x)}
    \end{equation*}
    for $x \leq 1/2$, where $\mu = 3\log(\dout)$. Then, solving the inequality $(\mu/2) x + \sqrt{x (2 - x)} \leq \epsilon/2$, we obtain
    \begin{equation*}
        x \leq \frac{(\epsilon \mu + 4) - 2 \sqrt{4 + 2 \epsilon \mu - \epsilon^2}}{\mu^2 + 4}
    \end{equation*}
    which holds whenever $\epsilon \leq \mu$. Therefore, we choose
    \begin{equation*}
        \delta = \frac{(\epsilon \mu + 4) - 2 \sqrt{4 + 2 \epsilon \mu - \epsilon^2}}{\mu^2 + 4}
    \end{equation*}
    which is a positive number for $0 < \epsilon \leq \mu$.
    For a fixed $0 < \epsilon \leq \min\{16, \mu\}$ and $\mu \geq \max\{(8/\epsilon) + 2\sqrt{(16/\epsilon^2) - 1}, \epsilon/2\}$, it can be verified that
    \begin{equation*}
        \frac{\mu}{\epsilon} \leq \frac{1}{\delta} \leq \frac{2 \mu}{\epsilon}.
    \end{equation*}
    Further, since $\binom{n}{k} \leq (en/k)^k$, we can write $\binom{N + d_2 - 1}{d_2 - 1} = \binom{N + d_2 - 1}{N} \leq (e (N + d_2 - 1)/N)^N$. Then, noting that $1/\delta^2 \leq N \leq 1/\delta^2 + 1$, the number of iterations needed for grid search to converge to a precision $0 < \epsilon \leq \min\{\mu, 16\}$ is bounded above by
    \begin{equation*}
        K_{\text{GS}} \leq \left(\frac{e \epsilon^2}{\mu^2} d_2 + e\right)^{\frac{4 \mu^2}{\epsilon^2} + 1}
    \end{equation*}
    when $\mu \geq \max\{(8/\epsilon) + 2\sqrt{(16/\epsilon^2) - 1}, \epsilon/2\}$. In order to avoid specifying simultaneous conditions on both $\epsilon$ and $\mu$, we note that $\mu = 3 \log(\dout) \geq 3$ for $\dout \geq 2$. Therefore, we can simply assume $0 < \epsilon \leq 3$ in the above equations.

    Next, as noted in the proof of Prop.~\ref{prop:max_iter_sum_capacity_alg_twosenderMAC_inputsize2}, the cost of computing $I^*$ to a precision of $\epsilon_I > 0$ is at most $O(d_1^3 \sqrt{d_1 + \dout} \log((d_1 + \dout) / \epsilon_I)$. Choosing $\epsilon_I = \epsilon/2K_{\text{GS}}$, we can ensure that after $K_{\text{GS}}$ calls to $I^*$, the error is at most $\epsilon/2$. Then, solving the grid search to a precision of $\epsilon/2$, we can infer that the total cost of computing the sum capacity is at most
    \begin{align*}
        O\bigg(d_1^3 &\sqrt{d_1 + \dout} \log\left(\frac{2 (d_1 + \dout) K_{\text{GS}}}{\epsilon}\right)\bigg) \left(\frac{e \epsilon^2}{4 \mu^2} d_2 + e\right)^{\frac{16 \mu^2}{\epsilon^2} + 1} \\
                &\leq O\left(d_1^3 \sqrt{d_1 + \dout} \log\left(\frac{2 (d_1 + \dout)}{\epsilon}\right)\right) \left(\frac{e \epsilon^2}{4 \mu^2} d_2 + e\right)^{\frac{32 \mu^2}{\epsilon^2} + 2}
    \end{align*}
    when $0 < \epsilon \leq 6$ and $\mu \geq \max\{(16/\epsilon) + 2\sqrt{(64/\epsilon^2) - 1}, \epsilon/4\}$. To obtain the last inequality, we used the fact that $\log(K_{\text{GS}}) \leq K_{\text{GS}}$ for $K_{\text{GS}} > 0$.
\end{proof}

\section{Noise-free subspace MAC\label{app:noise_free_subspace_MAC}}
We compute the sum capacity and the relaxed sum capacity for the examples constructed using the noise-free subspace MAC defined in Eq.~\eqref{eqn:NF_MAC}. For both the examples, we consider the input alphabets $\mathcal{A} = \{a_1, a_2\}$, $\mathcal{B} = \{b_1, b_2\}$ and the output alphabet $\mathcal{Z} = \{z_1, z_2\}$.

\subsection{Example 1}
For the first example, the MAC $\mathcal{N}_F^{(0)}$ has the probability transition matrix
\begin{equation*}
    \mathcal{N}_F^{(0)} = \begin{pmatrix}
                              1 & 0.5 & 0.5 & 0.5 \\
                              0 & 0.5 & 0.5 & 0.5
                          \end{pmatrix}.
\end{equation*}
Let us denote the input probability distribution of sender $A$ as $(p, 1 - p)$ corresponding to the symbols $(a_1, a_2)$. Similarly, denote the input probability distribution of sender $B$ as $(q, 1 - q)$ corresponding to the symbols $(b_1, b_2)$. When the input to the MAC is a product distribution, the output probability distribution determined by the channel is
\begin{equation*}
    p^Z(z_1) = \frac{1 + pq}{2} \text{ and } p^Z(z_2) = \frac{1 - pq}{2}.
\end{equation*}
The mutual information $I(A, B; Z)$ between the senders and the receiver is given by
\begin{equation*}
    I(A, B; Z) = -\frac{1 + pq}{2} \ln\left(\frac{1 + pq}{2}\right) - \frac{1 - pq}{2} \ln\left(\frac{1 - pq}{2}\right) - (1 - pq) \log(2).
\end{equation*}
We wish to compute the sum capacity
\begin{equation*}
    S(\mathcal{N}_F^{(0)}) = \max_{0 \leq p, q \leq 1} I(A, B; Z).
\end{equation*}

For $(p, q) \in \{(0, 0), (0, 1), (1, 0), (1, 1)\}$, we have $I(A, B; Z) = 0$, and hence we can focus on the interior of the domain. For maximization over $p$ for a fixed $q$, by complementary slackness (Lagrangian not written here), we can simply set the derivative of $I(A, B; Z)$ with respect to $p$ equal to $0$. The derivative is given as
\begin{equation*}
    \frac{\partial I}{\partial p} = \frac{q}{2} \ln\left(4 \frac{(1 - pq)}{(1 + pq)}\right).
\end{equation*}
Setting this equal to zero gives $pq = 3/5$.
In this case, no maximization over $q$ is necessary, and subsequently, we find that $S(\mathcal{N}_F^{(0)}) = h(4/5) - (2/5) \ln(2)$, where $h$ is the binary entropy function.

For computing $C(\mathcal{N}_F^{(0)})$, we maximize over all probability distributions over the inputs. For an arbitrary input probability distribution $p(a, b)$, we can write the mutual information between the inputs and the output as
\begin{equation*}
    I(A, B; Z) = -\frac{1 + p(a_1, b_1)}{2} \ln\left(\frac{1 + p(a_1, b_1)}{2}\right) - \frac{1 - p(a_1, b_1)}{2} \ln\left(\frac{1 - p(a_1, b_1)}{2}\right) - (1 - p(a_1, b_1)) \log(2).
\end{equation*}
In this case, we again find that the maximum is attained at $p(a_1, b_1) = 3/5$, and subsequently, the maximum value of the mutual information is equal to $h(4/5) - (2/5) \ln(2)$ as before. Thus, the capacity corresponding to the relaxed sum region matches the actual sum capacity.

\subsection{Example 2}
For the second example, the probability transition matrix is given as
\begin{equation*}
    \mathcal{N}_F^{(1)} = \begin{pmatrix}
                              1 & 0.5 & 0.5 & 0 \\
                              0 & 0.5 & 0.5 & 1
                          \end{pmatrix}.
\end{equation*}
As before, denote the input probability distribution of sender $A$ as $(p, 1 - p)$ and that of sender $B$ as $(q, 1 - q)$. Then, given an input probability distribution to the MAC $\mathcal{N}_F^{(1)}$, the output probability distribution is given by
\begin{equation*}
    p(z_1) = \frac{p + q}{2} \text{ and } p(z_2) = 1 - \frac{(p + q)}{2}.
\end{equation*}
Then, the mutual information $I(A, B; Z)$ between the senders and the receiver can be written as
\begin{equation*}
    I(A, B; Z) = -\frac{p + q}{2} \ln\left(\frac{p + q}{2}\right) - \left(1 - \frac{(p + q)}{2}\right) \ln\left(1 - \frac{(p + q)}{2}\right) - (p + q - 2pq) \ln(2).
\end{equation*}
We wish to compute the sum capacity $S(\mathcal{N}_F^{(1)}) = \max_{0 \leq p, q \leq 1} I(A, B; Z)$. The edge cases ($p, q = 0, 1$) will be handled separately. First, we perform the maximization over $p$ for each fixed $q$,
\begin{equation*}
    I^*(q) = \max_{0 \leq p \leq 1} I(A, B; Z).
\end{equation*}
We look for maxima in the interior $(0, 1)$, which can be obtained through $\frac{\partial}{\partial p}I(A, B; Z) = 0$. Note that the output probability is $(1, 0)$ or $(0, 1)$ only when $(p, q) \in \{(0, 0), (1, 1)\}$, and therefore, the derivative of $I(A, B; Z)$ is well-defined in the interior. This derivative is given as
\begin{equation*}
    \frac{\partial I}{\partial p} = \frac{1}{2} \ln\left(\frac{2 - (p + q)}{p + q}\right) - (1 - 2q) \ln(2).
\end{equation*}
Setting the derivative to zero, we find that $p + q = 2/(\kappa_q + 1)$, where $\kappa_q = 2^{2 - 4q}$. Therefore, the function $I^*(q)$ can be written as
\begin{equation*}
    I^*(q) = h\left(\frac{1}{\kappa_q + 1}\right) - \left(\frac{2}{\kappa_q + 1} - 2\left(\frac{2}{\kappa_q + 1} - q\right) q\right) \ln(2),
\end{equation*}
where $h$ is the binary entropy. Now, to compute the sum capacity, we wish to perform the maximization $S(\mathcal{N}_F^{(1)}) = \max_{0 \leq q \leq 1} I^*(q)$. Since $I^*(q)$ is a continuous function of $q$, the maximum is either attained at the interior or at the boundaries. The maxima in the interior can be found via $\frac{\partial}{\partial q}I^*(q)=0$:
\begin{equation*}
    \ln\left(\frac{1}{\kappa_q}\right) = \ln(2) \left[(4q - 2) - \frac{1}{\ln(2)} \left(\frac{\kappa_q +1}{\kappa_q}\right) (q(\kappa_q + 1) - 1)\right]
    \implies q = \frac{1}{\kappa_q + 1}.
\end{equation*}
It can be verified that $q = 1/2$ is a solution to the above equation, corresponding to which we have $p = 1/2$. Furthermore, this solution is unique. Thus, in the interior of the domain, the maximum occurs at $p = 1/2$, $q = 1/2$. Correspondingly, the maximum value of $I(A, B; Z)$ in the interior is equal to $0.5 \ln(2)$.

Now, we check the edge cases. For $p = 0$, we need to maximize $g(q) = I(p = 0, q) = h(q/2) - q\ln(2)$ over $q$. At $q = 0, 1$, we have $g(0) = g(1) = 0$. Then, the maximum in the interior can be obtained by setting the derivative with respect to $q$ to zero. This gives $q = 2/5$, and since $g(2/5) \approx 0.223 < 0.5\log(2)$, this is not the global maximum. On the other hand, for $p = 1$, we need to maximize $g(q) = I(p = 1, q) = h((1 + q)/2) - (1 - q) \ln(2)$. At $q = 0, 1$, we have $g(0) = g(1) = 0$. The maximum in the interior can be obtained by setting the derivative of $g(q)$ with respect to $q$ equal to zero. This gives $q = 3/5$, and correspondingly, we have $g(3/5) \approx 0.223 < 0.5\ln(2)$. Since $I(p, q)$ is symmetric under interchange of $p$ and $q$, the same results will be obtained when beginning with $q = 0$ or $q = 1$ and then maximizing over $p$. Therefore, we find that the sum capacity is equal to $S(\mathcal{N}_F^{(1)}) = 0.5 \ln(2)$.

Now, we compute the relaxed sum capacity. For this, note that output probability distribution of the MAC $\mathcal{N}_F^{(1)}$ when using an arbitrary input probability distribution $p(a, b)$ is given as
\begin{equation*}
    p(z_1) = \frac{1}{2} + \frac{p(a_1, b_1) - p(a_2, b_2)}{2} \text{ and } p(z_2) = \frac{1}{2} - \frac{(p(a_1, b_1) - p(a_2, b_2))}{2}.
\end{equation*}
Correspondingly, the mutual information between the senders and the receivers is given by
\begin{equation*}
    I(A, B; Z) = h\left(\frac{1}{2} + \frac{p(a_1, b_1) - p(a_2, b_2)}{2}\right) - \left(1 - p(a_1, b_1) - p(a_2, b_2)\right) \ln(2).
\end{equation*}
It can be verified that for $p(a_1, b_1) = p(a_2, b_2) = 1/2$ we have $I(A, B; Z) = \ln(2)$, which is the maximum possible value for the mutual information. Therefore, $C(\mathcal{N}_F^{(1)}) = \ln(2)$.

\end{document}